\documentclass[tikz,12pt]{article}
\usepackage{mheckII}

\usepackage[T1]{fontenc}

\usepackage{tikz-feynman}

\usepackage{manfnt}
\usepackage{longtable}

\usepackage{fancybox}

\usepackage{multirow}

\usepackage{blkarray}

\usepackage{tikz} 
\usetikzlibrary{matrix}

\theoremstyle{theorem}
\newtheorem*{thm}{Theorem}

\newtheorem{fact}{Fact}

\newtheorem*{lem*}{Lemma}

\newtheorem*{php}{Physical expectation}
\newtheorem{que}{Question}

\newtheorem{ans}{Answer}

\newtheorem*{claim}{Claim}

\theoremstyle{definition}
\newtheorem{defn}{Definition}

\newtheorem{rem}{Remark}

\newtheorem{exe}{Example}[section]

\def\Dsl{\,\raise.15ex\hbox{/}\mkern-13.5mu D}
\def\dsl{\,\raise.25ex\hbox{/}\mkern-10.5mu \partial}

\renewcommand{\arraystretch}{1.3}

\title{Symplectic Singularities, Color Confinement,\\ and the Quantum Dirac Sheaf
}

\authors{Sergio Cecotti\footnote{e-mail: {\tt cecotti@bimsa.cn}}\vskip 9pt

\centerline{Beijing Institute of Mathematical Sciences and Applications (BIMSA)}
\centerline{Huaibei Town, Huairou District, Beijing 101408, China}
\centerline{and}
\centerline{Qiuzhen College, Tsinghua University, Beijing, China,}
\centerline{SISSA,
via Bonomea 265, Trieste, Italy}
}

\abstract{A singularity $\C^{2r}/G$, with $G$ a split symplectic reflection group, may or may not be crepant. Then the total space $\mathscr{X}$ of the Donagi-Witten integrable system
is crepant for some 4d $\cn=2$ SCFT and non-crepant for others.
Which physical mechanism controls the (dis)crepancy? Surprisingly, it is the
 detailed physics of color confinement (and its generalizations for non-Lagrangian QFT).
A 4d $\cn=2$ SCFT carries a Frobenius algebra $\mathcal{R}$, 
the quantum cohomology ring of $\mathscr{X}$ (defined via mirror symmetry), and 
$\mathscr{X}$ is crepant iff its \emph{central Witten index} $\dim\mathcal{R}$ is equal to its Euler number $\chi(\mathscr{X})$. When the SCFT has a Lagrangian, $\mathcal{R}$ is fixed by compatibility with confinement, and physics may require a discrepancy to be present.
The quantum cohomology depends on quantum-geometric data, and a classical Seiberg-Witten geometry may have several inequivalent $\mathcal{R}$:  
a relevant quantum datum is the \emph{Dirac sheaf} $\mathscr{L}$ which refines Dirac charge quantization.

We get several other results of independent interest, and 
we fully classify all special geometries of $\bigstar$-type in rank $r>6$.}

\begin{document}

\maketitle

\tableofcontents

\newpage

\section{Introduction, Summary, and Overview}

The natural higher dimensional generalization of the 
McKay correspondence \cite{mckay1,mckay2}
are the  quotient \emph{symplectic singularities}  ($\equiv$ hyperK\"ahler orbifolds)
\be\label{quoquo}
\C^{2r}/G,\qquad G\subset Sp(r)\ \text{finite.}
\ee
When $r=1$ the McKay singularity $\C^2/G$ is \emph{crepant,}
i.e.\! it has a smooth resolution which is a holomorphic symplectic manifold.
In higher dimension this is not always the case. First there is a necessary condition:
 to have a crepant resolution $G$ must be
a \emph{symplectic reflection group}  \cite{ver} (see \S.\,\ref{s:refl}).
An important special case are the \emph{split} symplectic reflection groups: unitary reflection groups $G$  \cite{ST1,ST2} which act diagonally on $\C^r\oplus\C^r$
through their two conjugate reflection representations. 
The simplest instance is $G=W(\mathfrak{g})$,
the Weyl group of a rank-$r$ simple Lie algebra $\mathfrak{g}$. 
However the above condition is \emph{not} sufficient: for ``most'' unitary reflection groups $G$ the hyperK\"ahler
orbifold \eqref{quoquo} is \emph{non}-crepant \cite{SR0,SR1,SR2,SR3,SR4,SR5,SR6,SR7,SR8}.
This implies, \emph{inter alia,} that the singular hyperK\"ahler manifold which is the total space $\mathscr{X}$ of the Donagi-Witten integrable system\footnote{\ Henceforth $\mathscr{X}$ is called a ``special geometry''. The name reflects the fact that, while such a geometry induces a ``special K\"ahler geometry'' on its base $\mathscr{C}$, we never mention or use the underlying K\"ahler structure.} for a 4d $\cn=2$ SCFT \cite{donagi1,donagi2} has no crepant resolution in general, even if we limit (as we do in this note) to the simple situations where the special K\"ahler metric is locally flat.
Thus for most $\cn=2$ SCFTs the higher codimension singularities are \emph{intrinsic} and $\mathscr{X}$ cannot be resolved to a smooth symplectic manifold. This leads to some natural questions:

\smallskip

\begin{que}
What is the \emph{physical} meaning of these \emph{intrinsic} singularities?
\end{que}
\begin{que}
Which \emph{non-perturbative} physical mechanism
makes the special geometry to be crepant\,\footnote{\ \emph{A priori} the crepant resolution (when it exists) may be non-unique. This is typically the case for the symplectic singularities $\C^{2r}/G$ \cite{number}. However in \S.\,\ref{s:compact} we shall present some evidence suggesting that for a full special geometry (of the special class we consider in this note) there is \emph{at most} one crepant resolution.} for some $\cn=2$ SCFT and non-crepant for others?
\end{que}

In this note we advocate the following 

\begin{ans} The non-perturbative phenomenon which controls (dis)crepancy of $\mathscr{X}$ is \emph{color confinement} and its
analogues
for SUSY QFTs which have \textbf{no} weakly-coupled Lagrangian formulation.
\end{ans}

Let us explain. In their seminal paper \cite{SW1} Seiberg and Witten started from their geometric  solution
of non-perturbative $\cn=2$ SYM with gauge algebra $\mathfrak{g}=\mathfrak{su}(2)$, added a  mass perturbation flowing it to $\cn=1$ SYM, and concluded that the $\cn=1$ theory must confine color. In reverse, their argument says that
the sheer fact that $\cn=1$ SYM confines for all $\mathfrak{g}$ puts strong compatibility conditions
on the geometry of $\mathscr{X}$ \emph{for all} $\cn=2$ QFTs with a $\cn=1$ SUSY-preserving relevant perturbation which flows them to $\cn=1$ SYM in some vacua. In \S\S.\,3,\,4 of \cite{donagi1}
Donagi and Witten applied this strategy to an instance of our present set-up, namely $\cn=4$ SYM
with $\mathfrak{g}=\mathfrak{su}(n)$.

When spelled out in full detail,
``$\cn=1$ SYM confines color'' is the statement that its chiral ring
$\mathcal{R}$ \cite{cRing} is the Frobenius $\C$-algebra\footnote{\ The symbol $\mathring{\simeq}$ stands for isomorphism of abstract $\C$-algebras spoiled of other structures. 
} 
\be\label{right}
\mathcal{R}\,\mathring{\simeq}\, \C[X]/(X^{h^\vee}-\Lambda^{3h^\vee}), \qquad X\propto \mathrm{tr}(W^\alpha W_\alpha),
\ee 
($h^\vee$ $\equiv$ the dual Coxeter number of $\mathfrak{g}$) endowed with a trace map
$\mathfrak{f}\colon\mathcal{R}\to\C$ of degree $\deg \mathfrak{f}=-r\bmod h^\vee$ ($r$ $\equiv$ rank of $\mathfrak{g}$).
The algebra $\mathcal{R}$ determines the spectrum of BPS domain-walls which separate regions with distinct SUSY vacua \cite{AV,on}. 
 The special geometry $\mathscr{X}$ for $\cn=4$ SCFT
 (seen as a $\cn=2$ QFT) should reproduce the exact chiral ring \eqref{right} 
 in all vacua of the mass-perturbed theory
 where the only light degrees of freedom 
 are the ones of $\cn=1$ SYM with some gauge subalgebra
 $\mathfrak{g}^\prime\subseteq \mathfrak{g}$.
 In this note we focus on the special vacua with 
 $\mathfrak{g}^\prime\equiv\mathfrak{g}$. 
They are the vacua which sit at the origin $0$ in the Coulomb branch $\mathscr{C}$ of the undeformed
$\cn=4$ theory. 

\medskip

Let $\mathscr{X}$ be the special geometry of $\cn=4$ SYM with gauge algebra $\mathfrak{g}$
(set in its ``minimal'' form  to kill the torsion of $H_\bullet(\mathscr{X},\Z)$),
and let $\mathscr{X}_0$ be its fiber over the origin  of the Coulomb branch $\mathscr{C}$.
We define the \emph{quantum cohomology}
$H^\bullet(\mathscr{X}_0)_\text{quantum}$ \emph{via} the 2d $(2,2)$ Landau-Ginzburg model
mirror to $\mathscr{X}_0$ \cite{hori}. We stress that $H^\bullet(\mathscr{X}_0)_\text{quantum}$ is neither the topological cohomology of $\mathscr{X}_0$ nor
its ``usual'' orbifold cohomology: $H^\bullet(\mathscr{X}_0)_\text{quantum}$ is the quantum cohomology of the \emph{normalization} of
the central fiber $\mathscr{X}_0$ endowed with its canonical structure of Deligne-Mumford
stack (cf.\! \S\S.\,\ref{s:two},\,\ref{s:stack}). The purpose of the physical discussion in section 4 (which may look
excessively long and pedantic for such an ``elementary'' result) is to satisfy the reader that this is the precise notion of ``quantum cohomology'' which captures the dynamical aspects we are interested in.
Other ``quantum cohomologies'' of $\mathscr{X}_0$ also give precious (but different) physical informations.

The compatibility condition between geometry and non-perturbative physics requires $H^\bullet(\mathscr{X}_0)_\text{quantum}$ to be isomorphic (as a Frobenius algebra) to the chiral ring \eqref{right} dictated by confinement.
Since the total space $\mathscr{X}$ retracts
to $\mathscr{X}_0$, we write the equality in the more suggestive form  
\be\label{ture}
\mathcal{R}=H^\bullet(\mathscr{X})_\text{quantum}
\ee
which may be valid in more general situations than the ones studied in this note;
e.g.\! it holds in the original Seiberg-Witten set-up \cite{SW1} where the QFT is asymptotically-free. 

The detailed mechanism of color confinement for a given $\mathfrak{g}$ predicts
$H^\bullet(\mathscr{X})_\text{quantum}$ to be as in \eqref{right}. The existence of a crepant resolution for $\mathscr{X}$ yields
a second prediction for $H^\bullet(\mathscr{X})_\text{quantum}$. When the two predictions do not match, $\mathscr{X}$ (and $\C^{2r}/W(\mathfrak{g})$) have \textbf{\emph{no}} crepant resolution. Conversely, the crepant resolution exists whenever there is agreement.

\medskip

In this paper we study the more general class of $\cn=2$ SCFTs/geometries directly related to split symplectic
singularities. This class -- dubbed $\bigstar$-SCFTs -- contains all the $\cn=4$ and $\cn=3$
SCFTs, higher-rank Minahan-Nemeshanski (MN) models \cite{MN1,MN2,higher MN}, and a lot more.\footnote{\ See \S.\,\ref{s:star} for their definition, characterization, structure (classical \& quantum), and complete  classification.} Most $\bigstar$-SCFTs
have no weakly-coupled Lagrangian formulation.
In the broader context of $\bigstar$-geometries, 
\textbf{Answer 1} may be decoded in 

\begin{ans}
A ($\bigstar$-class, torsionless) special geometry $\mathscr{X}$ has a crepant resolution if and only if its quantum cohomology
algebra $H^\bullet(\mathscr{X})_\text{\rm quantum}$ is a deformation of the classical one $H^\bullet(\mathscr{X},\C)$, i.e.\! iff the two algebras, while \textbf{non}-isomorphic, have the same dimension
\be\label{766zz}
\dim\mathcal{R}\equiv \dim H^\bullet(\mathscr{X})_\text{\rm quantum}=\dim H^\bullet(\mathscr{X},\C).
\ee 
The ``quantum excess'' of Euler number
\be
\delta\chi\;\overset{\rm def}{=}\;\chi(\mathscr{X})_\text{\rm quantum}-\chi(\mathscr{X})_\text{\rm topological},
\ee
 if non-zero,
is the obstruction to the existence of a crepant resolution for $\mathscr{X}$. 
\end{ans}
 
Applied to $\cn=4$ SYM with a gauge algebra $\mathfrak{g}$, eq.\eqref{766zz} says that
$\mathscr{X}$, hence $\C^{2r}/W(\mathfrak{g})$,
has a crepant resolution if and only if the rank $r$ and the dual Coxeter number $h^\vee$ of $\mathfrak{g}$ are related by the condition 
\be\label{cond1}
\delta\chi\equiv h^\vee-(r+1)=0.
\ee
Needless to say, this simple criterion exactly reproduces the math theorems about
the existence of crepant resolutions for $\C^{2r}/W(\mathfrak{g})$ (see \S.\,\ref{s:refl}).
In this paper we shall reproduce more general deep math theorems
using the elementary physical insight \eqref{766zz}. 

\medskip

When the unitary reflection group $G$ is not a Weyl group, the associated special geometry $\mathscr{X}$ describes a SCFT
with \textbf{no} weakly-coupled Lagrangian formulation (\S.\,\ref{s:star}). In this situation we have no direct field-theoretic
prediction for $\mathcal{R}$. However the \textsc{rhs} of eq.\eqref{ture} is still well-defined
and easy to compute thanks to $tt^*$ methods \cite{tt*,on},  mirror symmetry \cite{hori}, and $4d/2d$ correspondences \cite{4d2d}. \textbf{Answer 2} is still true in this more general framework. 

\medskip

The (dis)crepancy is controlled by the \emph{quantum} cohomology of the total space $\mathscr{X}$,
\emph{not} by its classical special geometry. The quantum cohomology is not uniquely determined by
the underlying complex orbifold $\mathscr{X}$ but also by \emph{quantum-geometric}
structures over it. To a single classical special orbifold $\mathscr{X}$ there may correspond
several inequivalent \emph{quantum} geometries with different values of the ``central
Witten index'' $\dim\mathcal{R}$. In this paper we shall play with the idea of quantum-geometric structures on $\mathscr{X}$,
and describe the most natural one on which $\dim\mathcal{R}$
depends: \emph{the quantum Dirac sheaf}.
Varying this quantum datum we may be lead to distinct physical interpretations
of a single classical integrable system $\mathscr{X}$. For instance in rank-1 the same Kodaira fiber may be
endowed with different quantum-geometric structures, and may have diverse
physical meanings.

\medskip
In this paper we study a large class of split symplectic reflection groups $G$,
and the corresponding special $\bigstar$-geometries.
However, to give a flavor of our arguments, let us say some more introductory words about the ``classical''
case of Weyl groups and then quickly comment on the more general case. 

\paragraph{Weyl groups.}
When $G$ is a Weyl group $W(\mathfrak{g})$, the symplectic singularity $\C^{2r}/W(\mathfrak{g})$ is
a local model of the singularities of $\mathscr{X}$ for $\cn=4$ SYM 
with gauge algebra $\mathfrak{g}$ . The orbifold $\C^{2r}/W(\mathfrak{g})$ has a crepant resolution iff  there is a simply-connected Lie group $\cg$, with Weyl group $W(\mathfrak{g})$, such that the space of \emph{commuting
pairs} in $\cg$, modulo conjugacy,
\be\label{pairs}
\cp(\cg)=\big\{g_1,g_2\in \cg,\ g_1g_2=g_2g_1\big\}\big/\big((g_1,g_2)\sim (hg_1h^{-1},hg_2h^{-1}),\ h\in \cg\big)
\ee
admits the structure of a \emph{smooth}
complex projective variety.
Roughly speaking, the non-resolvable singularities of the $\cn=4$ geometry
are captured by the singularities of the space $\cp(\cg)$ which are
described by a theorem of Looijenga \cite{loo}: $\cp(\cg)$ is the weighted projective space  
$\mathbb{P}(d_0,d_1,\dots,d_r)$ \cite{delorme,dolgachev} where the weights $d_i$ are the dual Coxeter labels
of the Lie algebra $\mathfrak{g}$. $\cp(\cg)$ is smooth iff $d_i=1$ for all $i$,
i.e.\! iff the rank $r$ and the dual Coxeter number $h^\vee$ of $\cg$ satisfy the relation \eqref{cond1}.

Geometrically the connection between the symplectic singularity \eqref{quoquo}
and the space \eqref{pairs} 
is not surprising.
In terms of non-perturbative physics, however, the condition \eqref{cond1} is quite suggestive.
The Looijenga theorem \cite{loo} played a crucial role in Witten's non-perturbative analysis of $\cn=1$ SYM with gauge group $\cg$ \cite{WittenN1,kac1} (and also in \cite{AV}). We know that $\cn=1$ SYM is realized in a confining phase, with a mass gap, and 
its $\Z_{2h^\vee}$ chiral symmetry is broken down to $\Z_2$. The physics of color confinement then predicts $h^\vee$ isolated SUSY invariant\footnote{\ SUSY is unbroken since the mass gap rules out a goldstino.} gapped vacua, and
hence a Witten index $h^\vee$: its chiral ring should be as in eq.\eqref{right}. A naive computation of the index would instead return the Euler number $r+1$ of
$\cp(\cg)$, a result which is consistent with confinement only when \eqref{cond1} holds, that is,
only for $\cg=SU(r+1)$ or $Sp(r)$.
The apparent contradiction is solved \cite{WittenN1} by noticing that when $\cp(\cg)$ is smooth
eq.\eqref{cond1} holds, and there is agreement between the naive index and confinement, 
whereas when $h^\vee>r+1$ the space $\cp(\cg)$ has singularities,
and there are further contributions to the index from the singular loci in $\cp(\cg)$:
a careful computation of the index then produces full agreement with confinement \cite{WittenN1,kac1,borel}.
In \S.\,\ref{s:cartoon} we shall revisit the computation and propose an alternative interpretation of it (previously formulated in \cite{AV}) in terms of \emph{quantum cohomology} and \emph{quantum Euler numbers} as suggested by the equalities
\be\label{716cc}
\begin{aligned}
\chi(\cp(\cg))_\text{topological}=r+1&\equiv\big[\text{naive (wrong!) $\cn=1$ Witten index}\big]\\
\chi(\cp(\cg)^\text{can})_\text{stacky}=h^\vee&\equiv\big[\text{correct $\cn=1$ Witten index}\big],
\end{aligned}
\ee
where $\cp(\cg)^\text{can}$ is $\cp(\cg)$ endowed  with its canonical Mumford-Deligne (MD) stack structure (\S.\,\ref{s:quantum}).
Here we think of the stack $\cp(\cg)^\text{can}$ as a ``quantum'' version of its coarse moduli space $\cp(\cg)$.
We shall check that this quantum-geometric interpretation is the physically correct one using mirror symmetry (see also \cite{AV}).
We conclude that the singularities of $\cp(\cg)$ are precisely the ones required
-- in the light of quantum ($\equiv$ stacky) cohomology -- to reproduce the correct physics of color confinement: thus
confinement is the physical mechanism which forces the spaces $\cp(\cg)$ to have singularities.
While this non-perturbative mechanism is most easily seen in $\cn=1$ SYM which is a gapped confining theory,
it has also implications for $\cn=4$ SYM which is realized in a Coulomb phase (see \S.\,\ref{s:cartoon}).
One concludes that a (simply-connected) special geometry $\mathscr{X}$
is crepant iff  the ``central Witten index'' $\dim\mathcal{R}$ has its naive value $\chi(\mathscr{X})$.

\paragraph{General unitary reflection groups: $\bigstar$-SCFTs.}
The Looijenga theorem has been generalized in various directions \cite{kac2,jap,rus1,rus2,rains}. The first generalization is to arbitrary
unitary reflection groups $G$ acting on suitable Abelian varieties $A$. This result
 allows us to extend the argument from $\cn=4$ SYM to the larger class of $\cn=2$ $\bigstar$-SCFTs. 

The second generalization \cite{rus2,rains} involves the choice of an ample line bundle.
The general theorem states that for \emph{special}\footnote{\ See \S.\,\ref{s:strong} for the math characterizations of the special pairs  $(A,\mathscr{L})$ for which the statement in the text holds.
A short version is that these are the pairs such that $\mathsf{Proj}\,\ca_G$ is simply-connected.
The equivalent physical characterization in terms of IR dynamics at special points in the Coulomb branch is given in \S.\,\ref{s:root}.} pairs
 $(A,\mathscr{L})$ --- where $A$ is a $G$-invariant Abelian variety
and $\mathscr{L}\to A$ is a $G$-equivariant ample line bundle --- the 
 algebra of $G$-invariant sections
\be\label{ggral}
\ca_G:=\bigoplus\nolimits_{k\geq0}\Gamma(A,\mathscr{L}^k)^G
\ee  
is a graded polynomial algebra. The Chern class $c_1(\mathscr{L})$
yields the polarization of $A$.
 
The generalization of the variety $\cp(\cg)$  in eq.\eqref{pairs} is the projective scheme $\mathsf{Proj}\,\ca_G$
which is again a weighted projective space. This scheme is the coarse moduli space
of a MD stack $[A/G]$ defined by the action 
$G\curvearrowright (A,\mathscr{L})$.
While the \emph{scheme} $\mathsf{Proj}\,\ca_G$ depends on $\mathscr{L}$ only through its
Chern class $c_1(\mathscr{L})$, the \emph{stack} $[A/G]$ is sensitive to the particular $\mathscr{L}$:
hence the stack carries more information than the complex variety.
We may see the stack as a refinement of the (generalized) Looijenga space $\mathsf{Proj}\,\ca_G$.
In view of the mirror symmetry arguments in \S.\,\ref{s:cartoon}, it is natural to see
the underlying weighted projective space as the ``classical'' geometry while
its stacky structure is a ``quantum'' structure over it. As suggested by eq.\eqref{716cc}, the non-perturbative physics is controlled by ``quantum'' geometry, so the subtle stacky structure may be physically important. 

\smallskip

The relevance of the general theorem for $\bigstar$-geometries is obvious:
as \emph{schemes} these integrable systems are the \emph{relative} $\mathsf{\underline{Proj}}_\mathscr{C}$ over the Coulomb branch
$\mathscr{C}$ of the sheaf
of graded $\mathscr{O}_\mathscr{C}$-algebras
\be\label{relal}
\mathscr{A}_G=\bigoplus_{k\geq0}\Gamma(A\times \C^r,\mathscr{L}^k)^G,
\ee 
that is, they are the relative version $\mathsf{\underline{Proj}}_\mathscr{C}\mathscr{A}_G$ of the generalized Looijenga varieties $\mathsf{Proj}\,\ca_G$. 
Again, their classical geometry depends on the bundle $\mathscr{L}$ only through its 
Chern class $c_1(\mathscr{L})$, which yields the polarization of the Lagrangian fibers. 
The ``quantum'' special geometry --- when identified with its natural stack structure ---
depends on the choice of the ample line bundle
$\mathscr{L}$ in its
 N\'eron-Severi (NS) class.

%
%

\paragraph{The ``quantum'' Dirac sheaf.}
The issue with quantum geometry is best rephrased in terms of the \emph{Dirac sheaf}. A 4d field theory with gauge group $U(1)^r$
is encoded in a polarized Abelian $r$-variety $A$: the electromagnetic charges
take value in the  lattice $H_1(A,\Z)$, the Dirac pairing $H_1(A,\Z)\wedge H_1(A,\Z)\to\Z$
yields the polarization of $A$, while 
 its complex structure
  \be
  \tau\in Sp(2r,\Z)\backslash Sp(2r,\Z)/U(r)\quad \text{(for a principal Dirac polarization)}
  \ee
 captures the
 gauge couplings and $\theta$-angles. The Dirac pairing is the Chern class $c_1(\mathscr{L})$
 of an ample invertible sheaf $\mathscr{L}\to A$ (the \emph{Dirac sheaf}). Therefore we may write the gauge theory
 in term of the Dirac sheaf $\mathscr{L}$, producing a more elegant formulation.
 The relation Dirac sheaf vs.\! Dirac pairing is not one-to-one: two (non-isomorphic) sheaves $\mathscr{L}$ and $\mathscr{L}^\prime$ in the same NS
class produce the same Dirac pairing,
 so \emph{a priori} the datum of the Dirac sheaf is \emph{finer} than the datum of
 a Dirac pairing. The Appell-Humbert theorem \cite{cristine} identifies the sheaves $\mathscr{L}$
 in the class of the Dirac pairing $\langle-,-\rangle$ with the semi-characters
 \be
 \chi\colon\Lambda\to U(1),\quad \text{such that}\quad \chi(\gamma_1)\,\chi(\gamma_2)=(-1)^{\langle\gamma_1,\gamma_2\rangle} \chi(\gamma_1+\gamma_2)
 \ee
i.e.\! a Dirac sheaf is a (complex) \emph{quadratic refinement} in the sense of \cite{GMN1,GMN2}.

In our set-up $\mathscr{X}$ \emph{as a complex variety}  is independent of $\mathscr{L}$  in its NS class,
but depends on it as a DM stack.   
 The main question is whether the Dirac sheaf
 $\mathscr{L}$ is merely a convenient addition to the formalism or it has an intrinsic physical meaning
 (i.e.\! it affects the observables).
 In the standard QFT set-up the precise $\mathscr{L}$ is immaterial,
but there is circumstantial evidence that the exact Dirac sheaf may be important in presence of defects or couplings to a subtle topological sector.
In these situations the ``classical'' geometry should be replaced by the appropriate ``quantum''
one. Since a classical special geometry typically admits several Dirac sheaves,
leading to different quantum-topological invariants, this opens the possibility  
that a number of physically inequivalent $\cn=2$ SCFT share
the same classical Seiberg-Witten geometry.\footnote{\ We stress that the ``quantum'' Dirac sheaf is not the only
potential source of such multiplicity: another one is 
the possibility of different partial resolutions of the singularities (when they exist).}

In this note we address these issues in the simple context of SCFT whose integrable system is a $\bigstar$-geometry. We study in detail their higher codimension singularities armed with the generalized
Looijenga theorems. The extension of the analysis to a broader class of $\cn=2$ theories
is very welcomed.

\paragraph{Other Results \& Applications.}
In the course of the analysis we get several results of independent interest.
We discuss in detail the $A$-stratification of $\bigstar$-geometries.
We classify all irreducible $\bigstar$-geometries of ranks $r\geq7$.
This is a first step in the on-going program of
classifying all isotrivial special geometries. In doing this we clarify
 the \emph{physical} characterization of the Abelian varieties $A$
with $G\subset \mathsf{Aut}(A)$ ($G$ a unitary reflection group) which have the property that
the algebra $\ca_G$ in eq.\eqref{ggral} is polynomial (equivalently, with
$\pi_1(A/G)$ trivial).
It turn out that they correspond to $\cn=2$ effective IR field
whose low-energy dynamics has a very simple form.
Their characterization in terms of their algebra of line operators is rather elegant,
see \S.\,\ref{s:root}. The effective $\cn=2$ theories
with these property will be called \emph{root theories}. Often they are not fully-fledged non-perturbative QFTs but rather truncations to a subset of their topological/superselection sectors.\footnote{\ Examples of roots theories which
happen to be the complete story, are the higher-rank Minahan-Nemeshenski
models.}

\paragraph{Organization of the note.} In \S.\,2 we review some basic facts about
symplectic singularities for general background.
In \S.\,3 we review the basics of special geometry and then focus on the
 subclass of $\bigstar$-geometries which we construct in various ways,
study their properties, and classify in ranks $>6$. We introduce their stacky version.
The math theory of $G$-invariant lattices is physically reinterpreted in term of
 IR effective physics. In \S.\,4 we review Witten computation
of the index for $\cn=1$ SYM, give a 3d and 2d interpretation of it, and use 
mirror symmetry to compute the Euler number in ``quantum'' cohomology.
This part has some overlap with \cite{AV}.
Then we compare this computation with the central fiber of a $\bigstar$-geometry.
In \S.\,5 we discuss the general version of the Looijenga theorem, due to Rains,
and present its implications for special geometry. 
Side material is confined in two appendices.

\section{Symplectic reflections and singularities}\label{s:refl}

One of most elegant results in invariant theory is the Chevalley-Shephard-Todd (CST) theorem
\cite{ST1,ST2}:
 the quotient
$\C^n/G$, where $G\subset U(n)$ is a finite group, is
smooth if and only if $G$ is a \emph{complex} reflection group, i.e.\!
a group generated by reflections.\footnote{\ $g\in GL(n,\C)$ is a (complex) \emph{reflection} iff
$\mathsf{im}(g-1)$ has dimension $1$.} 
The ``next best case'' is when $G$ is a finite \emph{symplectic}  reflection group
(also called  \emph{quaternionic} reflection group \cite{qua}), that is, a finite subgroup
\be
G\subset Sp(2r,\C)\qquad n\equiv 2r
\ee
which is generated by symplectic reflections i.e.\! by elements $g$ such that the range of $(g-1)$
has dimension $2$. For $n=2$ this yields all the finite subgroups of $SU(2)$,
and the quotients $\C^2/G$ are described by the McKay correspondence \cite{mckay1,mckay2}.
$\C^2$ has a $G$-invariant holomorphic symplectic form $\Omega$, and
$\C^2/G$ has a crepant resolution, that is, a smooth resolution over which
$\Omega$ extends as a symplectic form.

Let $G\subset Sp(2r,\C)$ be a finite subgroup. It preserves the symplectic $(2,0)$-form $\Omega$ and its dual 2-vector $\pi$
\be
\Omega= dw^i\wedge dx_i,\qquad \pi=\frac{\partial}{\partial x_i}\wedge \frac{\partial}{\partial w^i},
\ee
where we write the coordinates in $\C^{2r}$ as $(w^1,\dots,w^r,x_1,\dots,x_r)$
which we grade as
\be
\deg w^i=0,\qquad \deg x_i=1.
\ee 
Hence the algebra
\be
\mathscr{A}_G=\big(\C[w^1,\dots,w^r,x_1,\dots,x_r]\big)^G
\ee
carries the structure of a graded Poisson algebra with bracket of degree $-1$.
The quotient
 \be
 \C^{2r}/G=\mathsf{Spec}\,\mathscr{A}_G
 \ee
  is necessarily singular by the CST
theorem, and we have two strategies to handle its singularities: we may look for a \emph{resolution} of the singularity or for a \emph{deformation} of it. For $r=1$ the two strategies are equivalent by Milnor's theorem. 
A \emph{resolution} of the singularity is a proper birational map
$\varpi\colon\cw\to\C^{2r}/G$ such that $\cw$ is smooth and $\varpi$ is an isomorphism away from the singularities of $\C^{2r}/G$. A resolution $\cw$ is \emph{crepant} iff the canonical class of $\cw$ is the pull-back of the canonical class of $\C^{2r}/G$. The resolution is \emph{symplectic} if
$\varpi^*\Omega$ extends to a symplectic form on $\cw$. A deformation of the singularity is
the spectrum of a deformation\footnote{\ $\hbar$ is a continuous deformation parameter i.e.\! 
$\mathscr{A}_G(0)\simeq\mathscr{A}_G$.} $\mathscr{A}_G(\hbar)$ of the algebra $\mathscr{A}_G$: it is a smooth Poisson deformation iff
it is smooth and the deformed algebra $\mathscr{A}_G(\hbar)$ is still Poisson. 
We collect a few math facts:
\begin{itemize}
\item[\rm(1)] If the singularity $\C^{2r}/G$ with $G\subset Sp(2r,\C)$ has
a crepant resolution then $G$ is a symplectic reflection group \cite{ver};
\item[\rm(2)] If $\varpi\colon X\to \C^{2r}/G$ is a resolution, $\varpi^*\Omega$ extends to a holomorphic form on $X$ (Proposition 2.4 of \cite{beauville} or Lemma 3.3 of \cite{K1}). However the extended form may fail to have maximal rank everywhere, so the resolution is \emph{not} symplectic in general;
\item[\rm(3)] If a resolution is crepant, it is automatically symplectic {\rm\cite{K1}};
\item[\rm(4)] Poisson deformations exist (in particular there is a standard one: the
so-called \emph{Calogero-Moser deformation} \cite{K2}).  
\end{itemize}

Then when $r=1$ we have a crepant resolution by Milnor's theorem; this is the usual
McKay correspondence. 

Fact (1) yields an analogy with one side of the CST theorem (the ``only if''). However the inverse implication does not hold
\be
G\ \text{a symplectic reflection group}\ \not\Rightarrow\ \C^{2r}/G\ \text{has a crepant resolution}.
\ee

Most of the math literature \cite{SR0,SR1,SR2,SR3,SR4,SR5,SR6,SR7,SR8} focus on the \emph{split} symplectic reflection groups, that
is, the rank-$r$ unitary reflection groups $G$ which act on $V\oplus V^\vee\simeq\C^{2r}$
where $V$ and $V^\vee$ are the reflection representation and its dual.
These are the only groups relevant for special geometry.
 
For some complex reflection groups the existence of a crepant resolution is rather obvious.
For instance for the Sphephard-Todd imprimitive groups $G(m,1,r)$ (the case $m=2$
is the Weyl group of $C_r$) we have
\be\label{hilbert}
\C^{2r}/G(m,1,r)\equiv \left(\C^2/\Z_m\right)^r/\mathfrak{S}_r\leadsto \Big(X_{A_{m-1}}\Big)^r/\mathfrak{S}_r\leadsto H^{[r]}(X_{A_{m-1}})
\ee
where we see $\Z_m$ as a subgroup of $SU(2)$ acting on $\C^2$, and $X_{A_{m-1}}$
is the crepant resolution of $\C^2/\Z_m$ in the McKay correspondence, which is a smooth
symplectic \emph{surface}, and $H^{[r]}(Y)$ the Hilbert scheme of $r$-points in the surface $Y$.
Since the Hilbert scheme of a smooth symplectic surface is smooth symplectic,
we got a symplectic resolution of the singularity in the \textsc{lhs} of \eqref{hilbert}.
In this case there are other resolutions: their number is given by a generalized 
Catalan number  \cite{number}
\be\label{calalan}
\frac{1}{rm+1}{rm+1\choose m}\equiv {rm \choose m}-(r-1){rm\choose m-1}.
\ee
A resolution of the $\C^{2r}/W(A_r)$
is also given by a Hilbert scheme construction, which is then crepant.
It turns out (see e.g.\! \cite{SR7} Theorem 6.8, and references therein) that, a part for the above obvious cases
\be
\mathfrak{S}_{r+1}\equiv W(A_r),\quad G(m,1,r),
\ee 
the only unitary reflection group which yields a crepant singularity is the exceptional Shephard-Todd group $G_4$. In facts $\C^4/G_4$ has two inequivalent crepant resolutions \cite{number}.
In addition there is a known example of a \emph{non-split} symplectic singularity which is
crepant \cite{bell3}. It is believed that there are no others \cite{bell4}.
Non-split symplectic reflection groups play no role for us. It is also not clear whether the group $G_4$
 may arise from a well-behaved SCFT to which standard physical arguments apply.

 \section{$\bigstar$-geometries}\label{s:star}
\subsection{Preliminaries}\label{s:prelimimaries}

In the present paper ``special geometry'' stands for the total space of the holomorphic\footnote{\ In facts \emph{algebraic}.} integrable system associated to the Seiberg-Witten geometry of a 4d $\cn=2$ QFT \cite{SW1,SW2,donagi1,donagi2} with all its defining structures. We consider only geometries of the superconformal type (i.e.\! $\C^\times$-isoinvariant), in facts the even more restricted class of $\bigstar$-geometries to be defined below.
Henceforth we shall omit the qualification ``$\C^\times$-isoinvariant'' and speak simply of ``special geometry'' to mean ``$\C^\times$-isoinvariant special geometries''. 
Preliminarily we recall some definitions and well-known facts mainly to fix notation and terminology.
Experts may skip this subsection.

\paragraph{Chiral ring and Coulomb branch.}
A (conformal) \emph{chiral ring} is a finitely-generated, graded, connected $\C$-algebra $\mathscr{R}$, with no zero-divisors nor nilpotents, that we may take to be normal without loss.
Its affine spectrum $\mathscr{C}\equiv\mathsf{Spec}\,\mathscr{R}$ is the \emph{Coulomb branch}. 
When $\mathscr{C}$ is smooth
(as it is the case for $\bigstar$-geometries),
 \be
 \mathscr{R}=\C[u_1,\dots,u_r],
 \ee
 and hence $\mathscr{C}\simeq \C^r$. In this note we consider only geometries with smooth
 Coulomb branch. The grades $\{\Delta_1,\dots,\Delta_r\}$
 of the free generators $\{u_1,\dots,u_r\}$ of $\mathscr{R}$ are rational numbers $\geq1$ called the \emph{Coulomb dimensions}.
In an \emph{interacting} $\bigstar$-SCFTs the $\Delta_i$'s are actually integers $\geq2$ with
$\gcd(\Delta_i)\in\{1,2,3,4,6\}$ \cite{char} and the multiplicity of $2$ in $\{\Delta_1,\dots,\Delta_r\}$ is at most $1$. The grading of $\mathscr{R}$
yields a linear action
$\C^\times \curvearrowright \mathscr{C}$ whose characters
 are the Coulomb dimensions $\{\Delta_i\}$. We take this $\C^\times$-action as part of the definition of the Coulomb branch.
$\mathscr{C}$ decomposes into $\C^\times$-orbits. There is only one closed orbit, 
namely the origin $\{0\}\in \C^r$, which belongs to the closure of all orbits.

When the special geometry arises from a 4d $\cn=2$ SCFT, 
the Coulomb branch $\mathscr{C}$ yields a parametrization of
 the supersymmetric quantum vacua in flat four-dimensions
with unbroken $R$-symmetry $SU(2)_R$. For SCFTs with $\cn>2$ SUSY we fix once and for all an $\cn=2$
superconformal subalgebra and ``unbroken $SU(2)_R$'' refers to the non-Abelian
$R$-symmetry of the chosen SUSY subalgebra.   
The unique closed orbit $\{0\}$ represents the 4d SUSY vacuum where the
superconformal symmetry is \emph{unbroken}. 

An alternative physical interpretation of the chiral ring $\mathscr{R}$ identifies it with the algebra of
local observables in the TFT obtained by twisting the $\cn=2$ SCFT \emph{\'a la} Witten \cite{TFT1,TFT2,TFT3}.
In particular $\mathscr{R}$ and all objects constructed out of it, such as the Coulomb branch $\mathscr{C}$,
are \emph{smooth invariants} of the Euclidean space-time $X$, that is, these objects are not affected by any smooth deformation of $X$. 

\medskip

\paragraph{Special geometry: elementary viewpoint.}
We are ready to define the main object in our discussion.
A \emph{special geometry} over the Coulomb branch $\mathscr{C}$,
\be\label{76zzz712q}
\pi\colon\mathscr{X}\to \mathscr{C}
\ee
 is a complex space $\mathscr{X}$ which carries a
holomorphic symplectic $(2,0)$-form $\Omega$, regular in codimension 2,
with a holomorphic projection $\pi$ onto the Coulomb branch $\mathscr{C}$,
such that the fibers of $\pi$ are Lagrangian submanifolds,
while the smooth fibers of $\pi$ are polarized Abelian varieties.
$r\equiv\dim\mathscr{C}\equiv\tfrac{1}{2}\dim\mathscr{X}$ is the \emph{rank} of the special geometry.

 A special geometry is, in particular, a holomorphic (algebraic)
\emph{integrable system}. In the superconformal case $\mathscr{X}$ has additional structures.
The action of $\C^\times$ on the base $\mathscr{C}$ extends to an action on the total space
$\mathscr{X}$
by automorphisms of all its defining geometric structures. In particular the fibers over any two points
of a $\C^\times$-orbit are isomorphic: when smooth they are isomorphic as polarized Abelian varieties;
more generally, the Albanese varieties of their two fibers are isomorphic Abelian varieties. 
The Lie algebra of $\C^\times$ is generated by a holomorphic \emph{Euler vector} $\ce$ on $\mathscr{X}$ which we normalize such that\footnote{\ $\mathscr{L}_v$ denotes the Lie derivative in the direction of the vector field $v$.}
\be
\mathscr{L}_\ce \Omega=\Omega.
\ee
The dual holomorphic differential $\lambda=\iota_\ce\Omega$ plays the role of the Seiberg-Witten differential, $\Omega=d\lambda$. We stress that in the SCFT situation (contrary to the
massive case\footnote{\ Indeed the class $[\Omega]$ is linear in the masses by the Duistermaat-Heckman theorem \cite{donagi2,antiaffine}.}) $\Omega$ is exact. The homogeneous coordinates $u_i$ of the base $\mathscr{C}$ then satisfy
\be
\mathscr{L}_\ce u_i= \Delta_i u_i.
\ee
We assume that the fibration
$\pi$ has a (holomorphic) \emph{zero section} $s_0$ which is preserved by the $\C^\times$ action. 
One may dispense with this assumption, and get a more general class of special geometries \cite{hoinv},
but we shall not do that in this note.
We say that a special geometry is \emph{interacting} iff $1$ is not a Coulomb dimension and
the geometry cannot be written as a product $\mathscr{X}_1\times\mathscr{X}_2$ of lower rank geometries.

\paragraph{Polarizations and Isogenies.}
 
The homology group $H_1(\mathscr{X}_u,\Z)\simeq \Z^{2r}$ of the smooth
fiber $\mathscr{X}_u$ at $u\in\mathscr{C}$ 
is physically identified
with the lattice $\Lambda$ of electromagnetic charges in the SUSY vacuum $u$ \cite{SW1,SW2}.
The Dirac electromagnetic pairing yields a polarization
\be
H_1(\mathscr{X}_u,\Z)\wedge H_1(\mathscr{X}_u,\Z)\to \Z 
\ee
to the smooth fibers. We shall give another interpretation of this polarization in the next paragraph.
In the QFT applications one usually assumes the polarization to be principal (we then say that the geometry is \emph{principal}). However for our purposes
 it is convenient to work with arbitrary polarizations
--- principal and not --- and to compare them. We say that two special geometries
over the same Coulomb branch\footnote{\ Recall that the $\C^\times$-action is part of the
definition of $\mathscr{C}$.} $\mathscr{C}$
\be
\pi_y\colon \mathscr{Y}\to\mathscr{C},\qquad \pi_x\colon \mathscr{X}\to\mathscr{C}
\ee 
are \emph{isogeneous} if there is a surjective map
\be
\alpha\colon\mathscr{Y}\to\mathscr{X},
\ee
which \textit{(i)} commutes with the projections $\pi_x,\pi_y$, \textit{(ii)} maps the zero section to the zero section, \textit{(iii)} the symplectic structure on $\mathscr{Y}$
is the pull-back of the one on $\mathscr{X}$, while \textit{(iv)} the restriction of $\alpha$
 to each smooth fiber is 
an isogeny of Abelian varieties. Isogeneous special geometries may
describe distinct SCFTs, but they may also be used to 
study slightly different versions of the ``same'' SCFT where we forget some of its
superselected sectors.
To a large extend we are free to replace our special geometry with an isogeneous
one. Our motivation for the notion of \emph{isogeneous special geometries} is mainly technical: in
many instances the computations become much simpler after we replace
the original geometry with an isogeneous one. In particular, replacing the geometry with an isogeneous one, we may kill the troublesome
torsion in the homology of $\mathscr{X}$. As we shall see, this killing has a QFT interpretation. 
\medskip

We recall the physical significance of the polarization type.
When the polarization is \emph{principal}, the only geometric restriction on the allowed electromagnetic charges
is Dirac quantization: in this situation the lattice of electromagnetic charges $\Lambda_\text{max}$
is maximal with respect to the Dirac electromagnetic pairing, i.e.\! the charge lattice is \emph{complete}.
We stress that the QFT may or may not
have actual local states (massive multi-particle states) which carry all the charges $q\in\Lambda_\text{max}$; we write $\Lambda_\text{min}$ for the sublattice generated by the charges of actual massive states.
However we may always define \emph{line operators} carrying any
charge in $\Lambda_\text{max}$.
On the contrary, when the degree of the polarization is $d>1$, we have
additional  selection rules on the geometrically allowed charges: they should take value  in a sublattice
$\Lambda_\text{geo}$ of index $d$ in some Dirac maximal lattice $\Lambda_\text{max}$.
To be consistent with local physics $\Lambda_\text{min}\subseteq \Lambda_\text{geo}$.
The line operators are now restricted to have charges in $\Lambda_\text{geo}$. 
One may wish to enlarge the charge lattice $\Lambda_\text{geo}$ to a Dirac maximal one.
To achieve this, one should construct an isogeneous special geometry with principally polarized fibers.
However this cannot be always achieved
\begin{fact} 
There exist special geometries which are not isogeneous to principally polarized ones.
\end{fact}
For examples see \cite{char} or appendix \ref{G!2}. We call such geometries \emph{absolutely incomplete}. In these situations we cannot define a complete set of line operators in a globally consistent way.
Incompleteness leads to the existence of non-trivial 1-form symmetries, see e.g.\! \cite{1form1,1form2,1form3,1form4,1form5,1form6}.

When an isogeneous principal geometry exists, it is typically non-unique.
 For a given local physics there may be several \emph{global} QFTs
which differ in their maximal sets of line operators \cite{readings}: we shall refer to these subtle different QFTs, which
share the same local dynamics, as diverse ``readings'' of the theory. The allowed ``readings''
are in one-to-one
 correspondence with the Dirac maximal lattices containing the lattice $\Lambda_\text{min}$
spanned by the charges of the local states, i.e.\! they correspond to the set of isogeneous principal geometries. 
For instance, for $\cn=4$ SYM
the local lattice is
\be\label{loclattice}
\Lambda_\text{min}\simeq Q\oplus Q^\vee
\ee
 where $Q$ (resp.\! $Q^\vee$) is the root lattice of the gauge Lie algebra $\mathfrak{g}$ (resp.\! dual root lattice\footnote{\ Beware of the distinction between the dual root lattice and the dual lattice of the root lattice!})
 whose index in a Dirac maximal lattice is the order of the center of the simply-connected Lie group $\cg$ with algebra $\mathfrak{g}$,
 and so $\Lambda_\text{min}$ is maximal only for $E_8$, $F_4$ and $G_2$.
In the other cases there is a plurality of ``readings'' $\Lambda_\text{max}\supset \Lambda_\text{min}$.

\paragraph{The 3d viewpoint.}
Consider the corresponding topologically twisted TFT \cite{TFT1,TFT2,TFT3}
on the Euclidean space-time $\R^3\times S^1$. Since the TFT is independent of the space-time metric,
its low-energy effective theory is exact  \cite{4man}, that is, the TFT is equivalent to the twisting of the low-energy effective $\cn=2$ theory of the original model. It is also independent of the radius $R$ of $S^1$, so we may use either the
4d IR effective theory (corresponding to $R\approx\infty$) or the 3d one, corresponding to $R$ asymptotically small but \emph{non-zero}.\footnote{\ \label{SWW3} As $R\to\infty$ $\mathscr{X}$ equipped with its hyperK\"ahler metric
collapses to the special K\"ahler metric $\mathscr{C}$ in the Gromov sense \cite{Tosatti}. As $R\to0$ the fiber
decompactifies. Hence the 3d effective theories for $R$ small and $R=0$ are different, and we need to keep $R$ non-zero. See ref.\cite{SW3} for a nice discussion  in the present context of the crucial difference between ``compactification to 3d'' and ``dimensional reduction to 3d''.}
The 3d effective theory is a 8-supercharge SUSY $\sigma$-model (coupled to a topological sector) with target space $\cz$ which --- away from singularities --- is hyperK\"ahler. The physical
hyperK\"ahler metric of the good locus $\mathring{\cz}\subset\cz$ is a complicate function of $R$ \cite{GMN1}, while TFT quantities
may depend only on protected geometric structures which do not flow with $R$.
We list some of these topologically protected structures.
First: by topological invariance the 4d and 3d effective theories
have the same algebra of local topological observables $\mathscr{R}$. We call $I$ the special complex structure on $\mathring{\cz}$
in which the elements of $\mathscr{R}$ are holomorphic functions.
This peculiar complex structure is then topologically protected.
$\mathscr{R}$ is the ring of global $I$-holomorphic functions\footnote{\ We stress that this is a special property of complex structure $I$. $\mathscr{X}$ is typically Stein in a \emph{generic} complex structure (which is not topologically protected). The holomorphic functions are actually regular algebraic functions on the algebraic variety $\mathscr{X}$.} on both $\mathscr{C}$ and $\cz$,
 and this produces a canonical
(surjective) affinization map
\be\label{789zza}
\pi_\cz\colon\cz\to \mathscr{C}, 
\ee
which is also a topologically protected geometric structure.
Comparison of 4d and 3d shows that the particular symplectic structure
of the hyperK\"ahler $\mathring{\cz}$ which is holomorphic in complex structure $I$
 is also topologically protected. 
3d $R$-symmetry predicts that, away from singularities of $\cz$, the fibers are
holomorphic Lagrangian submanifolds with no non-constant
holomorphic functions. The simplest possibility is that the smooth fibers are complex tori.\footnote{\ For the most general situation, see \cite{antiaffine}.} The third protected structure is the \emph{class} of the K\"ahler form
in complex structure $I$ when measured in the proper units of $1/R$. The restriction of $R[\omega|_{\cz_u}]$ to be fiber $\cz_u$ is integral by Dirac quantization of charge, and yields a polarization to the smooth fibers which thus are Abelian varieties.
$R[\omega|_{\cz_u}]$ is then the Chern class of an ample 
line bundle $\mathscr{L}_u\to\cz_u$. The fiberwise bundles $\mathscr{L}_u$ glue in a global (orbi)bundle $\mathscr{L}$
on $\cz$. 

Then, away from singularities, the affinization TFT map
\eqref{789zza} is the same as the special geometry
\eqref{76zzz712q}. However the ``naive'' special geometries $\mathscr{X}$ people use may differ from the ``ideal'' one
\eqref{789zza} in high codimension. In facts it is not even clear in which --- possibly non-commutative --- geometric category $\cz$ is defined.
What is important for us is that the low-energy physics (hence the exact TFT) is given by
a ``$\sigma$-model'' with target $\cz$ that, neglecting subtler issues,
we may interpret as a $\sigma$-model on $\mathscr{X}$. The space $\mathscr{X}$ inherits
quantum-geometric structures from $\cz$, and we shall try to understand some of them.

\paragraph{Stratifications of $\mathscr{C}$.}
The Coulomb branch $\mathscr{C}$ of a ($\C^\times$-isoinvariant) special geometry
has two natural stratifications \cite{strat1,strat2,inverse} called respectively \emph{$A$-stratification} and \emph{$R$-stratification,} which may be combined in an $AR$-stratification whose strata are the connected components of overlaps of $A$- and $R$-strata.

The $A$-stratification \cite{strat1,strat2,inverse} is dictated by the dimension of the Albanese variety of the fiber\footnote{\ In the following sense: the normalization of the smooth locus of the connected component of the fiber $\mathscr{X}_u$ which contains the image of the zero section is an algebraic Abelian group, and we consider the Albanese variety of this group in the sense of the Barsotti-Chevalley theorem \cite{milneG}.} $\mathscr{X}_u$, equivalently, by the rank
of the lattice of charges of BPS state which are massless at the given SUSY vacuum $u\in\mathscr{C}$ \cite{inverse}.
The $R$-stratification is by the order of the subgroup of $U(1)_R$ which is unbroken at the given vacuum. The $R$-stratification yields constraints on the allowed Coulomb dimensions $\Delta_i$
see \cite{char,inverse,caorsi,allowed}. 

 The codimension-$1$ strata of the $A$-stratification are the connected components of the smooth locus $\mathscr{D}^\text{sm}$ of the discriminant $\mathscr{D}$. By definition, the \emph{discriminant} $\mathscr{D}\subset \mathscr{C}$ is the divisor
of points $u\in\mathscr{C}$ whose fiber $\mathscr{X}_u$ is not smooth. 
The irreducible components of the discriminant carry their own special geometries of rank $r-1$ \cite{strat1,strat2,inverse}.
Then in codimension-$2$ we have the connected components of the smooth locus of the discriminant of the discriminant, and so on,
recursively in the codimension, until in codimension-$r$
we get the stratum consisting of the closed orbit $\{0\}$, which is the only closed stratum in an interacting SCFT. The union of all codimension-$k$ strata is then the locus in $\mathscr{C}$
whose fibers $\mathscr{X}_u$ have Albanese varieties of dimension $r-k$ \cite{inverse}.

The geometry of codimension-$1$ strata is pretty well understood \cite{HO1,HO2,HO3}. Up to some subtler detail \cite{hoinv},
their physics is essentially controlled by the local monodromy around each irreducible component of $\mathscr{D}$ \cite{inverse}. 
Higher codimension strata
are poorly understood in general, and one of the goals of this paper is to
improve our knowledge.

\paragraph{The central fiber.}
The central fiber $\mathscr{X}_0$ over the origin -- which is the highest codimension stratum -- is the most singular fiber,
 in the sense that a maximal number of cycles have collapsed to zero and
  its Albanese variety is trivial for an interacting SCFT. 
The cohomology $H^\bullet(\mathscr{X}_0)$ 
of the central fiber is an important invariant of the special geometry.
Since in a $\C^\times$-isoinvariant geometry $\mathscr{X}$ retracts to $\mathscr{X}_0$,  $H^\bullet(\mathscr{X}_0)$ is also the cohomology of the full geometry $\mathscr{X}$.
However one should be careful with the notion of ``cohomology''. 
There is the \emph{``usual'' orbifold cohomology}  \cite{orbifold1,orbifold2,orbifold3}
which, when
the geometry has a crepant resolution, coincides with the ordinary cohomology of the resolution;
in this case
 \be
 \text{(rank of flavor group of SCFT)}\leq \dim H^2(\mathscr{X}_0)_{\!\!\text{``usual''}\atop\text{orbifold}}-1.
 \ee

In this paper we dwell on a \emph{different} notion of quantum Euler number of the central fiber,
which \emph{does not} coincide with the orbifold one even for crepant $\bigstar$-geometries.
Rather  our $\chi(\mathscr{X}_0)_\text{quan}$ is a kind of ``central Witten index'' which,
as mentioned in the introduction, computes the dimension of the central quantum chiral ring
$\mathcal{R}$.
This quantum invariant depends on the choice of ``quantum structures'' on $\mathscr{X}$.

 \paragraph{Isotrivial geometries.}
 A special geometry is said to be \emph{isotrivial} iff all smooth fibers $\mathscr{X}_u$
 are isomorphic as polarized Abelian varieties to a fixed model fiber $A$
 \be
 \mathscr{X}_u\simeq A\quad\text{for }u\in\mathscr{C}\setminus\mathscr{D}.
 \ee
  A special geometry is isotrivial if and only if
 its local monodromies have  finite order \cite{inverse}.  For \emph{any} isotrivial geometry one shows:

\begin{fact}\label{firstfact} {\bf(1)} The model fiber $A$ of an interacting isotrivial geometry
 is isogeneous to the $r$-th power of a fixed elliptic curve $E_\tau$ (of period $\tau\in\mathbb{H}$)
\be
A\to \overbrace{E_\tau\times E_\tau\times \cdots\times E_\tau}^{r\ \text{\rm factors}}.
\ee
Then we have the isomorphism of complex tori
\be\label{astori}
A \sim_\text{\rm tori} E_{\tau/n_1}\times E_{\tau/n_2}\times \cdots\times E_{\tau/n_r}
\ee 
for certain integers $n_i$'s. However \eqref{astori} is \textbf{not} an isomorphism of polarized Abelian varieties,
i.e.\! the polarization of $A$ is \textbf{not} (in general) the product polarization of the \textsc{rhs} (cf. {\bf Example \ref{noniso}}).  \\
{\bf(2)} For an \emph{irreducible}\footnote{\ A special geometry is irreducible iff it is not trivial (i.e.\! not of the form $A\times \mathscr{C}\to\mathscr{C}$) and is not the product of lower rank geometries.} isotrivial special geometry we have two possibilities:
\begin{itemize}
\item[\rm(A)]  we have a one-parameter family of special geometries parametrized by a \emph{conformal manifold} which is a modular curve\footnote{\ \textbf{Notation}:  $\mathbb{H}$ is the upper half-plane
and $\Gamma\subset SL(2,\Z)$ is a congruence subgroup.} $\mathbb{H}/\Gamma$, while $2$ is a Coulomb dimension of multiplicity exactly $1$. In this case the special geometry
describes a SCFT with a Lagrangian formulation;
\item[\rm(B)] $E_\tau$ has complex multiplication by a Gaussian imaginary quadratic field,
i.e.\! a field $\mathbb{Q}(\sqrt{-d})$ with class number $1$. 
In this case the conformal manifold is trivial (i.e.\! the geometry is \emph{rigid}), $2$ cannot be 
a Coulomb dimension, and the SCFT has no Lagrangian formulation (the SCFT is inherently strongly coupled).
\end{itemize}
\end{fact}

\begin{rem}
The special isotrivial special geometries where eq.\eqref{astori} is actually an isomorphism of
polarized varieties were dubbed ``diagonal'' in \cite{char}. This happens e.g.\! 
for the higher rank Minahan-Nemeshanski
theories.
\end{rem}

\begin{rem}
Suppose that a rank $>1$ isotrivial special geometry has a SW curve:
then $A$ is the Jacobian of a curve $\Sigma$ which is isomorphic, as a complex torus,
to a product of elliptic curves. Then $\Sigma$ is a solution to the 
Ekedahl-Serre problem \cite{ESerre}; the Ekedahl-Serre curves are quite rare animals.
So isotrivial geometries typically don't have SW curves. 
\end{rem}

\begin{exe}\label{noniso} A geometry where \eqref{astori} is not a polarized isomorphism is the 
principal rank-2 isotrivial geometry based on the Shephard-Todd exceptional group $G_{12}$
\cite{ST1,ST2}. In this case  (see appendix \ref{G!3} or \cite{fuji}) 
\be
A=E_{\sqrt{-2}}\times E_{\sqrt{-2}}\quad\text{as a torus,}
\ee
but the geometry is not diagonal since
 the period matrix (in the unique $G_{12}$-invariant principal polarization) is
\be
\tau=\frac{1}{2}\begin{pmatrix}1+2\sqrt{-2} & \sqrt{-2}\\
\sqrt{-2} & 1+2\sqrt{-2}\end{pmatrix}.
\ee
\end{exe}

The general theory of isotrivial special geometries will be described elsewhere.
Here we limit ourselves to the simpler subclass of $\bigstar$-geometries.

\paragraph{$\bigstar$-geometries.}

An isotrivial special geometry is a $\bigstar$-geometry iff it is a global quotient of a \emph{trivial}
geometry $A\times \C^r$.  The dimensions $\Delta_i$ of an interacting $\bigstar$-geometry
are integers $\geq2$ ($\geq3$ for rigid $\bigstar$-geometries), see \S.\,\ref{s:classstr}.
The ``$\bigstar$-geometries'' owe their name to their monodromy characterization:
\begin{fact} A special geometry is a $\bigstar$-geometry if and only if all its local monodromies
around the irreducible components of the discriminant have a semisimple \emph{starred}  Kodaria type:
$I_0^*$, $II^*$, $III^*$ or $IV^*$.
\end{fact}

Inside the class of $\bigstar$-geometries we have the subclass of \emph{root} $\bigstar$-geometries
which have especially nice properties (see \S.\,\ref{s:root} for their definition and properties). 
As we shall see, all $\bigstar$-geometries are isogeneous to \emph{at least one}
root $\bigstar$-geometry so, when working modulo isogeny, we may limit ourselves to the root ones.
The higher dimensional Minahan-Nemeshanski models and the $\cn=4$ SYM with a gauge group $\cg$
which is simply-connected with trivial center (i.e.\! $\cg=E_8, F_4, G_2$) are described automatically by 
a root $\bigstar$-geometry.

\medskip

We saw above that the central fiber is the one with the largest number of collapsed cycles,
that is, its Albanese variety $\mathsf{Al}(\mathscr{X}_0)=0$ ($\mathscr{X}$ interacting).
For $\bigstar$-special geometries we may be more precise. Later we shall check the following:
 
 \begin{fact} In a $\bigstar$-geometry the normalization $\mathscr{X}_0$
 of the \emph{unresolved}\footnote{\ That is, the central fiber
 as a normal analytic space before resolving the quotient singularities above the origin $\{0\}$ (which most of the time cannot be done while preserving the symplectic structure).} central fiber is
 an irreducible analytic space with finite fundamental group.
$\mathscr{X}_0$ is simply-connected if and only if the $\bigstar$-geometry is
a \emph{root} geometry.
 \end{fact}

  \subsection{$\bigstar$-geometries: ``classical'' structure}\label{s:classstr}

\paragraph{$G$-invariant lattices and Riemann forms.}
Let $G$ be an irreducible (finite) unitary reflection group, $V\simeq \C^r$ its
reflection representation, and $V^\vee$ the dual one \cite{ST1,ST2}. We
require that $G$ leaves invariant a full lattice $\Lambda\subset V$ ($\Lambda\simeq \Z^{2r}$). 
Then $\Lambda\otimes_\Z\!\mathbb{Q}$ is a rational $2r$-dimensional representation of $G$
isomorphic over $\C$ to $V\oplus V^\vee$.
It follows that the traces of $G$ in $V$ are valued in $\mathbb{Q}$
or in a quadratic imaginary extension. The trace field $\mathbb{K}_G$ of a finite unitary group $G$  is also its field of definition \cite{ST2,hand}.
Therefore a $G$-invariant full lattice $\Lambda$ exists iff $\mathbb{K}_G$ is $\mathbb{Q}$ or an imaginary quadratic field.\footnote{\ Cf.\! \textbf{Theorem 3.2.} of \cite{dolgachev2}.}
One has the stronger result that $G\subset GL(r,\Z[\lambda_G])$ where  $\Z[\lambda_G]$ is the ring of integers in $\mathbb{K}_G$ \cite{ST2}.

The irreducible reflection groups with $\mathbb{K}_G=\mathbb{Q}$ are the Weyl groups of the simple Lie algebras $\mathfrak{g}$. The irreducible unitary reflection groups whose field of definition is imaginary quadratic are
listed in table \ref{tab}
(in the Shephard-Todd notation \cite{ST1,ST2}). The class number of $\mathbb{K}_G$ is $1$,
so all torsionless $\Z[\lambda_G]$-modules are free. Therefore when $G$ is not a Weyl group
\be
\Lambda\simeq \Z[\lambda_G]^r \qquad \lambda_G= i,\ \ e^{2\pi i/3},\ \ \sqrt{-2},\ \ \tfrac{1+\sqrt{-7}}{2}.
\ee
is an example of $G$-invariant lattice. But there are others. For rank-2 (and $G\neq G_{12}$)
the $G$-invariant lattices are listed in \cite{jap}; for $G_{12}$ see \cite{bolza,bolza2}.
The $G$-invariant lattices for arbitrary rank $r$ are classified in ref.\!\cite{popov};
their classification will be explained in \S.\,\ref{s:root} using physical
insights.
When $G\equiv W(\mathfrak{g})$ is the Weyl group of the Lie algebra $\mathfrak{g}$,
the invariant lattices form one-parameter families of the form 
 \be
\Lambda(\tau)=L_1+\tau L_2 \subset \C^{r},
\ee
where $L_1$ and $L_2$ are two congruent $W(\mathfrak{g})$-invariant lattices in $\R^r$ and $\tau$ is a point in the upper half-plane $\mathbb{H}$.
The Looijenga $W(\mathfrak{g})$-lattices are the ones with
\be
L_1=L_2=Q^\vee
\ee
where $Q^\vee$
 is the \emph{dual} root lattice of the Lie algebra\footnote{\ For $G(2,1,r)$ ($r\geq3$), which is the Weyl group of both $B_r$ and $C_r$, we have two distinct families of Looijenga lattices.} $\mathfrak{g}$. The polarization of $L(\tau)$ is principal iff
  $L_2=L_1^\vee$. Principal polarizations are in one-to-one correspondence with the (complete)  ``readings'' of $\cn=4$
  SYM with the gauge algebra $\mathfrak{g}$ (and coupling $\tau$);
they are permuted by $SL(2,\Z)$ acting on $\tau$
  \cite{readings}.

\begin{table}
\renewcommand{\arraystretch}{1.8}
$$
\begin{tabular}{c|p{12.5cm}}\hline\hline
$\Z[\lambda_G]$ & $G$\\\hline
$\Z[i]$ & $G(4,1,r)$ ($r\geq1$), $G(4,2,r)$ ($r\geq2$), $G(4,4,r)$ ($r\geq3$),
$G_8$ ($r=2$), $G_{29}$ ($r=4$), $G_{31}$ ($r=4$)\\\hline
$\Z[e^{2\pi i/3}]$ & $G(3,1,r)$ ($r\geq1$), $G(3,3,r)$ ($r\geq3$), 
$G(6,1,r)$ ($r\geq1$), $G(6,2,r)$ ($r\geq2$), $G(6,3,r)$ ($r\geq2$), $G(6,6,r)$ ($r\geq3$)
$G_4$ ($r=2$), $G_{5}$ ($r=2$), $G_{25}$ ($r=3$),
$G_{26}$ ($r=3$), $G_{32}$ ($r=4$), $G_{33}$ ($r=5$), $G_{34}$ ($r=6$)\\\hline
$\Z[\sqrt{-2}]$ & $G_{12}$ ($r=2$)\\\hline
$\Z\big[\frac{1+\sqrt{-7}}{2}\big]$ & $G_{24}$ ($r=3$)\\\hline\hline
\end{tabular}
$$\vskip-8pt
\caption{\label{tab}Unitary reflection groups whose defining field $\mathbb{K}_G$ is imaginary quadratic}
\end{table}

%
Next we consider the complex case $G\subset GL(r,\Z[\lambda_G])$.
Since $G$ acts irreducibly, there is a unique -- up to normalization -- $G$-invariant positive Hermitian form
defined over $\mathbb{Q}(\lambda_G)$
\be
H\colon \mathbb{Q}(\lambda_G)^r\times \mathbb{Q}(\lambda_G)^r\to \mathbb{Q}(\lambda_G)
\ee
This yields a $\mathbb{Q}$-valued  antisymmetric bilinear form
\be
\langle A,B\rangle= \frac{1}{\lambda_G-\bar\lambda_G}\Big(H(A,B)-H(B,A)\Big)\in\mathbb{Q}
\ee
which we normalize so that $\langle-,-\rangle$ is integral in $\Z[\lambda_G]^r$.
Then $\langle-,-\rangle$ is a Riemann form on the lattice $\Z[\lambda_G]^r$
which defines a polarization on the complex torus $\C^r/\Z[\lambda_G]^r$
making it into an Abelian variety $A$ on which $G$ acts by automorphisms.
Neglecting the sporadic reflection groups, for the infinite series $G\equiv G(m,p,r)$,
 the $G$-invariant Hermitian form ($\boldsymbol{a}\equiv(a_1,\dots,a_r)\in\R^r$ and similarly for $\boldsymbol{b},\boldsymbol{c},\boldsymbol{d}$)
\be
H(\boldsymbol{a}+\lambda_G\boldsymbol{b},\boldsymbol{c}+\lambda_G\boldsymbol{d})=
\sum_{i=1}^r (a_i+\overline{\lambda}_G b_i)(c_i+\lambda_G c_i)
\ee
yields a \emph{principal} polarization on the torus $\C^r/\Z[\lambda_G]^r$.
More general $G$-invariant lattices and polarizations will be discussed later in the paper.

\paragraph{Construction of $\bigstar$-geometries.} 
  
$G$ is an irreducible unitary reflection group which leaves invariant a full lattice
$\Lambda\subset V$ ($V\simeq\C^r$ is its reflection representation).   
We consider the $G$-representation $V\oplus V^\vee$. Let $w^i$ and $x_i$ be dual 
coordinates in, respectively, $V$ and $V^\vee$. We equip $V\oplus V^\vee$ with the
$G$-invariant (holomorphic) symplectic form
\be
\Omega=dw^i\wedge dx_i,
\ee   
 and take the quotient
\be
(V\oplus V^\vee)/(\Lambda\oplus 0)\equiv A\times V^\vee \xrightarrow{\ \pi\ } V^\vee.
\ee  
We endow the Abelian variety $A\equiv V/\Lambda$ with the minimal $G$-invariant polarization, so that $G\subset \mathsf{Aut}(A)$. Then we take the quotient
with respect to the diagonal action of $G$
\be\label{constr1}
\mathscr{X}\;\overset{\rm def}{=}\;(A\times V^\vee)/G_\text{diag}\xrightarrow{\ \pi\ } V^\vee/G= \C^r,
\ee
where the last equality follows from the Shephard-Todd-Chevalley theorem \cite{ST2}. 
The symplectic form $\Omega$ is invariant under the diagonal action of $G$,
so the smooth locus of $\mathscr{X}$ is symplectic, while the smooth fibers are 
Lagrangian and polarized Abelian varieties by construction. Thus \eqref{constr1} is
a $\C^\times$-isoinvariant special geometry with Euler vector
\be
\ce= x_i \frac{\partial}{\partial x_i}.
\ee
 Since the smooth fibers are all
isomorphic to the model fiber $A$, the geometry is \emph{isotrivial}. 
The Coulomb dimensions $\Delta_i$ are the degrees of the fundamental invariants of $G$ \cite{ST1,ST2} (in particular they are integers $\geq2$).
The geometry is interacting since $G$ is irreducible.

 The $\bigstar$-special geometries are the special class of isotrivial special geometries that can be constructed in the simple way \eqref{constr1} for some $G$, $\Lambda$, and compatible polarization.
 They have a section $s_0$ given by the zero element of $A$. 

\paragraph{Non-rigid $\bigstar$-geometries.} The $\bigstar$-geometries with a non-trivial
conformal manifold are expecially easy.
\begin{fact} An irreducible $\bigstar$-special geometry has a non-trivial conformal manifold
(i.e.\! $2$ is a Coulomb dimension) if and only $\mathbb{K}_G=\mathbb{Q}$ -- that is, iff
$G=W(\mathfrak{g})$ is the Weyl group of a simple Lie algebra $\mathfrak{g}$ -- in which case
the conformal manifold $\cm$ is a modular curve ($\equiv$ the quotient of the upper half-plane by a congruence subgroup of $SL(2,\Z)$). 
\end{fact}

Indeed, the dimension of the conformal manifold is the multiplicity of $2$ as a Coulomb dimension
and an irreducible reflection group has $2$ as a degree (with multiplicity precisely $1$) if and only if is
defined over a totally real number field. 
A $\bigstar$-special geometry with $2$ as a Coulomb dimension describes a $\cn=2$ SCFT
weakly-coupled at the cusps of $\cm$, hence a Lagrangian field theory. Comparing
the Coulomb dimensions, we conclude that it must be $\cn=2$ SYM with gauge algebra $\mathfrak{g}$
coupled to hypermultiplets in suitable representations of $\mathfrak{g}$ (which produce a vanishing $\beta$-function \cite{tachi}).
One class of SCFT with a $\bigstar$-geometry of this kind are the various readings of 
$\cn=4$ SYM with some gauge algebra $\mathfrak{g}$ \cite{readings}. When $W$ is the Weyl group of $C_r$
we have another Lagrangian SCFT: $Sp(r)$ SYM coupled to 4 fundamentals and one antisymmetric hypermultiplet which is the rank-$r$ Minahan-Nemeshanski (MN) model with $m=2$ \cite{higher MN}.
The unresolved $\bigstar$-geometry of this last model is equal to the geometry of a particular reading of
the $Sp(r)$ $\cn=4$ SYM, but it is more natural to think the MN geometry
to be described by the smooth Hilbert scheme of $r$-points in the corresponding $r=1$ geometry. In particular the second Betti number of the Hilbert scheme
is indeed the rank of the flavor group of the MN model plus $1$.

\paragraph{(Non)crepant $\bigstar$-geometries.}
The $\bigstar$-geometry in 
\eqref{constr1} has quotient singularities. 
The most severe quotient singularity is at $s_0(0)\in\mathscr{X}_0$.
One asks whether we may fix the singularities by
replacing $\mathscr{X}$ with a smooth resolution $\rho\colon\tilde{\mathscr{X}}\to \mathscr{X}$,
which is crepant (i.e.\! $\rho^*\Omega$ extends to a symplectic form on $\tilde{\mathscr{X}}$),
such that
\be
\tilde{\pi}\equiv \pi\circ\rho\colon\tilde{\mathscr{X}}\to\mathscr{C}
\ee
is a special geometry now with smooth total space. When all these conditions are met,
we call $\rho\colon\tilde{\mathscr{X}}\to \mathscr{X}$
a \emph{special resolution.}

The local geometry near $s_0(0)$ is modeled on the symplectic singularity
$\C^{2r}/G$. Therefore a necessary condition for the existence of a global special
resolution $\rho\colon\tilde{\mathscr{X}}\to \mathscr{X}$ is the existence of a local one.
We already know that this may happen only when 
\be\label{whhiG}
G= W(A_r),\quad G(m,1,r)\ \ m\in\{2,3,4,6\},\quad G_4.
\ee
For these groups \emph{a priori} we may have more than one special resolution.
\medskip

There are some obvious special resolutions. For $G(m,1,r)$ take the model fiber
to be $A=E^r$ ($E$ an elliptic curve with complex multiplication by $\mathbb{Q}(e^{2\pi i/m})$)
endowed with the natural product principal polarization. Then the total space in rank $r$
\be\label{HS1}
\mathscr{X}_r=(E^r\times\C^r)/(\Z_m^r\rtimes\mathfrak{S}_r)=\big((E\times \C)/\Z_m\big)/\mathfrak{S}_r=
\mathscr{X}_1^r/\mathfrak{S}_r
\ee 
is the symmetric $r$-th power of the rank 1 special geometry $\mathscr{X}_1$ which -- being a surface
with a canonical singularity -- has a crepant resolution $\tilde{\mathscr{X}}_1$. The Hilbert scheme of $r$-points
$\tilde{\mathscr{X}}_1^{[r]}$ is then a special resolution of $\mathscr{X}_r$. Likewise
when $G=W(A_r)$, the Looijenga variety $A$ fits in the exact sequence
\be
0\to A\times\C^r\to E^{r+1}\times \C^{r+1}\xrightarrow{\ \Sigma\times \Sigma\ } E\times\C\to0
\ee
where $\Sigma$ is the sum of coordinates. Then the $\bigstar$-geometry $\mathscr{X}_r$  
with group $W(A_r)$ and model fiber the Looijenga $W(A_r)$-variety fits in the sequence
\be\label{HS2}
0\to \mathscr{X}_r\to (E\times \C)^{r+1}/\mathfrak{S}_{r+1}\to E\times \C^r\to0.
\ee
Replacing the symmetry power of the surface $E\times \C$ with the corresponding Hilbert scheme,
we get a special resolution of $\mathscr{X}_r$. 

We phrase \eqref{HS2} in the physical language: the Hilbert scheme of $(r+1)$-points in 
$E\times\C$ is the special resolution of the $\bigstar$-geometry for $\cn=4$ SYM with
gauge group $U(r+1)$. Its local physics is $SU(r+1)$ SYM plus a decoupled free $U(1)$ model;
but now 
the $(r+1)$-ality of the allowed line operators is correlated to their $U(1)$ charges, so
globally the theory is not the product of two decoupled models ---
since the exact sequence \eqref{HS2} is not split.
Therefore when we decouple the $U(1)$ sector by hand we remain with
a smooth special geometry which is merely isogeneous, and not isomorphic,
to the $\bigstar$-geometry of $SU(r+1)$ SYM with
simply-connected gauge group. This isogeneous geometry has a polarization of degree $r+1$
whose line operators are the ones of the $U(r+1)$ theory with zero $U(1)$ charge. 

\medskip

The existence of a global special resolution for the full special geometry is
a stronger requirement than just the local resolution for the symplectic singularity $\C^{2r}/G$.
In the next subsection we argue (non-rigorously) that most local crepant resolutions cannot be
extended globally to the full $\bigstar$-geometry and that we remain with only the
 Hilbert schemes constructed above.

\subsection{Compact $\bigstar$-geometries}\label{s:compact}
In the construction \eqref{constr1} we replace the factor space $V^\vee$
with a second Abelian $G$-variety $B\equiv V^\vee/\Lambda^\prime$ (not necessarily isogeneous to $A$)\footnote{\ As always $Y$ is the normalization of $B/G$}
\be\label{constr2}
\mathscr{Y}\;\overset{\rm def}{=}\;(A\times B)/G_\text{diag}\xrightarrow{\ \pi\ } B/G= Y,
\ee
$\mathscr{Y}$ is now a \emph{compact} $\bigstar$-geometry over the normal
projective variety $Y$ (a compact ``special K\"ahler'' orbifold) with reflection group $G$.

\begin{defn} A \emph{crepant resolution of a compact special geometry}
$\mathscr{Y}\xrightarrow{\pi} Y$ -- where the ``Coulomb branch'' $Y$ is a (possibly non-smooth)
compact normal analytic variety -- is a crepant resolution $\rho\colon\tilde{\mathscr{Y}}\to\mathscr{Y}$ of the total space $\mathscr{Y}$
(with $\tilde{\mathscr{Y}}$ smooth and symplectic) such that we have a
commutative diagram
\be\label{diag}
\xymatrix{\tilde{\mathscr{Y}}\ar[d]_\rho\ar[rr]^{\tilde\pi} && Y\ar@{=}[d]\\
\mathscr{Y}\ar[rr]^{\pi} && Y}
\ee
with $\tilde\pi$ a proper fibration. By Matsushita theorem \cite{Mat1,Mat2}
$\tilde\pi$ is a Lagrangian fibration, and then the upper line of the diagram \eqref{diag}
is a compact special geometry with smooth total space and normal Coulomb branch.
The proper transform of the zero section of $\pi$ (when present) is a section of $\tilde\pi$.
\end{defn}
  
  \begin{rem} When $Y$ happens to be smooth,
it is necessarily a copy of $\mathbb{P}^r$ by a theorem of Hwang \cite{Hwang}.
However we do not make this assumption.
\end{rem}

A crepant resolution of a compact $\bigstar$-geometry is then an
isotrivial Lagrangian fibration of a smooth compact hyperK\"ahler manifold $\tilde{\mathscr{Y}}$
over the compact normal base $Y$. The isotrivial fibrations with these properties
have been recently classified in \cite{ABclass}.
According to \textbf{Theorem 1.4} of that paper, they appear in two types 
called A and B. The ones produced by construction \eqref{constr2} are of type A with a rational section.
\textbf{Theorem 1.6} of \cite{ABclass} fully classify them:

\begin{fact}[See \cite{ABclass}]Let $\tilde\pi\colon\tilde{\mathscr{Y}}\to Y$ ($
\tilde{\mathscr{Y}}$ smooth symplectic, $Y$ normal, both projective) be an isotrivial fibration of class A with a rational section. Then $\tilde\pi\colon\tilde{\mathscr{Y}}\to Y$ is a crepant resolution (as a special geometry!)
of a compact $\bigstar$-geometry given by the
 construction \eqref{constr2} where $G= W(A_r)$ or $G(m,1,r)$ ($m\in\{2,3,4,6\}$). In these cases
$Y\simeq\mathbb{P}^r$, and $\tilde{\mathscr{Y}}$ is a Hilbert scheme, namely the compact versions of the constructions in eqs.\eqref{HS1},\eqref{HS2}.
\end{fact}

We conclude that, while the symplectic singularity $\C^4/G_4$ has two non-isomorphic crepant resolutions, these local resolutions cannot be glued together to form a smooth \emph{compact}
special geometry $\tilde{\mathscr{Y}}$.  
 Since $G_4$ yields a $\bigstar$-geometry of dimension 4, we may
get the same conclusion using results in \cite{ou,HX22}.

\begin{rem} We saw in eq.\eqref{s:strong} that the symplectic singularities
$\C^{2r}/G(m,1,r)$ have a number of non-isomorphic crepant resolutions which gets
exponentially large as $r\to\infty$. The theorems in \cite{ABclass} imply that \emph{only one}
of this huge set of ``local'' resolutions glue together to form a
 resolution in the sense of compact special geometry. Indeed when the resolution of
the special geometry exists it is unique.
\end{rem}

We state the result in a different way, more in line with the physical discussion of \S.\ref{s:cartoon}:
\begin{fact} A necessary condition for a compact $\bigstar$-geometry with unitary reflection group $G$
to have a crepant resolution \emph{as a special geometry} is that there exists a
$G$-invariant Abelian variety $B$ such that $B/G$ is smooth in facts
\be
B/G\simeq\mathbb{P}^r.
\ee
\end{fact}
Inverting the role of the two factors $A$ and $B$ in eq.\eqref{constr2}
  we see that the same condition must hold
also for the variety $A$. We shall see later that when $G$ is a Weyl group (i.e.\! $W(A_r)$
or $G(2,1,r)\equiv W(C_r)$) the condition $A/G\simeq \mathbb{P}^1$
holds only for the respective families of Looijenga varieties $\C^r/(Q^\vee+\tau Q^\vee)$,
while for the complex groups $G(m,1,r)$ ($m\in\{3,4,6\}$) there is a \emph{unique} 
$G$-invariant variety $A=\C^r/\Lambda$ such that $A/G\simeq \mathbb{P}^1$,
that is, $A=\C^r/\Z[\lambda_G]^r$.
The $G$-varieties with this property have a principal polarization,
except for $G=W(A_r)$ where the polarization becomes principal after extending the gauge group
$SU(r+1)\leadsto U(r+1)$ as discussed after eq.\eqref{HS2}.

The crepant compact $\bigstar$-geometries have then the form
\be\label{87zzba}
\big(\C^{2r}/(\Lambda\oplus t \bar\Lambda)\big)/G\to (\C^r/t\bar\Lambda)/G\simeq \mathbb{P}^r,\quad t>0
\ee
(the overbar stands for complex conjugation).
Resolving the compact geometry \eqref{87zzba}, and then taking the decompactification limit
$t\to+\infty$, we get a resolution of the original non-compact $\bigstar$-geometry
with reflection group either $W(A_r)$ or $G(m,1,r)$ and a model fiber $A$
such that $A/G\simeq\mathbb{P}^1$ which, as we saw before, is a Hilbert scheme.
When $G=G(m,1,r)$ the crepant $\bigstar$-geometries
describe the rank-$r$ MN models.  

While there may be subtleties with the decompactification limit, we take
as our working hypothesis that no new crepant special resolutions emerge 
in the limit $t\to\infty$, and consider the above to be the full story. More precisely
we take the above discussion as circumstantial evidence that the sporadic group
$G_4$ behaves differently and that, presumably, it does not produce
fully fledged crepant $\bigstar$-geometries. This is in line with the physical intuition underlying
this paper. 
The sporadic group $G_4$ is also suspicious from the physical side since
it leads to absolutely incomplete
geometries.

\subsection{$\bigstar$-geometries: Stratifications} 
The Coulomb branch $\mathscr{C}$ of a special geometry carries two natural
stratifications: the \emph{$R$-stratification} \cite{inverse,char} and the \emph{$A$-stratification} \cite{char} which combine in the finer
$AR$-stratification.

\subsubsection{$R$-stratification} The $R$-strata are labelled by a subgroup $Z\subset U(1)_R$;
they are the domains in $\mathscr{C}$
where  $R$-symmetry $U(1)_R$ is broken exactly to $Z$. When the Coulomb dimensions $\Delta_i$
($i=1,\dots,r$) are integers (as in the $\bigstar$-geometries)
the open $R$-strata are associated to the divisors
$d\mid\mathrm{lcm}(\Delta_i)$ of the Coulomb branch dimensions
\be
\mathscr{C}^{(d)}=\big\{u_i\neq0\ \text{for }d\mid\Delta_i\ \text{and } u_j=0\ \text{for }d\nmid \Delta_j\big\}\subset\mathscr{C}.
\ee
In particular, the $R$-strata are connected and a union of $\C^\times$-orbits. 
The only closed $R$-stratum is the origin $\{0\}$ where the full $U(1)_R$ is unbroken.

For a $\bigstar$-geometry the $R$-stratification
 essentially reduces to the fundamental results of Springer theory for reflection groups \cite{ST2,spring1,spring2} as already discussed in  \cite{caorsi}.
 We sketch the story.
Suppose that the integer $d$ is \emph{regular} in the sense of Springer theory,\footnote{\ \textbf{Definition 11.21} of \cite{ST2}.}
then the fiber $\mathscr{X}_u$ over the generic point in the stratum $\mathscr{C}^{(d)}$
is smooth, i.e.\! a copy of $A$. Over a \emph{regular} $R$-stratum
 \be
 \overline{\mathscr{C}^{(d)}}=\big\{u_j=0\ \text{for }d\nmid \Delta_j\big\}\quad d\ \text{regular}
 \ee  
 we have a special geometry of rank $r(d)\equiv\dim\overline{\mathscr{C}^{(d)}}$.
 Indeed the automorphism $\exp(2\pi i\,\ce/d)$ of $\mathscr{X}$ fixes the fibers $\mathscr{X}_u\simeq A$
 over the points of $\overline{\mathscr{C}^{(d)}}$ \cite{char,inverse},
 and its action on $H_1(\mathscr{X}_u,\mathbb{Q})$ has the eigenvalue
 $e^{2\pi i/d}$ with multiplicity exactly $r(d)$. Hence we have an isogeny\footnote{\ Cf.\! \textbf{Theorem 13.6.3.} of \cite{cristine}.}
 \be
 X\times Y\to A
 \ee 
 where $X$ and $Y$ are Abelian varieties, with 
 \be
 \dim X=r(d)\quad\text{and}\quad 
 \Omega(\mathsf{Lie}\,X)=T^*\mathscr{C}^{(d)},
 \ee
  where $\Omega\colon T\mathscr{X}\to T^*\mathscr{X}$ is the symplectic form.
  Let $G_d\subset G$ (resp.\! $N_d\subset G$)
  the subgroup which fixes $\overline{\mathscr{C}^{(d)}}$ point-wise (resp.\! set-wise).
  $G_d\triangleleft N_d$, and the quotient group is a unitary reflection group acting irreducibly
  on $\mathsf{Lie}\,X$ by automorphisms of the fiber.\footnote{\ See \textbf{Theorems 11.15, 11.24, 11.38} of \cite{ST2}. Note that the action of $\exp(2\pi i\,\ce/d)$ on $X$ is induced by the action of an element of $G$.}
  Hence over a closed $R$-stratum $\overline{\mathscr{C}^{(d)}}$ with $d$
  a regular integer we have an induced rank-$r(d)$ $\bigstar$-geometry with
  unitary reflection group $N_d/G_d$ and model fiber isogeneous to $X$.
  When $d$ is not regular, the $R$-stratum is contained in the discriminant $\mathscr{D}$
  and the situation is best studied in terms of the $AR$-stratification.
  Vacua in an irregular $R$-stratum are physically intriguing:
  e.g.\! they are related to the subtler gapped vacua in \S.\,4 of \cite{donagi1}.

 \subsubsection{$A$-stratification.}
  We turn to $A$-stratification.
As before $V^\vee$ is the cover of the Coulomb branch on which $G$ acts through the
dual reflection representation. The branch locus of this cover is the set of
reflection hyperplanes of $G$ in $V^\vee$ which 
form a central hyperplane arrangement $\mathscr{A}\subset V^\vee\simeq\C^r$ \cite{Harrang}.
The intersections of families of hyperplanes in $\mathscr{A}$ are the \emph{flats} of $\mathscr{A}$
(the flat $V^\vee$ is the intersection of the empty family). The group $G$ permutes them.
The set $P(\mathscr{A})$
of all flats is a poset, in facts a \emph{geometric lattice}\footnote{\ See \textbf{Theorem 2.1} of \cite{Harrang}. Recall that a lattice is a poset such that the joint and the meet exists for all pairs of elements \cite{Harrang}. A geometric lattice has some additional properties, see \cite{Harrang}.} whose combinatorics is
conveniently represented by a Hasse diagram. A codimension-$s$
flat can be written (non uniquely!) as the intersection of $s$ reflection hyperplanes
$H_i\in\mathscr{A}$. Taking the image of $P(\mathscr{A})$
under the cover map $\varpi\colon V^\vee\to V^\vee/G\equiv\mathscr{C}$
we get
 a lattice of $\C^\times$-invariant submanifolds in the Coulomb branch.
The lattice
\be
\varpi P\equiv \{\varpi(\alpha)\}_{\alpha\in P(\mathscr{A})}
\ee
 gives the $A$-stratification of $\mathscr{C}$:
each $\varpi(\alpha)\subset\mathscr{C}$ is the \emph{closure} of an open $A$-stratum.
For each element of $\alpha\in P(\mathscr{A})$
let $G_\alpha\subset G$ be the subgroup which fixes
$\alpha\subset V^\vee$ pointwise. By  \textbf{Corollary 9.51} of \cite{ST2} this is the reflection group generated by the reflections in $G$ whose fixed hyperplanes
\emph{contain} $\alpha$.  Hence
\begin{fact} There is a one-to-one correspondence  between conjugacy classes of subgroups $S\subset G$ generated by a family of reflections in $G$ and Coulomb branch $A$-strata $\mathscr{C}_S\subset\mathscr{C}$
\be
[S]\leftrightarrow \mathscr{C}_S.
\ee
 {\rm The subgroup $S$ may be reducible. We think of the trivial group $1$ as the
reflection subgroup generated by an empty family of reflections.} If $S\subset G$ is a reflection subgroup,
the corresponding $A$-stratum is
\be
\mathscr{C}_S= \varpi V^\vee_S,\qquad V^\vee_S\overset{\rm def}{=}\big\{x\in V^\vee\ \text{such that } \{g\in G\; |\; gx=x\}=S\big\}.
\ee
All $A$-strata are connected and the union of $\C^\times$-orbits.
 Hence the only closed $A$-stratum is the origin
$\{0\}= \mathscr{C}_G$; all the other $A$-strata are open.
{\rm One has $\mathscr{C}_1=\mathscr{C}\setminus \mathscr{D}$.} In particular
\be
P(\mathscr{A})=\{V_S^\vee\colon S\subset G\, \text{a reflection subgroup}\}.
\ee 
\end{fact}

The subgroup $S$ fixes the subspace $V_S^\vee$ \emph{pointwise}.
Let $J_S\subset G$ be the subgroup which fixes $V_S^\vee$ \emph{as a set}.
Clearly $S\triangleleft J_S$ and we set $G_S=J_S/S$.

Fix a reflection subgroup $S\subset G$. 
$S$ acts rationally on $H_1(A,\mathbb{Q})$. We decompose this action in the maximal trivial subrepresentation
and the complementary one. The projector on the maximal trivial subrepresentation
\be
P_S=\frac{1}{|S|}\sum_{g\in S} g\colon \begin{cases} V\to V\\
V^\vee\to V^\vee\\
 H_1(A,\mathbb{Q})\to H_1(A,\mathbb{Q}) \end{cases}
\ee
is a $J_S$-map defined over $\mathbb{Q}$, that is, $P_S$ and $1-P_S$
are orthogonal idempotents contained in the center of $\mathbb{Q}[J_S]$.
 The multiplicity of the trivial representation in $H_1(A,\mathbb{Q})$ is twice
the complex dimension of $V^\vee_S\subset V^\vee$ as a vector subspace (since the rational representation is isomorphic over $\C$ to $V\oplus V^\vee$). 
The \emph{essential action} of $S$ is the non-trivial action on $(1-P_S)V$ as a reflection group
of degree 
\be
\mathrm{codim}(\mathscr{C}_S)\equiv r-\dim\mathscr{C}_S
\ee
The essential action of $S$
 controls the dynamics of the
degrees of freedom which become massless
 as we approach a point of $\mathscr{C}_S$. 
  
Using \textbf{Theorem 13.6.3.} of \cite{cristine} we get 
a $J_S$-equivariant isogeny
\be\label{crist}
A_S\times B_S\to A
\ee 
 where $A_S$ is an Abelian variety of dimension $\dim\mathscr{C}_S$ on which $S$ acts trivially.
Eq.\eqref{crist} is not an isomorphism in general. We replace $A_S$ by its
quotient $\tilde A_S\equiv A_S/K_S$ where $K_S$ is the kernel of the $G_S$-equivariant map
\be
A_S\times\{0\}\to A.
\ee
$\tilde A_S$ is an Abelian subvariety of $A_S$. We get the inclusion
\be
\tilde A_S\times V^\vee_S\subset A\times V^\vee.
\ee
The restriction of $\Omega$ to $\tilde A_S\times V^\vee_S$ is still a 
 symplectic form (away from higher codimension singularities) and the smooth fibers
 of $\tilde A_S\times V^\vee_S\to V^\vee_S$ are Lagrangian polarized Abelian varieties.
 Let $\tilde G_S$ be the image of $G_S$ in $\mathsf{Aut}(\tilde A_S)$
 ($\tilde G_S$ naturally acts on $V^\vee_S\simeq \mathsf{Lie}(\tilde A_S)^\vee$).
  Hence we have the \emph{isotrivial geometry}
 \be\label{uyazq}
 (\tilde A_S\times V^\vee_S)/\tilde G_S\to V_S^\vee/\tilde G_S,
 \ee  
with Coulomb branch $V_S^\vee/\tilde G_S$. We stress that, in general, this Coulomb branch
is \emph{not} isomorphic to the closure $\overline{\mathscr{C}_S}$ of the $S$-stratum, but merely a finite cover of its normalization.
The two spaces are isomorphic when  
 $G_S$ is a (typically reducible) unitary reflection group acting on $V_S^\vee$;
 this happens when, say, $G=G(m,1,r)$.
In any case to each closed $A$-stratum there is associated 
a, possibly reducible, $\bigstar$-geometry over a slightly different complex manifold. While $\overline{\mathscr{C}_S}$ is often non-smooth, its avatar $V_S^\vee/\tilde G_S$ is expected to be regular. 

\begin{exe} $G=G(m,1,r)$ contains two conjugacy classes of cyclic reflection groups of orders $m$ and $2$.
The first one $R_1$ is generated by the diagonal matrix $\mathrm{diag}(e^{2\pi i/m},1,\dots,1)$
and the second one $R_2$ by the permutation $w_1\leftrightarrow w_2$.
One has $\tilde G_{R_1}=G(m,1,r-1)$ and  $\tilde G_{R_2}=\Z_m\times G(m,1,r-2)$,
both reflection groups of rank $(r-1)$. Then $V^\vee_{R_a}/\tilde G_{R_a}\simeq\C^{r-1}$
and the base of the $\bigstar$-geometry along the codimension-1 strata is smooth.
Note that in the second instance the geometry is a product of a rank-1 geometry and a rank-$(r-2)$ one.
\end{exe}

\begin{exe} $G=W(A_2)$. The discriminant contains a single irreducible component which is the cuspidal 
cubic $u_2^2-u_1^3=0$ associated to the unique conjugacy class of reflection subgroups $R$. In this case $\tilde G_R=1$ and $V_R^\vee/G_R\simeq\C$ which is the normalization of the singular cubic.
\end{exe}
 
When, as in the last example, $G$ is defined over $\mathbb{Q}$
 the special geometries of \emph{all positive-dimension strata} have a
 conformal manifold of positive dimension, that is,
 a Lagrangian $\bigstar$-geometry is stratified in Lagrangian $\bigstar$-geometries.
 
 \subsubsection{The central fiber}\label{s:cfiber}
 
 The codimension-$r$ stratum is the origin $\{0\}$ of $\mathscr{C}$.
 The fiber is the normal projective variety
 \be
 A/G
 \ee
 where $A\equiv \C^r/\Lambda$ is the model fiber. From the criterion formulated
 at the end of \S.\,\ref{s:compact} we see that our $\bigstar$-geometry may be crepant
 only if -- possibly after replacing $A$ by an isogenous variety --
 the central fiber is smooth, hence isomorphic to $\mathbb{P}^r$. Now
 \be
 A/G=(\C^r/\Lambda)/G= \C^r/(\Lambda\rtimes G)\equiv \C^r/\cw\qquad \text{where }\ \cw\equiv \Lambda\rtimes G.
 \ee
 Then we may rephrase the natural questions in the Introduction in the more technical form
 \begin{que}
 What is the \underline{physical} mechanism which forces the (normalization of the unresolved)
 \emph{central fiber} $A/G$
 of the special geometry to be non-smooth?
 \end{que} 
 
 This question will keep us busy in \S.\,\ref{s:cartoon}.
 
 \medskip
 
 The topological argument of \cite{shva} says that the normalization of $\C^r/\cw$ may
 be simply-connected
only if the group $\cw\equiv \Lambda\rtimes G$ is generated by reflections, i.e.\!
 by affine maps $\C^r\to\C^r$ which leave fixed a hyperplane.
 We shall elaborate mathematically on this issue  in \S.\,\ref{s:strong};
 for the moment we focus on the physics of the issue which eventually
 will clarify what the abstract math theorems really mean.
 
 \subsection{IR physics of $A$-stratification: \emph{Root} lattices}\label{s:root}
 
At each stratum $\mathscr{C}_S$ of the $A$-stratification we have 
\emph{two} lower-rank special geometries \cite{inverse}.
 The first one is defined over the avatar of the closed stratum $\overline{\mathscr{C}_S}$, and is a \emph{bona fide}
 isotrivial special geometry, see eq.\eqref{uyazq},
  while the second one
 is defined over an infinitesimal tubular neighborhood $U_S$ of $\mathscr{C}_S$
 and makes sense only asymptotically as we approach $\mathscr{C}_S$ from the normal direction.
 This asymptotic geometry describes the effective $\cn=2$ SCFT which 
 governs the dynamics of the degrees of freedom which become light at $\mathscr{C}_S$.
In the RG language, this is the Seiberg-Witten geometry of the $\cn=2$ QFT obtained by
integrating out all degrees of freedom which remain massive along the stratum
$\mathscr{C}_S$. This geometric description
 becomes exact when the separation between heavy and light
modes gets parametrically large, i.e.\! in the asymptotic limit.
\medskip 
 
It is convenient to work in the branched cover $V^\vee\simeq\C^r$ of the Coulomb branch;
$V^\vee$  parametrizes SUSY vacua with the proviso that
 two points in the same $G$-orbit represent the same physical
vacuum. The cover $\varpi\colon V^\vee\to \mathscr{C}$ is branched
at the hyperplane arrangement $\mathscr{A}\subset V^\vee$. 
%

In a generic point of $V^\vee$ only $r$ IR-free photons
are light and the IR theory is free. 
Let us now approach a codimension-1 stratum, i.e.\! a generic point $x$
in the discriminant. Its closure is the $\varpi$-image
of the fixed hyperplane $H_\ell\subset V^\vee$ of the cyclic group of reflections of $G$ generated by a primitive reflection $g_\ell$.
We write $H^\vee_\ell\subset V$ for the fixed hyperplane of $g_\ell$
in the dual space $V$, and $\ell$ for the complex line orthogonal to $H_\ell^\vee$.
We call $\ell$ the \emph{root line} of the reflections at $H_\ell$. $\cl$ is the
set of all root lines of $G$.
Locally, at a generic point of $H_\ell$, the infinitesimal tubular neighborhood looks like
\be
U_\ell\simeq H_\ell\times D_\epsilon,
\ee
 where $D_\epsilon$ is a disk of radius $\epsilon$ and $U_\ell\cap H_\ell\simeq H_\ell\times\{0\}$. 
Locally in $U_\ell$ the cover $\varpi\colon U_\ell\to \varpi U_\ell$ is modeled on the branched cover
\be
H_\ell\times D_\epsilon\to H_\ell\times D_{\epsilon^m},\qquad (x,z)\mapsto (x,z^m)
\ee
where $m\in\{2,3,4,6\}$ is the order of the cyclic group of reflections at $H_\ell$,
equivalently of the local monodromy in the punctured
disk $\{x\}\times (D_{\epsilon^m}\setminus\{0\})\subset\mathscr{C}$.

Besides the everywhere massless photons,
the states which are light in $U_\ell$ are the BPS particles with electromagnetic charges
in the rank-2 symplectic sublattice
\be
\Lambda_\ell =\Lambda\cap \ell\subset \Lambda.
\ee  
Over $U_\ell$ we have the trivial family  $A\times U_\ell$ of Abelian varieties; 
fix a primitive element $\sigma$ in $\Lambda_\ell\subset H_1(A,\Z)$ and set
\be
a\equiv a(x)=\int_{\{x\}\times\sigma} \iota_\ce \Omega, \quad x\in U_\ell.
\ee
$a$ is a good  local coordinate in $U_\ell$ in the direction transverse to $H_\ell$ because $\Omega$ is symplectic
($a$ is a vanishing period of the SW differential, i.e.\! a special coordinate on the local cover $U_\ell$ centered at $H_\ell$). 
Then $a^m$ is a global coordinate
in $\varpi U_\ell\subset\mathscr{C}$ vanishing along the discriminant component $\varpi H_\ell$. We conclude that the asymptotic transverse special geometry at a generic point
$\varpi(x)\in\varpi H_\ell$ is a rank-1 $\bigstar$-geometry with charge lattice $\Lambda_\ell$
and Coulomb dimension $\Delta=m$. The light degrees of freedom along the discriminant component $\varpi H_\ell$
are  governed asymptotically  by some effective SCFT that we denote as $\ct_\ell$.
The codimension-$1$ transverse theories $\ct_\ell$ are described by the above asymptotic
special geometry modelled on the one for the rank-1 MN SCFT with the same
dimension $m$
(however there are finitely many \emph{distinct} rank-1 SCFT with the same geometry \cite{M1,M2,M3,M4,M5,caorsi2,caorsi3},
so, physically speaking, we still have some ambiguity in the asymptotic dynamics of these light modes).  
In particular all BPS states (operators) of $\ct_\ell$ correspond to a cycle in $H_1(A,\Z)\equiv \pi_1(A)$
which shrinks to zero as $a(x)\to0$.

\begin{figure}
$$
\includegraphics[width=0.35\textwidth]{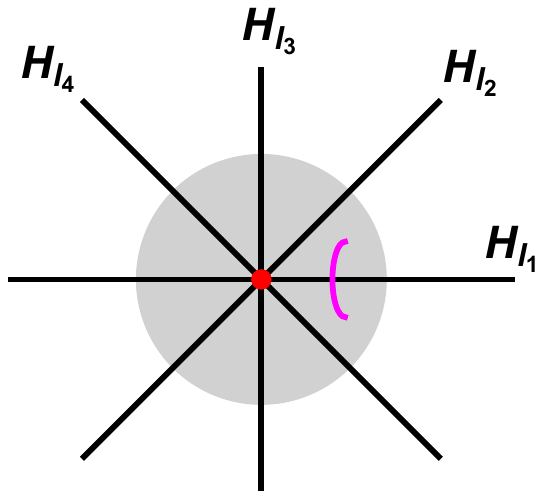}\quad\qquad\includegraphics[width=0.35\textwidth]{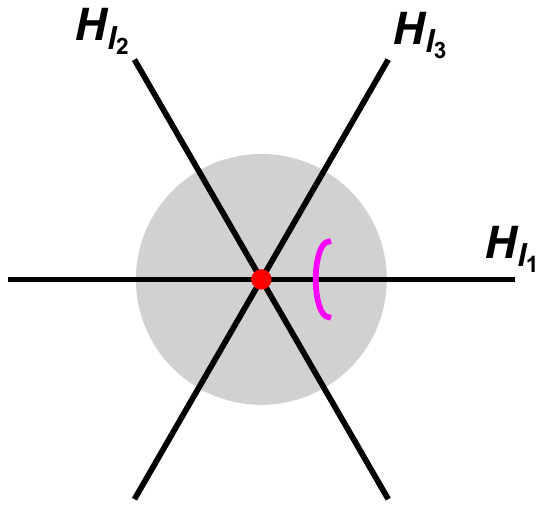}
$$
\caption{\label{fig}Two examples of covering stratification of $V^\vee$ in rank-2
SCFTs. The left figure corresponds to $SO(5)$ SYM with one vector and four spinor hypermultiplets,
while the right one to $\cn=4$ SYM with $\cg=SU(3)$. The black half-lines are the pre-images of the codimension-1 strata and the red point is the pre-image of the codimension-2 stratum. The gray disk represents an infinitesimal neighborhood of the codimension-2 stratum. The local monodromy at $H_{\ell_1}$ is computed along the magenta path which connects one point in the gray region to its
image under the reflection at $H_{\ell_1}$. This path projects to a closed loop in the Coulomb branch
which goes around the corresponding component of the discriminant.}
\end{figure}

\medskip

Consider next the pre-image in $V^\vee$ of a codimension-2
 stratum. It correspond to a locus $H_{\{\ell_1,\dots,\ell_s\}}$ which is the pairwise intersection of some $s$-tuple of fixed planes $H_{\ell_j}$ ($j=1,\dots,s$): 
see figure \ref{fig} for two typical examples in rank-2. Let 
$U_{\{\ell_1,\dots,\ell_s\}}$ be an infinitesimal tubular neighborhood of this codimension-2 locus (gray area in the figure). What is the IR physics at a generic vacuum in $U_{\{\ell_1,\dots,\ell_s\}}$? Clearly the states with
charges in the sublattices $\Lambda_{\ell_j}$ ($j=1,\dots,s$) are light in $U_{\{\ell_1,\dots,\ell_s\}}$,
 and, more generally, are light all states with charges in the sublattice 
\be
\Lambda^0_{\{\ell_1,\dots,\ell_s\}}\overset{\rm def}{=}\Lambda_{\ell_1}+\Lambda_{\ell_2}+\cdots+\Lambda_{\ell_s}\subset\Lambda
\ee
which has rank 4. The light degrees of freedom of the several SCFTs $\ct_{\ell_j}$
along the codimension-$1$ strata $H_{\ell_j}$
which meet at $U_{\{\ell_1,\dots,\ell_s\}}$
 are light in this codimension-2 locus and they interact with each other\footnote{\ However we stress that these degrees of freedom are \emph{not} all independent. When the simplest scenario applies, the degrees of freedom which are light on two distinct hyperplanes suffice to describe the physics at the intersection, that is, in the simplest scenario the states which are light along the other hyperplanes in the $s$-tuple are bound states of the ones which are light in the first two hyperplanes.} and produce bound states with charges in the lattice $\Lambda^0_{\{\ell_1,\dots,\ell_s\}}$.
In the simplest dynamical scenario this
is the full IR story at the codimension 2 locus. This cheap scenario is certainly a possibility: it happens,
say, for the SCFT represented in the left part of figure \ref{fig}.
But more intricate IR phenomena may also occur. This happens, say, for the SCFT
in the right part of the same figure which corresponds to $\cn=4$ SYM with gauge group the simply-connected form of $SU(3)$.

In the general case the IR effective theory
in $U_{\{\ell_1,\dots,\ell_s\}}$ is described by some sort of asymptotic $\bigstar$-SCFT
$\ct_{\{\ell_1,\dots,\ell_s\}}$ with reflection group
$G_{\{\ell_1,\dots,\ell_s\}}=\langle g_{\ell_1},\dots,g_{\ell_s}\rangle$ and charge lattice 
$\Lambda_{\{\ell_1,\dots,\ell_s\}}$. We know that 
\be
\Lambda^0_{\{\ell_1,\dots,\ell_s\}}\subset 
\Lambda_{\{\ell_1,\dots,\ell_s\}}
\ee
but the inclusion may be proper: when this is the 
case, $\ct_{\{\ell_1,\dots,\ell_s\}}$ contains light states with
charges in $\Lambda_{\{\ell_1,\dots,\ell_s\}}\setminus \Lambda^0_{\{\ell_1,\dots,\ell_s\}}$ which cannot arise as bound states of $\ct_{\ell_j}$ particles. More on the point, when
$\Lambda^0_{\{\ell_1,\dots,\ell_s\}}\neq 
\Lambda_{\{\ell_1,\dots,\ell_s\}}$
 the algebra of line operators of the effective theory $\ct_{\{\ell_1,\dots,\ell_s\}}$ is not fully generated by the line operators of the several theories $\ct_{\ell_j}$. $\ct_{\{\ell_1,\dots,\ell_s\}}$ contains additional 
 line operators with ``fractional''  charges valued in the finite group $\Lambda_{\{\ell_1,\dots,\ell_s\}}/\Lambda^0_{\{\ell_1,\dots,\ell_s\}}$. 

As an example, consider the right part of fig.\! \ref{fig}. By tuning the modulus of the fiber $A$,
we may take the low-energy effective theory in the gray region to be weakly-coupled $\cn=4$ SYM
with $\cg=SU(3)$, and the $SU(3)$ 't Hooft operator is part of its effective description;
however there is no way of writing this operator in terms of the light line operators of the effective
weakly-coupled $\cn=4$ $SU(2)$ theories which live along the discriminant components.
Thus the effective operator algebra at the crossing locus is not fully generated by the effective algebras along the several discriminant components;
 however all elements of the algebra
at the intersection has a power which belongs to the small algebra of lines with charges in $\Lambda^0_{\{\ell_1,\dots,\ell_s\}}$. 

\medskip

As we shall see momentarily, the
dycothomy between these two RG scenarios is reflected in the geometry as follows:
in the model in the left part of fig.\! \ref{fig} the fiber $\mathscr{X}_0$ over the origin is simply-connected,
while in the second model in the right figure it is not: one still has $H_1(\mathscr{X}_0,\R)=0$,
but now $H_1(\mathscr{X}_0,\Z)\simeq\Z_3$. The torsion classes may be thought of associated to the extra ``fractional'' light 't Hooft lines.
The homology group $H_1(\mathscr{X}_0,\Z)$ is torsion since a power of the 't Hooft line is part
of the small algebra.  
\medskip

However the charges which are ``new''  in codimension-2
\be
\Lambda_{\{\ell_1,\dots,\ell_s\}}^\text{new}\equiv\Lambda_{\{\ell_1,\dots,\ell_s\}}\setminus \Lambda^0_{\{\ell_1,\dots,\ell_s\}}
\ee  
cannot be arbitrary: the actual charge lattice 
$\Lambda_{\{\ell_1,\dots,\ell_s\}}$ should be consistent with Dirac quantization.
There are two issues: first the new charges should have integral Dirac pairing with the old ones,
and, second, the new charges should have integral pairing between themselves.

The first condition is conveniently restated as follows.
Let $R_{\ell_j}$ be the $2r\times 2r$ integral matrix which represents  on $H_1(A,\Z)$ the local monodromy around
$H_{\ell_j}$. One must have 
\be
\lambda\in\Lambda_{\{\ell_1,\dots,\ell_s\}}\quad\Rightarrow\quad(1-R_{\ell_j})\lambda\in\Lambda_{\ell_j}
\ \ \text{for all }j=1,\dots,s,\ee
or more invariantly, 
\be
\Lambda_{\{\ell_1,\dots,\ell_s\}}\subset \Lambda_{\{\ell_1,\dots,\ell_s\}}^*
\ee
where
\be\label{eeerzx}
\Lambda_{\{\ell_1,\dots,\ell_s\}}^*
\overset{\rm def}{=} \big\{\lambda\in \Lambda^0_{\{\ell_1,\dots,\ell_s\}}\otimes\mathbb{Q}\colon
 (1-g)\lambda\in \Lambda^0_{\{\ell_1,\dots,\ell_s\}}\ \forall\; g\in G_{\{\ell_1,\dots,\ell_s\}}\big\}. 
\ee

This story may be reloaded recursively in higher codimension.
Summarizing: the simplest RG scenario is that all states which are light in higher codimension may be thought of as bound states of BPS particles which are light at some discriminant component, that is, the total rank-$2r$ charge lattice $\Lambda$ has the form
\be\label{r1}
\Lambda=\Lambda^0\overset{\rm def}{=} \sum_{\ell\in\cl} \Lambda_\ell\equiv\sum_{\ell\in\cl}(\Lambda\cap \ell)
\ee
where the sum is over the root lines of $G$. 
While this simple scenario does occur, most $\bigstar$-SCFTs
have a more intricate dynamics and $\Lambda\neq \Lambda^0$.
However \emph{geometrically} we may always replace our $\bigstar$-geometry with an isogeneous
one where the simplest scenario holds
by making $\Lambda\leadsto \Lambda^0$.
\emph{Physically} this looks like truncating by hand the theory to a subset of its superselected sectors, and the simplest-scenario geometry may not be a complete description of any quantum consistent field theory.
However, it is technically useful to focus on such a subsector, study it in detail, and then extend
the analysis to the other sectors if present.

$\Lambda^0$ is a $G$-invariant lattice of rank $2r$.
Let $\{g_{\ell_1},\dots,g_{\ell_n}\}$ be a set of reflections which generate $G$ ($n=r$ or $r+1$ depending on $G$ \cite{ST2});
 we have \cite{popov}
\be\label{r2}
\Lambda^0=\sum_{j=1}^n \Lambda_{\ell_j}.
\ee

\begin{defn}
A rank-$2r$ $G$-invariant lattice of the form \eqref{r1} (equivalently \eqref{r2})
is called a \emph{root} lattice  \cite{popov}. 
\end{defn}

\subsubsection{Classification of $G$-invariant lattices}

The above physical discussion shows that all $G$-invariant lattices  $\Lambda$
have a root sublattice of finite-order  (just forget a few line operators!).
This entails that
the classification of the $G$-invariant lattices is as follows:

\begin{thm}[Popov, \textbf{Theorem 4.2.3} of \cite{popov}] Let $G$ be an irreducible unitary reflection group defined over $\mathbb{Q}$
or over an imaginary quadratic field. A (full) lattice $\Lambda\subset\C^r$ is $G$-invariant if and only if
it has the form (cf.\! eq.\eqref{eeerzx})
\be
\Lambda_\text{\rm root}\subseteq \Lambda \subseteq \Lambda^*_\text{\rm root}
\ee
for some root lattice $\Lambda_\text{\rm root}$, that is, all $G$-invariant lattice $\Lambda$
is the extension by a root lattice $\Lambda_\text{\rm root}$
\be
0\to \Lambda_\text{\rm root}\to\Lambda\to L\to 0
\ee
of a subgroup $L$ of the finite Abelian group
\be\label{uttt6z}
\Lambda^*_\text{\rm root}/\Lambda_\text{\rm root},
\ee
where
\be
\Lambda^*_\text{\rm root}\;\overset{\rm def}{=}\;\big\{\lambda\in \Lambda_\text{\rm root}\otimes\mathbb{Q}\colon
 (1-g)\lambda\in \Lambda_\text{\rm root}\ \forall\; g\in G\big\}.
\ee
If $\Lambda$ is a $G$-invariant lattice, the corresponding root sub-lattice
is
\be\label{assrot}
\Lambda_\text{\rm root}\equiv \Lambda^0\;\overset{\rm def}{=}\; \sum_{\ell\in\cl}(\Lambda\cap \ell).
\ee
\end{thm}
The construction of all $G$-invariant lattices proceed in two steps: first list down the
root lattices for each unitary reflection group $G$. This step is already done: the root lattices for all $G$'s are listed in tables 2, 3 of \cite{popov}. The second step is totally elementary: for each root lattice
$\Lambda_\text{root}$ one constructs the group \eqref{uttt6z} and its subgroups.

\begin{rem}\label{iuuuzqw} A root lattice may be a proper sublattice of another root lattice for the same $G$.
This happens, say, for $G(4,2,r)$, see table 2 in \cite{popov}. The same phenomenon takes place for
$G(3,1,r)$ and $G(4,1,r)$.
\end{rem}

This simple two-step algorithm produces all $G$-invariant rank-$2r$ lattices.
However to get the lattices which may be charge lattices
of a $\bigstar$-geometry we need a third step:
we have to require that the Dirac pairing is integral in the resulting lattice $\Lambda$
(where $\Lambda_\text{root}\subseteq \Lambda\subseteq\Lambda_\text{root}^*$).
The Dirac pairing induces a skew-symmetric bilinear form
\be\label{weil}
\Lambda^*_\text{root}/\Lambda_\text{root} \wedge \Lambda^*_\text{root}/\Lambda_\text{root} \to \mathbb{Q}/\Z
\ee
(called the \emph{Weil pairing} \cite{lange})
and we must require that the subgroup $L\subset \Lambda_\text{root}^*/\Lambda_\text{root}$ is \emph{isotropic} for this bilinear form.
Of course, any subgroup is isotropic for \emph{some} Weil pairing, e.g.\! for the trivial one. So we can always fix the problem with the Dirac pairing by changing the polarization (except when a specific pairing
is dictated by physical considerations).

\paragraph{Root Abelian varieties \& $\bigstar$-geometries}
\begin{defn} A \emph{root Abelian variety} $A_\text{root}$ for the (irreducible) unitary reflection group $G$
is a complex torus of the form $\C^r/\Lambda_\text{root}$,
where $\Lambda_\text{root}$ is a 
root $G$-lattice equipped with the unique minimal $G$-invariant polarization.
\end{defn}

\begin{defn} A \emph{root $\bigstar$-geometry} $\mathscr{X}_\text{root}$ is a $\bigstar$-geometry, with unitary reflection group $G$, whose model fiber $\C^r/\Lambda_\text{root}$ is a root  $G$-invariant Abelian variety.
\end{defn}

The root $\bigstar$-geometries precisely correspond to
putative $\cn=2$ $\bigstar$-SCFTs whose IR physics follows the simplest RG scenario,
that is, they have no new light charged particles, nor new light effective line operators, emerging in higher codimension. 

\medskip

The model fiber of any $\bigstar$-geometry may be written as $A_\text{root}/L$
for some root variety $A_\text{root}$ and subgroup $L\subset A_\text{root}^*/A_\text{root}$. Then 
\begin{fact}
All $\bigstar$-geometry is isogeneous to a root $\bigstar$-geometry. 
\end{fact}
To focus on a subsector of the theory, 
 we replace a geometry with model fiber $A$ with an isogeneous one with
fiber $A_\text{root}$ such that $A=A_\text{root}/L$. We do not have
the luxury of choosing an independent polarization, and we need to require that the subgroup $L$ is isotropic for the original Weil pairing on $A$ dictated by the physical Dirac pairing.

\medskip

The above RG arguments suggest that the root $\bigstar$-geometries are the less singular
ones in their isogeny class (there may be more than one in an isogeny class).
The central fiber $\mathscr{X}_{\text{root},0}$ of a root $\bigstar$-geometry -- hence the full $\bigstar$-geometry $\mathscr{X}_\text{root}$ -- is \emph{simply-connected} since
all elements of $H_1(A_\text{root},\Z)\simeq \Lambda_\text{root}$ 
may be written as a sum of elements in some $\Lambda_\ell\subset H_1(A_\text{root})$
which are represented by loops which shrink to zero in $\pi^{-1}(H_\ell)$ hence \emph{a fortiori}
 in the central fiber.
On the contrary for a general $G$-invariant lattice $\Lambda$
the loops representing non-trivial classes in $\Lambda/\Lambda^0$ will not
shrink; they represent the ``extra'' line operators present in the subtler RG scenario. 
From the characterization of root lattices in \S.\,\ref{s:strong}
and $H_\bullet(\mathscr{X},\Z)\simeq H_\bullet(\mathscr{X}_0,\Z)$
we get

\begin{fact} If $\mathscr{X}$ is a non-root $\bigstar$-geometry, $H_1(\mathscr{X},\Z)$
is a non-trivial torsion group.
\end{fact}

\subsubsection{Examples}\label{s:examples}
To illustrate the method to construct all $G$-invariant lattices, we run the algorithm for the imprimitive groups $G(m,q,r)$
($m\in\{2,3,4,6\}$ and $q\mid m$). This subsection may be skipped in a first reading.
The case $r=2$ is special since
\be
\begin{gathered}
\mathfrak{so}(5)\simeq\mathfrak{sp}(2),\quad G(2,2,2)\simeq W(A_1)\times W(A_1),\quad G(3,3,2)\simeq W(A_2),\\
G(4,4,2)\simeq	W(B_2),\qquad G(6,6,2)\simeq W(G_2),
\end{gathered} 
\ee
 and we assume $r\geq3$ unless explicitly stated otherwise.   
     
     \begin{exe}\label{5mafia} Consider the imprimitive reflection group $G(2,1,r)\simeq W(B_r)\simeq W(C_r)$.
We write $e_i$ ($i=1,\dots,r$)
     for the usual orthonormal basis of $\C^r$.
     There are four $W(B_r)$-invariant lattices lattices in $\R^r$ (cf.\! \cite{bourbaki} \textsc{plache II,III}):
     \be
     \begin{aligned}
   Q(B_r)&=\sum\nolimits_i \Z e_i=Q(B_r)^\vee\\
   Q(C_r)&= \left\{\sum\nolimits_i a_i e_i\colon a_i\in\Z,\ \sum\nolimits_ia_i=0\bmod2\right\}\\
   P(B_r)&= \sum\nolimits_i \Z e_i+\Z \frac{1}{2}\sum\nolimits_i e_i\equiv Q(C_r)^\vee\\
  P(C_r)&=Q(C_r)+\Z e_1\equiv Q(B_r),
     \end{aligned}
     \ee
     where $Q(\mathfrak{g})$ and $P(\mathfrak{g})$ stand, respectively, for the root and weight lattices of $\mathfrak{g}$ and the superscript $\vee$ stands for the dual lattice (do not confuse with the dual root system $Q^\vee(\mathfrak{g})$!).
Table 2 of \cite{popov} lists 5 one-parameter families of \emph{root} lattices in $\C^r$ for $G(2,1,r)$:
 \be\label{rooTB}
 \begin{aligned}
 \Lambda_\text{root}^{(1)}&=\sum\nolimits_i(\Z+\tau\,\Z)e_i= Q(B_r)+\tau\, Q(B_r)\\
 \Lambda_\text{root}^{(2)}&=Q(B_r)+\frac{1+\tau}{2} Q(C_r)\\
 \Lambda_\text{root}^{(3)}&=\frac{1}{2}Q(C_r)+\tau\, Q(B_r)\\
 \Lambda_\text{root}^{(4)}&=Q(B_r)+\frac{\tau}{2}\, Q(C_r)\\
 \Lambda_\text{root}^{(5)}&=\frac{1}{2}Q(C_r)+\frac{\tau}{2}\, Q(C_r)
 \end{aligned}
 \ee
 where $\tau\in\mathbb{H}$ is a point in the upper half-plane.
 $\C^r/\Lambda_\text{root}^{(1)}$ and $\C^r/\Lambda_\text{root}^{(5)}$
are the Looijenga Abelian varieties for, respectively, $Sp(r)$ and $\mathsf{Spin}(2r+1)$.
The interpretation of these 5 root lattices in terms of affine Lie algebras will be given
in \S.\,\ref{s:affaff}.
 
 The 5 lattices $\Lambda_\text{root}^{(a)*}$ are obtained from \eqref{rooTB} by the replacements
 $Q(B_r)\leadsto P(B_r)$, $Q(C_r)\leadsto P(C_r)$, so that $|\Lambda_\text{root}^{(a)*}/\Lambda_\text{root}^{(a)}|=4$ for $a=1,\dots,5$.
The invariant sublattices which admit a principal polarization are the ``readings''
of $\cn=4$ SYM with $\cg=Sp(r)$ \cite{readings}:
\begin{small}\be
\begin{tabular}{ll||ll}
gauge theory & lattice & gauge theory & lattice\\\hline
$\mathsf{Spin}(2r+1)$ & $P(B_r)+\tau\,Q(C_r)$ & $(Sp(r)/\Z_2)_+$ & $Q(C_r)+\tau\, P(B_r)$\\
$SO(2r+1)_+$ & $Q(B_r)+\tau\,P(C_r)$ & $Sp(r)$ & $P(C_r)+\tau\, Q(B_r)$\\
 $SO(2r+1)_-$ & $Q(B_r)+\tau\,Q(C_r)+\Z f(\tau)$ & $(Sp(r)/\Z_2)_-$ & $Q(C_r)+\tau\, Q(B_r)+\Z g(\tau)$\\\hline
\end{tabular}
\ee\end{small}
\be
\text{where}\qquad f(\tau)=\frac{1}{2}\sum\nolimits_i e_i+\tau\,e_1,\qquad g(\tau)= \frac{\tau}{2}\sum\nolimits_i e_i+e_1.
\ee
The $Sp(r)$ reading is special since its lattice is the root lattice $\Lambda^{(1)}_\text{root}$
which -- as it is obvious from the QFT side -- coincides with the charge lattice of $Sp(r)$ SYM with 4 fundamentals and one antisymmetric hyper, which -- as all MN models -- must have a charge lattice
which is both root and principal. As already mentioned, the corresponding Abelian variety $\C^r/\Lambda^{(1)}_\text{root}$ is the $Sp(r)$
Looijenga one. 
     \end{exe}
     
     \begin{exe}\label{exemm} We consider the groups $G(m,m,r)$ with $m\in\{2,3,4,6\}$ and $r\geq3$.
These groups have a unique invariant root lattice
 \be\label{mmroot}
 \begin{aligned}
&m=2 &\quad\Lambda_2&=\Big\{\sum\nolimits_j(a_j+\tau b_j)e_j\colon a_j,b_j\in\Z,\ \sum\nolimits_j a_j=\sum\nolimits_j b_j=0\bmod2\Big\}\\
&m=3 &\quad\Lambda_3&=\Big\{\sum\nolimits_j(a_j+e^{2\pi i/3} b_j)e_j\colon a_j,b_j\in\Z,\ \sum\nolimits_j (a_j+b_j)=0\bmod3\Big\}\\
&m=4 &\quad\Lambda_4&=\Big\{\sum\nolimits_j(a_j+i\, b_j)e_j\colon a_j,b_j\in\Z,\ \sum\nolimits_j (a_j+b_j)=0\bmod2\Big\}\\
&m=6 &\quad\Lambda_6&=\Big\{\sum\nolimits_j(a_j+e^{2\pi i/6} b_j)e_j\colon a_j,b_j\in\Z\Big\},
 \end{aligned} 
 \ee
 where $\tau\in\mathbb{H}$. Then
     \be\label{uuuy778}
     \begin{aligned}
     &m=2 &\quad\Lambda_2^*&=\sum\nolimits_j(\Z+\tau\Z)e_j+\frac{1}{2}\Z\sum\nolimits_je_j+\frac{\tau}{2}\Z\sum\nolimits_j e_j\\
    &m=3 &\quad \Lambda_3^*&=\sum\nolimits_j(\Z+e^{2\pi i/3}\Z)e_j+\frac{1}{3}\Z(1+2 e^{2\pi i/3})\sum\nolimits_j e_j\\ 
     &m=4 &\quad\Lambda_4^*&=\sum\nolimits_j(\Z+i\,\Z)e_j+\frac{1}{2}\Z(1+i)\sum\nolimits_j e_j\\ 
     &m=6 &\quad\Lambda_6^*&=\sum\nolimits_j(\Z+e^{2\pi i/6}\Z)e_j\equiv \Lambda_6.
     \end{aligned}
     \ee   
   For $m=2$, $G(2,2,r)\simeq W(D_r)$, the root lattice is $\Lambda=Q(D_r)+\tau Q(D_r)$
     and 
     \be
     \Lambda^*=P(D_r)+\tau P(D_r).
     \ee
      Neither is principal, but there are plenty of principal invariant lattices, in one-to-one correspondence with the readings of $\cn=4$ SYM with $\mathfrak{g}=D_r$
     \cite{readings}. For $G(6,6,r)$ ($r\geq3$!) we have a \emph{unique} invariant lattice
     \be
     \Lambda_6=\sum\nolimits_j(\Z+e^{2\pi i/6}\Z)e_j,
     \ee
     which is both root and principal. When $m=3,4$ and $r\geq3$,
 we have
 \be
 [\Lambda_m^*:\Lambda_m]=p^2\quad\text{where }p\ \text{is prime with }p\mid m\ \text{i.e. }p=\begin{cases}3 &m=3\\
 2 & m=4.\end{cases}
 \ee
  In facts, for $m=3,4$
  \be
  \Lambda_m^*/\Lambda_m= \Z/p\Z\oplus\Z/p\Z.
  \ee
In these two cases, besides the invariant lattices
 $\Lambda_m$ and $\Lambda_m^\ast$ we have the 3 invariant ``GSO-projection'' lattices:
   \be
\begin{aligned}
&\Lambda_m^v=\sum\nolimits_j(\Z+e^{2\pi i/m}\Z)e_j\\
&\Lambda_m^s=\Lambda_m+\Z \frac{1}{p} \big(1+(p-1)e^{2\pi i/m}\big)\sum\nolimits_je_j\\
&\Lambda_m^c=\Lambda_m+\Z \Big[\frac{1}{p} \big(1+(p-1)e^{2\pi i/m}\big)\sum\nolimits_je_j+e_r\Big]
\end{aligned}
   \ee
   These 3 lattices are principal.
   \begin{rem} Note that in this example the lattices $\Lambda_m^*$ coincide with the dual lattices $\Lambda_m^\vee$.
   \end{rem} 
     \end{exe}
     
     \begin{exe} Consider $G(6,q,r)$, $q\mid 6$ ($r\geq3$!). Since $G(6,6,r)\triangleleft G(6,q,r)$,
     while $G(6,q,r)$ must have at least one invariant root lattice,
     we conclude from eq.\eqref{uuuy778} that for all $q\mid6$ there is a \emph{unique} invariant lattice
     \be
     \Lambda=\sum\nolimits_j(\Z+e^{2\pi i/6}\Z)e_j
     \ee
 which is automatically root. It is obviously principal for all $q\mid 6$.
 \end{exe}
     \begin{exe} 
    Next we study the invariant lattices for $G(3,1,r)$ and $G(4,1,r)$. 
    Since $G(m,m,r)\triangleleft G(m,1,r)$, the invariant lattices must be in the list of 5 lattices
    we found in {example \ref{exemm}} for $m=3,4$.
    The root lattice $\Lambda_m$ in eq.\eqref{mmroot} is invariant also for $G(m,1,r)$
    and is \emph{a fortiori} a root lattice for the latter group that we call $\Lambda^{(2)}_m$. Also the primitive lattice
    $\Lambda_m^v$ is now an invariant root lattice that we call $\Lambda^{(1)}_m$. The two ``spinor'' lattices $\Lambda^s_m$,
    $\Lambda^c_m$ are not invariant under the bigger group. We have
    \be
    \Lambda_m^{(1)*}=\Lambda_m^\ast,\qquad \Lambda_m^{(2)*}=\Lambda_m^{(1)},
    \ee
    where $\Lambda_m^\ast$ is the lattice in eq.\eqref{uuuy778} (with $m=3,4$). 
    Only $\Lambda_m^{(1)}$ is principal. Note that for these lattices $\Lambda_m^{(1)*}=\Lambda_m^{(2)\vee}$.
        \end{exe}
     
     \begin{exe}\label{42r} For $G(4,2,r)$, the candidate invariant lattices are again the 5
     ones in {example \ref{exemm}} for $m=4$. It is immediate to check that all 5 are invariant.
$\Lambda_4$ and $\Lambda_4^v$ are root lattices on the nose with $\Lambda_4\subset \Lambda_4^v$.
cf. \textbf{Remark \ref{iuuuzqw}}. $\Lambda_4^v$ is root and principal. The two lattices $\Lambda_4^s$, $\Lambda_4^c$ are principal and non-root. The lattice $\Lambda_4^\vee\supset \Lambda_4^v$
is non-root and non-principal.
     \end{exe}

 \subsection{Classification of ``classical'' $\bigstar$-geometries}\label{s:root2}
 
 By the construction \eqref{constr1} we see that, modulo isomorphism of complex spaces, 
 the classification of all irreducible $\bigstar$-geometries is equivalent to the classification of
 all invariant lattices for the irreducible unitary groups with $\mathbb{K}_G$ rational or imaginary quadratic. 
 Since all $\bigstar$-geometry is isogeneous to a root one,
their classification \emph{modulo isogeny} is given\footnote{\ There may be more than one root lattice in the same isogeny class, cf.\! \textbf{Remark \ref{iuuuzqw}}. Thus the list of \cite{popov} produces some repetition modulo isogeny.} by the list of root lattices in table 2 of \cite{popov}
 with the types $[G(4,2,s)]^*_1$, $[K_{12}]^*$ and $[K_{31}]^*$ removed from the list.\footnote{\ These types will lead to geometries without zero section.}
 For each isogeny class with model fiber $A_\text{root}$ we determine the 
 finite group
 \be
 M=\big\{a\in A_\text{root}\colon g a=a\ \text{for all }g\in G\big\}\subset A_\text{root}
 \ee
 and we have one $\bigstar$-geometry with model fiber $A_\text{root}/L$ per subgroup $L\subset M$.
 The algorithm to construct all such $L$'s is straightforward.
   However, since classification is not the primary goal of this note, we shall run it for all
   but finitely many sporadic $G$'s and write all $\bigstar$-geometries up to a few exceptional cases
   (many of which are absolutely incomplete). Our list is exhaustive for ranks $r\geq7$. The classification for ranks $r\leq6$ will be given elsewhere.
 
 \subsubsection{Non-rigid $\bigstar$-geometries: connection to affine Lie algebras}\label{s:affaff}
  The non-rigid $\bigstar$-geometries are the interacting ones with model fiber isogeneous to $E_\tau^r$, $E_\tau$ a general elliptic curve of period $\tau$.
They come in families parametrized by modular curves $\cm=\mathbb{H}/\boldsymbol{\Gamma}$.
$G$ is the Weyl group $W(\mathfrak{g})$ of a simple Lie algebra $\mathfrak{g}$ of rank $r$.

A non-rigid $\bigstar$-geometry describes a $\cn=2$ SCFT with conformal manifold the curve $\cm$. The SCFT becomes weakly coupled at the cusps of $\cm$ (in the appropriate duality frame); hence the underlying SCFT must be a Lagrangian theory: $\cn=2$ SYM 
 with a simple gauge algebra $\mathfrak{g}$ coupled to matter in suitable representations \cite{tachi}.
 The condition that the special geometry belongs to the $\bigstar$-class sets very strong constraints
 on the matter sector. If we require a maximal set of
 line operators to be present, its $W(\mathfrak{g})$-lattice $\Lambda$ must be principal.
 
 To avoid exceptional cases we assume rank $r\geq3$.
 
 \paragraph{Relation to affine Lie algebras.}
 Consider the set of Dynkin graphs $\varGamma$ of the affine Lie algebras drawn in the precise shape
  as in tables I-II-III of ref.\!\cite{kac2} to which we refer for conventions. If we delete the leftmost node of $\varGamma$ we get the Dynkin graph of a simple Lie algebra. We write $\overline{W}(\varGamma)$ for the (finite) Weyl group of this Lie algebra.    
 One has
 \begin{thm}[\textbf{Theorem 3.1} in \cite{rus1} see also \cite{kac2}] The root lattices with $G=W(\mathfrak{g})$
 can be read from the set of affine Dynkin graphs with $\overline{W}(\varGamma)=W(\mathfrak{g})$.  $B^{(1)}_r$ and 
 $C^{(1)}_r$ produce equivalent root lattices for $W(B_r)\equiv W(C_r)$.
 \end{thm}
 
 There are 5 affine Dynkin graphs with $\overline{W}(\varGamma)\simeq G(2,1,r)$ ($r\geq3$)
 see figure \ref{weylC}. They correspond to the 5 families of root lattices in eq.\eqref{rooTB}.
 However only 4 families are actually distinct, while the difference between $A^{(2)}_{2r}$
 and $D^{(2)}_{r+1}$ turns out to be inessential for most ``classical'' purposes \cite{rus2}. 
 
  \begin{figure}
\begin{align*}
&\begin{smallmatrix}\text{affine}\\ \text{algebra}\end{smallmatrix} &&\qquad\text{Dynkin graph} && \begin{smallmatrix}\text{weighted projective}\\ \text{\emph{scheme}}\end{smallmatrix}\\
&B_r^{(1)} &&\begin{gathered}\xymatrix{& 1\ar@{-}[d]\\
1\ar@{-}[r] & 2 \ar@{-}[r] & 2 \ar@{-}[r] &\cdots\cdots \ar@{-}[r] & 2 \ar@{=>}[r]& 2}\end{gathered}
&&\mathbb{P}(1,1,2,\dots,2)\\
\\
&C_r^{(1)} &&\xymatrix{
1\ar@{=>}[r] & 2 \ar@{-}[r] & 2 \ar@{-}[r] &\cdots\cdots \ar@{-}[r] & 2 \ar@{<=}[r]& 1}
&&\mathbb{P}(1,1,2,\dots,2)\\
\\
&A^{(2)}_{2r} &&\xymatrix{
2\ar@{<=}[r] & 2 \ar@{-}[r] & 2 \ar@{-}[r] &\cdots\cdots \ar@{-}[r] & 2 \ar@{<=}[r]& 1}
&&\mathbb{P}(1,\dots,1)\\
\\
&A^{(2)}_{2r-1} &&\begin{gathered}\xymatrix{& 1\ar@{-}[d]\\
1\ar@{-}[r] & 2 \ar@{-}[r] & 2 \ar@{-}[r] &\cdots\cdots \ar@{-}[r] & 2 \ar@{<=}[r]& 1}\end{gathered}
&&\mathbb{P}(1,1,1,2,\dots,2)\\
\\
&D^{(2)}_{r+1} &&\xymatrix{1\ar@{<=}[r] & 1 \ar@{-}[r] & 1 \ar@{-}[r] &\cdots\cdots \ar@{-}[r] & 1 \ar@{=>}[r]& 1}&&\mathbb{P}(1,\dots,1)
\end{align*}
\caption{\label{weylC}The 5 affine Dynkin graphs with the property that deleting the leftmost node we get the Dynkin graph of a finite-dimensional Lie algebra with Weyl group $G(2,1,r)\equiv W(B_r)\equiv W(C_r)$. The integer attached to each node is its Coxeter label.
These 5 graphs lead to just 3 different weighted projective \emph{schemes} but, as we shall see in \S.\,\ref{s:strong}, to \emph{four} diverse weighted projective \emph{stacks}. The isomorphism of the stacks associated the first two graphs may be seen as a manifestation of $S$-duality at the level of the full ``quantum'' geometry.}
\end{figure}
 
Summarizing, we have three classes of simple gauge algebras $\mathfrak{g}$:
 \begin{itemize}
 \item[1.] \textbf{Simply-laced:} $\mathfrak{g}\in ADE$. There is a single family of root lattices
 $\Lambda_\text{root}$  which produces the Looijenga varieties $\C^r/(Q^\vee+\tau Q^\vee)$
  of $\mathfrak{g}$ whose polarization has a degree equal to
$|Z(\cg)|$ where $\cg$ is the associated simply-connected Lie group;
 \item[2.] \textbf{Non-simply-laced with $P(\mathfrak{g})=Q(\mathfrak{g})$}:  
$F_4$ and $G_2$ (cf.\! \cite{bourbaki} \textsc{plache VIII,\,IX}).
Table 2 of \cite{popov} lists 3 families of root lattices for $F_4$ and 4 for $G_2$
but there are identifications (table 3 of the same paper) so the essentially distinct
families are just 2 for both groups. 
Two families of root Abelian varieties one principal (degree 1)
 and one of degree $4$ for $F_4$ and, respectively, degree $3$ for $G_2$. The non-principal root varieties are the Looijenga ones for both $F_4$ and $G_2$;
 \item[3.] $W(B_r)\equiv W(C_r)$: see example \ref{5mafia} and figure \ref{weylC}.
The graph $D^{(2)}_{r+1}$ gives the Looijenga variety for $C_r$ ($r\geq1$)
whose (minimal) polarization has degree $1$. The graph $A^{(2)}_{2r-1}$ yields the
Looijenga variety for $B_r$ ($r\geq3$) with degree 4. The root variety associated to
$A^{(2)}_{2r}$ ($r\geq2$) is also principal. The one associated to the two affine graphs
$B^{(1)}_r$ and $C^{(1)}_r$ has degree 2.    
 \end{itemize}
 
As already remarked, the principal $W(\mathfrak{g})$-invariant lattices are in one-to-one
correspondence with the readings of $\cn=4$ SYM with gauge algebra $\mathfrak{g}$.
The $C_r$ Looijenga variety is principal and corresponds -- besides to $\cn=4$ with
gauge group $Sp(r)$ -- also to the $D_4$ MN model of rank-$r$. A particularly
interesting reading of $\cn=4$ SYM is the one associated with the simply-connected form of the gauge group $\cg$:
\begin{fact} The model fiber $A_\cg$ of the $\bigstar$-geometry of $\cn=4$ SYM with simply-connected gauge group $\cg$
is isogeneous to the Looijenga variety $A^\text{\rm Lo}_\cg$ of type $\cg$
\be
\phi\colon A^\text{\rm Lo}_\cg\to A_\cg
\ee 
where the kernel of $\phi$ is a subgroup invariant under the Weyl group which is totally isotropic for the Looijenga Weil form. 
\end{fact}
 This \textbf{Fact} is Remark 5.5 in \cite{rootlattices} where the Abelian varieties $A_\cg$ are studied in detail. 
 
 \subsubsection{Rigid $\bigstar$-geometries}
 
These are the $\bigstar$-geometries whose unitary reflection groups $G$ are 
defined over an imaginary quadratic field. Up to 12 exceptions in ranks $r\leq6$,
$G$ must have the form $G(m,q,r)$ with $m\in\{3,4,6\}$ and $q\mid m$.
In rank $\geq3$ the classification follows from the examples in \S.\,\ref{s:examples}.
We list the possible irreducible $\bigstar$-geometries for each allowed $r$-tuple of Coulomb dimensions
when $G=G(m,q,r)$. This is the complete list of rigid interacting $\bigstar$-geometries for $r>6$.
For each geometry we say if it has or not a crepant resolution, a result that will be established in \S\S.\,\ref{s:cartoon},\ref{s:strong}.  

\paragraph{Coulomb dimensions $\{6,12,\dots, 6(r-1),6r\}$.}
The $\bigstar$-geometry is based on $G(6,1,r)$
and there is just one invariant lattice which is both root and principal,
with a special resolution.
At the ``classical'' level of eq.\eqref{constr1} there is just one $\bigstar$-geometry
with these dimensions: it is the rank-$r$ MN of type $E_8$.
However, this unique ``classical'' geometry
may be endowed with  several inequivalent stacky structures, see \S.\,\ref{s:strong}.
Hence at the face value we have still room for an even finer ``quantum'' classification of such geometries
which perhaps may be of physical significance.

\paragraph{Coulomb dimensions $\{6,12,\dots,6(r-1), kr\}$ with $k=1,2,3$.} These $\bigstar$-geometries are associated to the reflection group $G(6,6/k,r)$.
For each $k$ there is a unique ``classical'' $\bigstar$-geometry which is both root and principal,
but now never crepant. Again the ``quantum'' story may be subtler (\S.\,\ref{s:strong}).

\paragraph{Coulomb dimensions $\{3,6,\dots,3(r-1), 3r\}$ and $\{4,8,\dots,4(r-1), 4r\}$.} 
These geometries have group $G(3,1,r)$ and $G(4,1,r)$ respectively.
For each $m=3,4$ we have 3 ``classical'' $\bigstar$-geometries. 
The first one, based on the lattice $\Lambda^{(1)}_m$ is a $\bigstar$-geometry which is root, principal, and crepant: for $m=3$ (resp.\! $m=4$) it corresponds to
the rank-$r$ MN SCFT of type $E_6$ (resp.\! $E_7$). The remark about potential finer structures for MN type $E_8$ applies also to these two MN geometries. For both $m=3,4$
 we have a second ``classical'' $\bigstar$-geometry which is root but non-principal (polarization of degree 3 and, respectively, 2). Finally we have a geometry which is neither. The last two geometries have no special resolution.

\paragraph{Coulomb dimensions $\{3,6,\dots,3(r-1), r\}$ and $\{4,8,\dots,4(r-1), r\}$.} 
The group is $G(3,3,r)$ and  respectively $G(4,4,r)$. By example
\ref{exemm} we have 5 ``classical'' geometries: one root non-principal, 3 non-root principal, and one non-root non-principal. None is crepant.

\paragraph{Coulomb dimensions $\{4,8,\dots,4(r-1), 2r\}$.}  5 ``classical'' geometries: one root and principal,
one root non-principal, 2 non-root principal, and 1 non-root and non-principal. None is crepant.

\subsection{``Quantum'' $\bigstar$-geometries: structure}\label{s:quantum}

Next we consider $\bigstar$-geometries from a more algebraic standpoint.
In particular we wish to endow them with a \emph{stacky} structure i.e.\! see them as Mumford-Deligne (MD) stacks whose coarse moduli space is their ``classical'' geometry in the sense of \S.\,\ref{s:classstr}.
The physical rationale for performing this stacky upgrading of special geometry will be clarified in \S.\,\ref{mirror}.

Let $A$ be the model fiber of a $\bigstar$-geometry and $G\subset \mathsf{Aut}(A)$
an irreducible unitary reflection group. We fix an ample $G$-equivariant line bundle $\mathscr{L}\to A$ 
in the NS class of the polarization. $G$ fixes the origin $0\in A$,
and $G$ acts on the fiber $\mathscr{L}_0$ at $0$ via a character of $G$.
For a given $A$ there are typically  a few interesting bundles and $G$-actions.
We shall present examples momentarily.
We start with the following

\begin{fact}\label{isoroot} Let $G$ be an irreducible unitary reflection group of rank $r\geq7$,
$A$ a polarized Abelian variety with $G\subset \mathsf{Aut}(A)$, and $\mathscr{L}$
a $G$-equivariant  ample invertible sheaf on $A$ which induces
a principal polarization on $A$. 
There is an isogeny
\be
\phi\colon A_\text{\rm root}\to A
\ee
with $A_\text{\rm root}$ a root variety for $G$, such that the kernel 
$L\subset A_\text{\rm root}$ is contained in the kernel $K$ of the $G$-action
\be
L\subset K\overset{\rm def}{=}\; \{a\in A_\text{\rm root}\colon ga=a\ \text{for all }g\in G\}\subset A_\text{\rm root},
\ee
while $L$ is totally isotropic for the Weil pairing in $A_\text{\rm root}$ induced
by the polarization $\phi^*\mspace{-2mu}\mathscr{L}$. 
\end{fact}
The statement ought to be true for the following reason. When $G$ is a Weyl group we reduce back to the readings of $\cn=4$ SYM. Otherwise there is a unique $G$-invariant 
Hermitian form on $\Lambda$ which, when restricted to the sublattice $\Lambda_\text{root}$,
should give the unique $G$-invariant polarization of $\Lambda_\text{root}$. The
rest of the statement follows from, say, \textbf{Proposition 13.8} of \cite{milne}. 
Thus
the natural $G$-equivariant ample sheaves $\mathscr{L}$
on a principal $G$-variety are functorially related to the
natural ones on its covering root  Abelian variety, and we 
may restrict our analysis to the latter one with no essential loss.

 \paragraph{The construction.}
 
 Let $(A,\mathscr{L})$ be a pair where $A$ is an Abelian variety and $\mathscr{L}$
 an ample line bundle on $A$, on which the unitary reflection group $G$ acts
 by automorphisms (fixing the origin of $A$). We use the symbol $\mathscr{L}$
 also to denote the pulled-back bundle over $A\times \C^r\to A$ that we call the
 \emph{Dirac sheaf}.
 Consider the algebra of $G$-invariant Dirac sections 
 \be
 \mathscr{A}\equiv\bigoplus_{k\geq0}\mathscr{A}_k\;\overset{\rm def}{=}\left(\bigoplus_{k\geq0}\Gamma(A\times\C^r,\mathscr{L}^k)\right)^{\!\!G},
 \ee
 that we see as a connected graded algebra over the chiral ring
 \be
 \mathscr{R}=\Big(\C[x_1,\dots,x_r]\Big)^{\!G}\equiv \Gamma(\mathscr{C},\co_\mathscr{C}),\quad
 \mathscr{C}\equiv\mathsf{Spec}\,\mathscr{R},
 \ee
with
\be
\mathscr{A}_k\equiv\big(\Gamma(A\times \C^r,\mathscr{L}^k)\big)^{\!G}.
\ee
Then the special geometry is the projective Abelian \emph{scheme} $\mathscr{X}$ over 
the scheme $\mathscr{C}$
\be
\mathscr{X}\equiv \mathsf{Proj}\,\mathscr{A}\to \mathsf{Spec}\,\mathscr{R}\equiv\mathscr{C}.
\ee
By construction $\mathscr{X}$ is a normal variety with only finite quotient singularities.
There is a unique canonical Deligne-Mumford stack $\mathscr{X}^\text{can}$
 whose coarse moduli is the scheme $\mathscr{X}$;  all other Deligne-Mumford stacks $\mathscr{Y}$
 with coarse moduli $\mathscr{X}$ factor through $\mathscr{X}^\text{can}$ \cite{barbara}. 
As we shall argue physically in \S.\ref{s:cartoon}, the ``quantum'' special geometry may depend on its stacky structure
which may not be the
 canonical one for the underlying scheme $\mathscr{X}$.  One has
 \be\label{jaqert}
\mathscr{Y}= [\mathbb{T}_\mathscr{L}/(\mathbb{G}_m\times G)]
 \ee
where $\mathbb{T}_{\mathscr{L}}$ is the $\mathbb{G}_m$-torsor on $A\times \mathbb{A}^r$
corresponding to $\mathscr{L}$, that is, the complement of the zero section in $\mathscr{L}$. 

The isogeneous geometry with $A$ replaced by $A/L$ may equally be viewed as
\be
\mathsf{Proj}\left(\bigoplus_{k\geq0} \Gamma(A\times\C^r,\mathscr{L}^k)^{G\ltimes L}\right),
\ee
i.e.\! as the quotient of the same covering geometry $A\times \C^r$ by the group $G\ltimes L$. 

\paragraph{The central fiber: geometry}

Let $x\in \C^r$ and
\be
G_x=\{g\in G\colon gx =x\}\subset G,
\ee 
and $\varpi\colon \C^r\to \C^r/G\equiv\mathscr{C}$.
Then the fiber over $u\in\mathscr{C}$ as a scheme is
\be
\mathscr{X}_u\simeq\mathsf{Proj}\left(\bigoplus_{k\geq0}\big(\Gamma(A,\mathscr{L}^k)\big)^{G_x}\right)\quad\text{where }\ \varpi(x)=u,
\ee
or the corresponding MD stack $[\mathbb{T}_x/(\mathbb{G}_m\times G_x)]$
where $\mathbb{T}_x$ is the restriction of $\mathbb{T}_\mathscr{L}$ to $x\in\C^r$.
In particular the central fiber is
\be\label{X0}
\mathscr{X}_0=\mathsf{Proj}\left(\bigoplus_{k\geq0}\big(\Gamma(A,\mathscr{L}^k)\big)^{G}\right) \simeq A/G
\ee
where $\simeq$ is isomorphism of normal complex spaces.

When $G$ is a Weyl group $W$, $A$ is a Looijenga $W$-variety, and the action of $W$ on the
fiber of $\mathscr{L}$ over $0$ is trivial, the geometry of the central fiber $A/W$
is described by the original Looijenga theorem.
For a general unitary reflection group $G$, and $A$ an arbitrary \emph{root} Abelian $G$-variety, the central fiber $A/G$ is given by the generalizations of the theorem first conjectured
in \cite{rus2} and recently proven by Rains in \cite{rains}.

As stated in the introduction, the Looijenga theorem played a pivotal role in the understanding of color confinement in super-Yang-Mills \cite{WittenN1}. We shall review that story in the next section,
and use the analogy with that problem to get a physical picture of what is going on at the central fiber, i.e.\!
at the most severe singularity of the $\bigstar$-geometry. Before going to that we present
an example; it may be skipped in a first reading.

\subsection{An example in rank-1}\label{s:example1} 

To put some flesh on the elusive Dirac sheaf, let us work out the simplest non-trivial example. 
We consider the rank-1 geometries with chiral ring $\C[u]$ where $\Delta(u)=6$.  
Let $E$ be the model fiber of such a geometry: it is an elliptic
curve with $\tau=e^{2\pi i/3}$. Then $G\simeq\Z_6$ and the kernel $K$ is trivial.
We take $\mathscr{L}\equiv\cl[0]$, the line bundle on $E$ with divisor $(0)$. Riemann-Roch gives
$h^0(\mathscr{L}^k)=k$. We write $w$ for a generator of $\Gamma(E,\mathscr{L})$,
$x$ for a generator of $\Gamma(E,\mathscr{L}^2)$ linearly independent of $w^2$, and
$y$ for a generator of $\Gamma(E,\mathscr{L}^3)$ not in the span of $w^3$ and $wx$.
$\Gamma(E,\mathscr{L}^4)$ resp.\! $\Gamma(E,\mathscr{L}^5)$ are generated by the monomials
in $x$, $y$, $w$ of degree 4 resp.\! 5. Between the 7 monomials of degree 6 there must be a linear relation
which (keeping into account the $\Z_6$ automorphisms) can be written
as the degree 6 curve in the weighted
projective space $\mathbb{P}(1,2,3)$
\be
y^2=x^3+w^6.
\ee
A subgroup
\be
H_{d_2,d_3,d_6}\equiv\Z_{d_2}\times \Z_{d_3}\times\Z_{d_6}\subset \Z_2\times \Z_3\times \Z_6\qquad \text{with}\quad d_s\mid s=2,3,6
\ee
 acts as
 \be
(k_2,k_3,k_6)\colon (x,y,z)\to (e^{2\pi ik_2/d_2}x,e^{2\pi ik_3/d_3}y,e^{2\pi ik_6/d_6}w).
\ee
Identifying $x/w^2=\wp$ and $y/w^3=\wp^\prime$ we see that 
$H_{d_2,d_3,d_6}$ acts on $E$ by $\Z_6$ automorphisms iff $\mathrm{lcm}(d_s)=6$.
This gives 9 equivariant action on the total space of $\mathscr{L}$ which cover the ``classical'' $\Z_6$
action on the base $E$. Then we have 9 MD stacks
\be
\Big[(\mathscr{L}\setminus 0)\times \C)\big/(\C^\times \times H_{d_2,d_3,d_4})\Big]
\ee
whose underlying coarse moduli is the classical rank-1 $\Delta=6$ special geometry.
%
%
%
%
We follow Rains.\footnote{\ See \textit{Example 3.3} and \textit{Example 3.11} in \cite{rains}. Our discussion of the rank-1 case follows his
analysis of the rank $r\geq3$ situation rather than the rank-1 one, in order to get a larger family of putative ``quantum'' geometries associated to the rank-1 $\Delta=6$ classical geometry (which is unique). The authors' feeling is that all these different ``quantum'' versions may be needed to accommodate what we know from the physical side.}
We write $S=\oplus_{n\geq0} S_n$ for the homogeneous coordinate ring 
\be
S=\C[w,x,y]/(y^2-x^3-w^6)
\ee graded as above; $S_n$ is the homogeneous component of degree $n$.
If we take $n=1$, i.e.\!
take as Dirac bundle the principal one $\cl[(0)]$, we get a non-trivial action of $\Z_n$
on the fiber over the origin. 
For all nine interesting groups $H_{d_2d_3d_6}$ the ring of invariants
\be
\big(S^{(1)}\big)^{H_{d_2,d_3,d_6}}
\ee
is a free
graded ring in two generators of respective degree $q_0$, $q_1$, see table \ref{9cases}.
This yields to $E/\Z_6$ the structure of the
Deligne-Mumford stack $\mathbb{P}(q_0,q_1)$, see discussion in .
Its coarse moduli is (of course) $\mathbb{P}^1$ for all choices of $H_{d_2,d_3,d_6}$
with $\mathrm{lcm}(d_s)=6$.

\begin{table}
\renewcommand{\arraystretch}{1.4}
$$
\begin{tabular}{c@{\hskip0.4cm}c@{\hskip0.8cm}c@{\hskip0.8cm}c}\hline\hline
$n$ & $(d_6,d_3,d_2)$ & $S(1)^{H_{d_6,d_3,d_2}}$ & weights $(q_0,q_1)$\\\hline
$6$ & $(6,3,2)$ & $\C[x^3,y^2]$ & $(6,6)$\\
& $(3,3,2)$ & $\C[w^3, x^3]$ & $(3,6)$\\
& $(2,3,2)$ & $\C[w^2, x^3]$ & $(2,6)$\\
& $(1,3,2)$ & $\C[w, x^3]$ & $(1,6)$\\\hline
$3$ & $(6,3,1)$ & $\C[y, x^3]$, & $(3,6)$\\
 & $(2,3,1)$ & $\C[y, w^2]$ & $(2,3)$\\\hline
 $2$ & $(6,1,2)$ & $\C[x, y^2]$ & $(2,6)$\\
 & $(3,1,2)$ & $\C[x, w^3]$ & $(2,3)$\\\hline
 $1$ & $(6,1,1)$ & $\C[x, y]$ & $(2,3)$\\\hline\hline
\end{tabular}
$$
\caption{\label{9cases} The 9 possible actions for rank-1 $\Delta=6$.
The first column is the minimal degree of the Dirac bundle with trivial action of $G$
on the fiber over the origin. }
\end{table}   

\section{Central fiber: physical predictions from confinement}\label{s:cartoon}

The most severe singularities are in the \emph{central fiber} over $\{0\}$.
These singularities control the singularity of the full $\bigstar$-geometry:
we argued in \S.\,\ref{s:cfiber} that the total space of the $\bigstar$-geometry admits a crepant
resolution iff the central fiber is smooth, hence a copy of $\mathbb{P}^r$. 
If we wish to understand the physical meaning of the singularities,
 we need to study directly the central fiber $\mathscr{X}_0$. Looijenga-like theorems give a nice geometric handle on $\mathscr{X}_0$, but
we need physical predictions about the central fiber $\mathscr{X}_0$
to compare with the math theorems on symplectic singularities, and establish a detailed dictionary between physical phenomena and singularities in special geometry. 
At present the math theorems give us interesting informations only when the
SCFT is described by a $\bigstar$-geometry. Most of these $\bigstar$-SCFTs are inherently strongly coupled,
with no regime in which the traditional techniques of QFT may provide any clue.
The only QFTs where we may make independent physical predictions are the weakly coupled ones.
The $\bigstar$-SCFTs with a weakly-coupled Lagrangian are:
\begin{itemize}
\item[(1)] when $G$ is a Weyl group $W(\mathfrak{g})$ we have
$\cn=4$ SYM with gauge algebra $\mathfrak{g}$ (seen as a $\cn=2$ by forgetting half of its supersymmetries); 
\item[(2)] when $G=W(C_r)$ we also have $\cn=2$ SYM with gauge group
$Sp(r)$ coupled to 4 fundamentals and one antisymmetric hypermultiplet. This SCFT
is described by a ``perfect''
$\bigstar$-geometry which is root, principal, and special crepant.
\end{itemize}
As stressed before, 
 $\cn=4$ SYM with gauge algebra $\mathfrak{g}$ is not unique:
while the local physics depends only on $\mathfrak{g}$, at the global level
there are inequivalent ``readings'' of the SCFT \cite{readings}. 
This subtlety is important, but in this section we ignore the issue. 
In the language of geometry this means that we work ``modulo isogeny'', that is,
we feel free to replace the actual geometry with an isogeneous one which is simpler
to analyze. 
  \textbf{Fact \ref{isoroot}} says that any $\bigstar$-geometry is isogeneous to
at least one \emph{root} $\bigstar$-geometry whose central fiber is described by nice
math theorems. 
For $\cn=4$ SYM with gauge group the simply-connected version of  $Sp(r)$, $E_8$, $F_4$, and $G_2$, the $\bigstar$-geometry is automatically root.
%
\smallskip
%

Our logic follows the one used by Seiberg and Witten \cite{SW1} to show that monopoles condense in $\cn=1$ SYM with $\cg=SU(2)$,
thus checking the non-perturbative mechanism responsible for color confinement.
We invert their argument: we start from the well established fact that $\cn=1$ SYM 
confines color, is gapped, and breaks chiral symmetry for all gauge groups $\cg$,
and look what this implies for $\cn=2$ QFTs and their special geometry.

\subsection{Interpolating $\cn=1$ model and color confinement}\label{s:interpola}
We consider 4d $\cn=1$ SYM, with a simple simply-connected gauge group $\cg$, coupled to three chiral supermultiplets
$X_1$, $X_2$, $X_3$, in the adjoint representation of $\cg$, with superpotential
\be\label{76zzq12}
W(X_1,X_2,X_3)=W(X_i)_0+m\big(\mathrm{tr}\,X_1^2+
\mathrm{tr}\,X_2^2+\mathrm{tr}\,X_3^2\big)
\ee
where
\be
W(X_i)_0\propto f_{abc}\,\epsilon^{ijk}X_i^a X_j^b X_k^c
\ee
 is the cubic superpotential of $\cn=4$ SYM (written with $\cn=1$ superfields).
 We write
 \be
 \tau= \frac{\theta}{2\pi}+\frac{4\pi i}{g^2}\in\mathbb{H},
 \ee
 for its gauge coupling.
As $m\to0$ we get back the $\cn=4$ theory that we see as $\cn=2$ SYM coupled to a massless adjoint 
hypermultiplet.
The model \eqref{76zzq12} with $m\neq0$ has been studied in \cite{donagi1,WWW,strongS}.
The $F$-term equations are
\be
[X_i,X_j]=m\,\epsilon_{ijk} X_k
\ee
so that the matrices $X_i\in\mathfrak{g}$ yield a representation of $\mathfrak{sl}(2,\C)$;
using the $D$-term equations and dividing by the complexified gauge group,
we conclude that the SUSY vacua  correspond to the
homomorphisms of compact Lie algebras
\be\label{homo}
\rho\colon\mathfrak{su}(2)\to\mathfrak{g}
\ee
identified modulo conjugacy \cite{donagi1}. 
The unbroken gauge group $\tilde\cg$ at each vacuum is the centralizer in $\cg$ of the image of
$SU(2)$. The light degrees of freedom are the $\cn=1$ vector supermultiplets of $\tilde\cg$.
If this group is semisimple, the light degrees of freedom are confined and the vacuum is gapped,
otherwise we have a Coulomb-phase vacuum.
The authors of ref.\cite{donagi1}
perform the analysis in two steps: first they give mass to $X_1$ and $X_2$, getting
$\cn=2^*$ SYM thus lifting away the Higgs branch,
 and then perturb this theory by adding
a small mass for $X_3$, finding the localization of the several vacua at appropriate loci
of the $\cn=2^*$ Coulomb branch.

To analyze the central fiber $\mathscr{X}_0$ we focus on the vacua which are localized 
at the origin of $\mathscr{C}$. These are the vacua which do no break the global $SO(3)$ symmetry
which rotates the $X_i$'s: they are the only vacua that don't escape at infinity in field space as $m\to\infty$.

\paragraph{Witten index.} We consider the
Witten index $\mathrm{Tr}(-1)^F$ of the $\cn=1$ model \eqref{76zzq12}.
As usual  
 we put the theory in a large periodic box $(S^1)^3$ to make the computation well-defined.
The index $\mathrm{Tr}(-1)^F$ is then independent of $m$ as long as it is not zero nor infinity. 
The $\cn=1$ chiral ring is a product of Frobenius algebras associated to each homomorphism \eqref{homo}
\be
\mathcal{R}=\prod_\rho \mathcal{R}_\rho.
\ee
Since we are interested only in the $SO(3)$-invariant vacua, we need to understand only
the factor $\mathcal{R}_0$ associated to the trivial homomorphism $\rho=0$.  
These are the vacua which survive in the strict limit $m\to\infty$,
and hence $\dim\mathcal{R}_0$ is
the index of pure $\mathcal{N}=1$ SYM with the same gauge algebra $\mathfrak{g}$.
The crucial observation \cite{WittenN1} is that this last index is fully determined  by the fact that $\cn=1$ SYM is a gapped confining theory
with spontaneous breaking of its discrete $\Z_{2h^\vee}$ chiral symmetry\footnote{\ $h^\vee$ is the dual Coxeter number of the gauge Lie algebra $\mathfrak{g}$.} down to $\Z_2$. Thus one has $h^\vee$
isolated vacua and no Goldstino; hence
\be
\mathrm{Tr}(-1)^F=h^\vee.
\ee
We follow Witten's elegant analysis \cite{WittenN1} (see also \cite{kac1,borel}); however, for our geometric purpose,
we will need to reinterpret his results in a fancier way (already pointed out by Acharya and Vafa
\cite{AV}) that we review in detail below to highlight its geometric meaning.
Note that the discrete chiral symmetry of $\cn=1$ SYM may be seen as invariance under periodic
shifts of the Yang-Mills vacuum angle
\be\label{shiftsym}
\theta\leadsto\theta+2\pi n,\qquad n\in\Z,
\ee
(cf. footnote 12 of \cite{WittenN1})
and, when so interpreted, it is clearly an exact symmetry of the model \eqref{76zzq12} for all $m$.


\paragraph{Interpretations of the index.}
In abstract terms we may think of $\mathrm{Tr}(-1)^F$ as the ``Euler characteristics $\chi(\cm)$ 
of the space $\cm$ of low-energy configurations''. Indeed, for the purpose of computing the index, we are allowed to replace the theory with \emph{any one} of its legitimate low-energy effective descriptions, 
in particular a SUSY $\sigma$-model with target $\cm$.
However the actual story is more subtle, and one may easily go astray:
to get the correct result we must be very careful with what we mean by the 
``space of low-energy configurations'' and by ``Euler characteristics''. Neither notion has its obvious meaning.

Consider a weakly-coupled SUSY theory with 4 supercharges.
In general the manifold $\mathscr{B}$ of bosonic low-energy configurations contains 
low-lying locally-propagating degrees of freedom as well as \emph{purely topological} global modes:
to get the correct space $\cm$ we must disentangle the two sets.
Luckily,\footnote{\ The geometric description in the main text is a bit imprecise. In the situation we are interested there is a surjective map $\mathscr{B}\to\cm$ with section, so $\cm$ may be seen as a subspace of $\mathscr{B}$. $\mathscr{B}$ is non-connected, and the fibers over different points may  belong to different components.} the sub-manifold $\cm\subset\mathscr{B}$ of low-lying propagating fields has a simple characterization:\footnote{\ For a more detailed analysis see chap.\! 2 of \cite{mybook}.}
in a supersymmetric theory the propagating local degrees of freedom must come in Bose-Fermi pairs,
whereas the topological modes may be realized in purely bosonic terms. Hence 
$\cm\subset\mathscr{B}$ is the integral submanifold of the distribution $\mathfrak{F}\subset T\mspace{-1.5mu}\mathscr{B}$ spanned by the SUSY variations of the bosons $\delta_\epsilon\Phi\in T\!\mathscr{B}$.  
In other words: $\cm\subset\mathscr{B}$ is a submanifold whose 
 tangent bundle $T\!\cm\subset T\!\mathscr{B}$ is spanned by the low-energy propagating fermionic modes.
So defined $\cm$ is automatically K\"ahler \cite{mybook}, as required for the target space of
a 4-supercharge $\sigma$-model, whereas the ambient bosonic space 
$\mathscr{B}$ is not even complex in general. The Witten index of the original theory
is then the index of the 4-supercharge $\sigma$-model with target $\cm$, 
which is typically coupled to a tricky topological sector locally parametrized by the normal bundle\footnote{\ This is highly oversimplified: $\mathscr{B}$ is often not smooth and the notions of distributions $\mathfrak{F}\subset T\mspace{-1.5mu}\mathscr{B}$ and normal bundles are not that innocent. These subtleties are crucial to get the right physics of confimenent. However singularities are in real codimension at least 4, and we shall worry about them in later paragraphs.} of $\cm$ in $\mathscr{B}$.

\paragraph{Refined index.}
The Witten index admits in $\cn=1$ SYM  an important refinement \cite{WittenN1}.
Morally speaking, all the cohomology groups of the low-energy effective manifold -- and not just its
Euler number -- are invariant under deformations of continuous parameters.
More properly,
we identify the refined index with the vector space $\mathbf{H}^\bullet$ of SUSY vacua
 endowed with 
the effective action $\Z_{2h^\vee(G)}/\Z_2 \curvearrowright \mathbf{H}^\bullet$
 of the chiral symmetry group, that is, with a decomposition
into $\Z_{h^\vee(G)}$ characters
\be
\mathbf{H}^\bullet=\bigoplus_{\ell\in\Z/h^\vee(G)\Z}\mathbf{H}^\ell.
\ee
The prediction of confinement is that $\mathbf{H}^\bullet$ is the regular representation of
$\Z_{h^\vee(G)}$.

\paragraph{More refinements. The chiral ring $\mathcal{R}$.} We may refine the $\cn=1$ SYM Witten index even further.
We consider the situation in which in a half-space we have the SUSY
vacuum of degree $k\bmod h^\vee$ and in the other half-space a different SUSY vacuum of
degree $k^\prime\bmod h^\vee$. The two half-spaces are connected by a BPS domain wall
\cite{AV}. The number of distinct BPS domain walls which interpolates between the two
SUSY vacua depends only on
\be
s\equiv k-k^\prime\bmod h^\vee
\ee
 by cyclic symmetry. The BPS multiplicities $\mu_s$
satisfy the reflection symmetry $\mu_{h^\vee-s}=\mu_s$. The $[h^\vee/2]$-tuple $\{\mu_s\}$
 of independent domain-wall multiplicities may be seen as a further refinement of the index.
 As explained in \cite{on}, these multiplicities endow \emph{inter alia} $\mathbf{H}^\bullet$
 with the structure of a Frobenius $\C$-algebra with a multiplication table and a Frobenius trace $\mathfrak{f}$.
 The space $\mathbf{H}^\bullet$ of SUSY vacua with these algebraic structures is the chiral ring $\mathcal{R}$
 \cite{tt*,on,cRing}.

\paragraph{Justifications and change of perspective.} Before going on with the argument, we pause a moment to
present some justification of our claims about the relations between $\mathscr{B}$
and $\cm$ for the benefit of the skeptical reader. Since the refined index $\mathbf{H}^\bullet$ is independent of the lengths
of the three circles in the $(S^1)^3$ box, we may make one of them very small,
and then work with the effective 3d theory obtained by circle compactification of the 4d IR effective theory, 
which we now quantize in a 2-torus $(S^1)^2$.
Roughly speaking the resulting 3d theory must be a $\sigma$-model:
its construction is the $\cn=1$ version of the Gaiotto-Moore-Neitzke method
to analyze the SUSY protected sector of 4d $\cn=2$ theories and prove the wall-crossing formula \cite{GMN1}. 
The main point is that a 4d vector $A_\mu$ decomposes into a 3d gauge vector $A_i$ ($i=0,1,2$)
and an adjoint \emph{periodic} scalar $A_3$. In 3d a gauge vector is dual to a second \emph{periodic} scalar
$\sigma^a$
\be\label{theduality}
dA^a=\ast d\sigma^a.
\ee
Thus $r$ 4d Abelian vectors combine in 3d in $2r$ scalars
taking value in a torus $(S^1)^{2r}$ which gets the structure of a polarized Abelian variety of dimension $r$
 from the Dirac pairing and the gauge coupling $\tau_{ij}$. The fiber $\mathscr{X}_u$
at a generic point of the Coulomb branch $u$ is the Abelian variety which arises in this way.
Thus, after performing the duality \eqref{theduality}, in a free 3d Abelian theory all the \emph{propagating} bosonic
degrees of freedom are represented by scalars which take value in a K\"ahler manifold $\ck$
(hyperK\"ahler for $\cn=2$ \cite{GMN1}). When the 4d vectors are not 
a bunch of free Abelian gauge fields, this cannot be the full story: the K\"ahler $\sigma$-model must be, in addition, coupled to a topological sector which does not propagate any additional local degree of freedom.\footnote{\ The last condition follows from equality in the numbers of \emph{local} bosonic and fermionic degrees of freedom.} When the parent 4d theory has a weakly-coupled
Lagrangian formulation, the extra topological sector of the 3d theory is explicitly known: it takes the form
of a Chern-Simons (CS) theory with a much bigger 3d gauge group $\mathscr{G}$ (which may be non-compact), and a specific CS level-matrix. 
The CS vector fields \emph{do not} propagate local modes, as it is attested by the fact that they have \emph{no Fermi partners} despite unbroken supersymmetry. 
The scalars parametrize the K\"ahler manifold $\ck$, the fermions are sections of $T\mspace{-1mu}\ck$, while the CS vectors
live in the directions of the bosonic field space $\cb$ which are normal to the propagating submanifold $\ck\subset\cb$. This is the
de Wit-Herger-Samtleben duality \cite{dWHS}. The exact Lagrangian for the 3d $\sigma$-model
coupled to Chern-Simons which arises by dimensional reduction (as contrasted to compactification!\footnote{\ Cf.\! footnote \ref{SWW3}.}) from weakly-coupled 4d $\cn=4$ SYM
can be read from \S.\,8.4 of \cite{mybook}. The exact 3d Lagrangian describing the interpolating model \eqref{76zzq12} when $R$ and $\tau$ are large
may be obtained (in principle) using those techniques. However we are \emph{not} interested in the exact theory,
but only to its extreme IR limit, i.e.\! to the low-energy manifolds $\cm$ and $\mathscr{B}$,
not to their exact siblings $\ck$ and $\cb$. Moreover we are mainly interested in 4d theories
which \emph{do not} have a weakly-coupled Lagrangian formulation, where we may figure out $\cm$
by index theory (or, in favorable circumstances, by Seiberg-Witten methods), while we have no clue about their detailed dynamics. 
\medskip

We may go one step further, and take \emph{two} of the three circles to be very small,
and work in the 2d model obtained by compactying the interpolating model on $(S^1)^2$. 
This is a  2d (2,2) $\sigma$-model with target $\cm$ which is endowed with
extra topological sectors and subtle quantum structures. Having done that, we are allowed to apply to our problem the
panoply of exact non-perturbative techniques for 2d (2,2) theories: $tt^*$ geometry \cite{tt*}, mirror symmetry \cite{hori},
$4d/2d$ correspondences \cite{4d2d}, \emph{etc.}

\medskip

We shall use the 4d, 3d, and 2d viewpoints interchangeably, according to convenience.  

\paragraph{Identifying $\cm$: IR symmetries.}
Our goal in this subsection is to determine the (abstract) IR effective target space $\cm$ for the interpolating model
\eqref{76zzq12} in the $SO(3)$-invariant vacua, with large $m$,
quantized in a finite box $(S^1)^3$. 
\emph{A priori} the refined Witten index $\mathbf{H}^\bullet$ identifies the effective target
space $\cm$ only at the level of $Q$-cohomology. 
However we have a lot more informations about $\cm$ than that.
First we know that $\cm$ has a natural K\"ahler structure. Replacing it by its normalization,
$\cm$ is a (possibly singular) normal complex analytic space.
From a physical perspective we expect other geometric aspects of $\cm$
to be robust against continuous deformations of the parameters $m$, $\tau$, and $R$: all these SUSY-protected geometric properties of $\cm$ may be reliably recovered from the physics of confinement which governs the $m\to\infty$ limit. In a 4-supercharge system the K\"ahler metric \emph{is} renormalized, while the complex structure is SUSY protected; therefore the group of holomorphic automorphisms of $\cm$
is also SUSY-protected as are its Hodge numbers $h^{p,q}(\cm)$. Assume $h^{1,1}(\cm)=1$.
Then, for large $m$ and $-i\tau$,
the \emph{compact} holomorphic automorphisms of $\cm$ play the role of IR emergent symmetries in the  effective low-energy description as a $\sigma$-model with target $\cm$. 
Indeed, thinking in terms of the effective 2d (2,2) theory,
in the extreme IR we are reduced to the $tt^*$ geometry description, which is sensitive to the complex structure but totally independent of the K\"ahler metric except for its class \cite{tt*,p1}; since $h^{1,1}(\cm)=1$, we may find a K\"ahler metric in the class which is invariant under the group of \emph{compact} holomorphic automorphisms. Hence the compact holomorphic automorphisms of the effective K\"ahler manifold $\cm$ are automatically IR emergent symmetries.
The perturbations $\delta m$, $\delta\tau$, $\delta R$ are inert under these emergent symmetries, which are then robust under their variations. In particular, if in the limit of large $m$ and $\mathrm{Im}(\tau)$
the IR effective target space is a weighted projective space $\mathbb{P}(d_0,\dots,d_r)$,
the low-energy effective target of the effective theory in the $SO(3)$-invariant vacua
must be $\mathbb{P}(d_0,\dots,d_r)$ also for small positive $m$ and $\mathrm{Im}(\tau)$. Its K\"ahler metric, however, is an intricate function of everything. 

Finally, to confirm that we have identified our IR effective target space \emph{exactly,}
we may use it to compute the domain-well spectrum and check that
it reproduces the physically correct one (as indeed it does \cite{AV}).

\subsection{Witten's computation}
For $\cn=1$ SYM quantized in a periodic box, $\mathscr{B}$ is the space
of flat $\cg$-connections on the 3-torus $(S^1)^3$ \cite{WittenN1}, that is, the space 
of commuting \emph{triples} $U_1,U_2,U_3$ in $\cg$ modulo overall conjugacy.
$\mathscr{B}$ is not connected
in general. The connected component of $\mathscr{B}$
of maximal dimension, $\mathscr{B}_\text{max}$,
has real dimension $3r$ ($r$ being the rank of $\cg$),
hence is not a complex space, let alone K\"ahler.
Taking into account Gauss' law, the subspace $\cm_1\subset\mathscr{B}_\text{max}$
whose tangent space is spanned by the low-energy Fermi modes, is (see eq.(4.13) of \cite{WittenN1})
\be\label{whatmmax}
\cm_1\simeq (\mathbf{T}_\cg\times\mathbf{T}_\cg)/W(\cg),
\ee 
where $\mathbf{T}_\cg$ is the maximal torus in the simply-connected gauge group $\cg$
and $W(\cg)$ its Weyl group. 
$\cm_1$ is the space of commuting \emph{pairs} $U_1,U_2$ in $\cg$ modulo conjugacy,
and the surjective map $\mathscr{B}_\text{max}\to\cm_1$ just forgets $U_3$.
As we shall review momentarily,
$\cm_1$ is canonically a normal projective variety over $\C$,
hence K\"ahler and compact. 
The contribution to the index $\mathrm{Tr}(-1)^F$ from the component $\mathscr{B}_\text{max}$ of $\mathscr{B}$
is then 
\be
\chi(\cm_1)_\text{top}
\ee
 where $\chi(\cdot)_\text{top}$ is the
 \emph{topological}
Euler characteristics. In general this is \emph{not} the full Witten index.
We shall present the full computation in a moment. 
The component $\mathscr{B}_\text{max}$ contributes
to the \emph{refined} index $\mathbf{H}_1$ a direct summand
(cf.\! eq.(4.13) of \cite{WittenN1})
\be  
\mathbf{H}_1^\bullet=\bigoplus_{\ell=0}^r \mathbf{H}_1^\ell
\ee
where for $\ell=0,1,\dots,r$ 
\be
\mathbf{H}_1^{\ell}=\bigoplus_{p+q=2\ell} H^{p,q}(\cm_1)\simeq \left(\bigoplus_{p+q=2\ell} H^{p,q}(\mathbf{T}_\cg\times\mathbf{T}_\cg)\right)^{\!\!W(\cg)}=\C
\ee
 since the Hodge numbers are $h^{p,q}(\cm_1)\simeq\delta_{p,q}$.
 To get the refinement of the refinement, notice that $\mathbf{H}_1^\bullet$
 is a ring and we have a trace map
 \be\label{fro}
\mathfrak{f}\colon \mathbf{H}_1^{\bullet}\to \C,\qquad \mathfrak{f}(\xi)=\int_{\cm_1}\xi\qquad \mathfrak{f}(\eta)=0\ \text{for }\eta\in\mathbf{H}_1^\ell, \ \ell\neq r,
 \ee
 of degree $-r$.
 
 However this is not yet the full story.
To get the full refined index $\mathbf{H}^\bullet$ one has to  add to $\mathbf{H}_1^{\ell}$ the contributions from the other connected components of the moduli space $\mathscr{B}$
of flat $\cg$-bundles over $(S^1)^3$ \cite{WittenN1}. However below we shall present
a fancier alternative interpretation, where taking the sum over components is \emph{not} necessary.
When one collects all contributions to the refined index, one gets $h^\vee(G)$
SUSY ground states, all bosonic, and the total refined index $\mathbf{H}^\bullet$ carries the regular representation 
of the broken symmetry $\Z_{h^\vee(G)}\simeq \Z_{2h^\vee(G)}/\Z_2$, i.e.\! of the shift symmetry \eqref{shiftsym}. We stress again that this result
(Witten's clock rule \cite{WittenN1,borel})
is required by consistency with the physics of confinement and can be deduced from it \cite{WittenN1}. 
\paragraph{The orbifold $\cm_1$.} 
The space in eq.\eqref{whatmmax} is described by
the Looijenga theorem \cite{loo} (see also \cite{kac2,jap,rus1,rus2}).
Let us recall the story to make contact with the discussion in \S\S.\,\ref{s:root} and \ref{s:strong}.
By definition $\mathbf{T}_\cg=\R^r/Q^\vee$ where $Q^\vee$ is the dual root lattice of $\cg$.
Then we may choose
a complex structure on the torus $\mathbf{T}_\cg\times\mathbf{T}_\cg$
which makes it into a Looijenga Abelian variety
\be
A_\cg=\C^r/(Q^\vee+\tau Q^\vee)\qquad \tau\in\mathbb{H}.
\ee
 $A_\cg$
is a root variety in the sense of \S.\,\ref{s:root} which is equipped with
a $W(\cg)$-invariant polarization $c_1(\mathscr{L})$ where
$\mathscr{L}$ is a $G$-equivariant ample line bundle
such that $G$ leaves invariant the fiber over $0\in A$. 
The Looijenga theorem \cite{loo} (a particular instance of the results
we shall review in \S.\,\ref{s:strong}) states that
$\cm_1\simeq A_\cg/W(\cg)$, as a normal projective variety is
 the (complex) weighted projective space \cite{dolgachev}
\be\label{uyqazqwer}
\cm_1=\mathbb{P}(d_0,d_1,\dots,d_r),
\ee
that is, it is the space of non-zero vectors in $\C^{r+1}$ modulo the equivalence relation
\be\label{whatWP}
(x_0,x_1,\dots,x_r)\sim (\lambda^{d_0}x_0,\lambda^{d_1}x_1,\dots,\lambda^{d_r}x_r),\quad \lambda\in\C^\times,
\ee 
where the weights $d_i$ are the dual Coxeter numbers of the gauge group $\cg$,
 i.e.\! the coefficients of the expansion of the maximal coroot in terms of the simple coroots
\be
\theta^\vee =\sum_i d_i\, \alpha^\vee_i,\quad\text{with}\quad \sum_i d_i=h^\vee,
\ee
where $h^\vee$ is the dual Coxeter number of $\cg$. For all $\cg$ we have 
\be\label{whatDD}
d_0=1\quad
\text{and}\quad \gcd(\hat d_0,d_1,\dots,d_r)=1,
\ee 
i.e.\! the weights $\{d_0=1,d_1,\dots,d_r\}$ are automatically well-formed in the Delorme sense \cite{dolgachev,delorme}\footnote{\ That is $\gcd(d_0,\dots,\hat d_i,\dots,d_r)=1$ for all $i$ (a hat over a weight means omission from the list).}.
Eqs.\eqref{whatWP},\eqref{whatDD} show, in particular, that $\cm_1$
 is a complex orbifold with at worst finite cyclic quotient singularities in codimension at least 2. 
   In view of our conclusions in the previous subsection, eq.\eqref{uyqazqwer} identifies $\cm_1$ as a concrete \emph{complex (algebraic) variety,}
 and not merely at the level of ``cohomology''.
 
 \medskip
 
 When $d_i=1$ for all $i$, that is, when $\cg=SU(N)$ or $Sp(n)$,
$\cm_1$ is smooth, $\mathscr{B}_\text{max}$ is the \emph{only} component of $\mathscr{B}$,
the Witten index of $\cn=1$ SYM is just the topological Euler number $\chi(\cm_1)_\text{top}=r+1$,
and the refined index is $\mathbf{H}^\ell= H^{2\ell}(\cm_1)$, with $\ell=0,\dots h^\vee-1$ in agreement with the predictions from color confinement. For the other gauge groups one has to take into account the other connected components of $\mathscr{B}$.

 \paragraph{The contribution from the other components.}
For the other groups, let $k>1$ be a non-trivial divisor of some of the $d_i$'s. The locus 
\be
\cm_k\equiv \{x_i=0\quad \text{for }k\nmid d_i\}\subset \mathbb{P}(d_0,d_1,\dots,d_r)
\ee
is a $\Z_k$ orbifold singularity of $\cm_1$ whose dimension $r(k)$ equal the number of
$d_i$'s divisible by $k$ minus 1. In facts, let $d_{i_s}$ ($s=0,1,\dots,r(k)$) be the list of the $d_i$'s
with $k\mid d_i$. One has
\be
\cm_k\simeq\mathbb{P}(d_{i_0}/k,d_{i_1}/k,\dots,d_{i_{r(k)}}/k).
\ee
We observe in passing that all $\cm_k$ that arise in this way for some
$\cg$ and $k$ are equal to $\cm_1$ for some other simple Lie group.
In particular at least one of the $d_{i_s}/k$ is equal $1$, that is, the divisor $k$
is aways one of the dual Coxeter labels $d_i$.

To each such divisor $k$ there are associated $\phi(k)$
components of the moduli $\mathscr{B}$ of flat $\cg$-bundles on $(S^1)^3$.
Here $\phi(\cdot)$ is Euler's totient function i.e.\!
$\phi(k)=|(\Z/k\Z)^\times|$.
Each one of these connected components contributes to the refined index a direct summand
\be
H^{2\bullet}(\cm_k)[\delta_{k,\sigma}]\equiv H^{2\bullet+2\delta_{k,\sigma}}(\cm_k)
\ee
where the shift in degree $\delta_{k,\sigma}\in\Z/h^{\vee}(\cg)\Z$ depends on $k$ and on the particular component
 according to
Witten's ``clock rule'' \cite{WittenN1,borel,vectorF}. The $\cn=1$ SYM refined index is then
\be\label{Wittres}
\mathbf{H}^\bullet=\bigoplus_{k\in\{d_i\}}\bigoplus_{\delta\in (\Z/k\Z)^\times} H^{2\bullet+2\delta_{k,\sigma}}(\cm_k).
\ee
 The clock rule entails that \textit{(i)} all cohomology either corresponds to Bose ground states
or to Fermi ones, so that there is no cancellation in pairs, and \textit{(ii)} $\mathbf{H}^\bullet$ is the regular representation of $\Z/h^{\vee}(\cg)\Z$.
 In conclusion \cite{WittenN1}: for $\cn=1$ SYM with simple, simply-connected gauge group $\cg$,
the Witten index is
\be\label{Wittres2}
\mathrm{Tr}(-1)^F=\dim\mathbf{H}^\bullet= \sum_{k\in\{d_i\}} \phi(k)\, \chi(\cm_k)_\text{top}\equiv h^\vee.
\ee
Here $\chi(\cdot)_\text{top}$ is the topological Euler number, that is, the signed sum of the Betti numbers
as computed by \emph{classical} cohomology. See \cite{WittenN1} for the corresponding computation when $\cg$ is not simply-connected.

In addition $\mathbf{H}^\bullet$ carries the structure of a Frobenius algebra.
It may be written in term of a single generator $X$ of degree $1\bmod h^\vee$,
which satisfies the relation $X^{h^\vee}=\text{const.}$ The Frobenius trace map 
must have a definite degree, hence degree $-r$ by eq.\eqref{fro}
i.e.
\be
\mathfrak{f}(X^\ell)=0\quad \ell\neq r\bmod h^\vee.
\ee

  \subsection{Alternative interpretations: Mirror symmetry}\label{mirror}
  We wish to reinterpret Witten's result \eqref{Wittres},\eqref{Wittres2}
  in a suggestive way. The following analysis is essentially equivalent 
  to the one in \cite{AV}, that we frame in our own language to make
  our point about quantum geometry in the $\cn=2$ context.
  
  \medskip
  
  In the paragraph about justifications (\S.\,\ref{s:interpola})
  we argued that one may compute the 4d $\cn=1$ index 
  using a
  2d (2,2) $\sigma$-model with additional topological interactions.
  From this point of view it would be unnatural to have a disconnected target space, so we propose that
  the target space of the 2d $(2,2)$ model is just the connected orbifold
  $\cm_1$ without extra components. This proposal may look to be in contradiction with
  the value of the index as fixed by confinement, but it is not so. 
The Witten index
  of a 2d $\sigma$-model is not uniquely defined by its underlying
  complex manifold. Roughly speaking, there are diverse notions of
  ``Euler number'' which compute the Witten indices
  of physically \emph{inequivalent} $(2,2)$ $\sigma$-models whose target spaces
  have the same underlying complex manifold. 
  Which ``Euler characteristics'' gives the physically correct
  index depends on the specific ``non-classical'' geometric
  structures on the target space.
  A basic example are the 
\emph{stacky} structures especially 
emphasized by the authors of \cite{gerb1,gerb2,gerb3}. ``Non-commutative''
geometries (in various senses) also play a role here.
Stated in plain English:
\begin{quote}
The Witten index of a 2d $(2,2)$ $\sigma$-model is not the Euler characteristics of the effective target $\cm$
as computed by classical cohomology, but its Euler number defined by the appropriate  
``quantum cohomology'' whose precise definition depends on non-classical geometric structures on the SUSY $\sigma$-model. 
\end{quote}
When the refined index $\mathbf{H}^\bullet$ is defined, it is
given by the full \emph{quantum} cohomology, not by the de Rham one.
The quantum cohomology is -- by definition -- sensitive to ``non-classical''
geometric structures on the target $\cm$.
The simplest instance of this phenomenon is the ``orbifold Euler characteristics''
as defined by string theory \cite{orbifold1,orbifold2,orbifold3}. When the target is a ``good'' orbifold $X/G$
 the usual topological Euler characteristic
  \be
  \chi(X/G)_\text{top}=\frac{1}{|G|}\sum_{g\in G} \chi(X^g)_\text{top},
  \ee
  ($X^g\subset X$ is the fixed locus on $g\in G$),
should be replaced by the \emph{orbifold Euler characteristic} \cite{orbifold1,orbifold2,orbifold3}
  \be\label{76zzzq123}
  \chi(X/G)_\text{orb}=\sum_{[g]} \chi\big(X^g/C(g)\big)_\text{top}
  \ee
  where the sum now is over the conjugacy classes $[g]$ of the finite group $G$ and
  $C(g)\subset G$ is the centralizer of $g$. However even in this simple situation
the story is rather subtle.
There are at least two inequivalent notions of complex analytic ``orbifolds''\footnote{\ See e.g.\! \S.\,4.4 of \cite{sasaki} for a discussion.}:
 \textit{(i)} as quotients $X/H$ by a non-freely acting group $H$, or \textit{(ii)}
 as complex spaces whose analytic structure has finite-quotient singularities.
Before applying the formula \eqref{76zzzq123} we need to understand which one is the appropriate
 notion of ``orbifold'' which applies to our particular physical problem -- 
  for a given pair $(X,G)$ it may be either one depending on
the ``non-classical'' geometric structures of the $\sigma$-model. 
To make the point obvious, we borrow an example from \cite{orbifold3}. Consider $\mathbb{P}^1$:
as analytic spaces $\mathbb{P}^1\simeq\mathbb{P}^1/\Z_n$ for all $n$:
the space $\mathbb{P}^1/\Z_n$ is not an orbifold in the analytic sense although it is a quotient by a non-freely acting group with two fixed points.
Applying the formula \eqref{76zzzq123} to the quotient we get 
\be\label{uyqazq}
\chi(\mathbb{P}^1/\Z_n)_\text{orb}=2n, 
\ee  
 which is not the \emph{topological} Euler characteristic of the underlying analytic manifold $\mathbb{P}^1$.
 Hirzebruch and H\"ofer conclude from this example that the string-theorists forgot to state the precise
 conditions under which formula \eqref{76zzzq123} is supposed to work  \cite{orbifold3}. But the string-theorists were perfectly right: indeed, there are 2d (2,2) $\sigma$-models with target space $\mathbb{P}^1$
 which realize all values \eqref{uyqazq} (and also the odd integers) as their Witten indices.
 The underlying classical manifold is the same one for all these models, but their \emph{quantum} geometry is different as we are going to review.
 
 \subsubsection{Two inequivalent orbifold structures}\label{s:two}
 Before going on with the argument, we caution the reader against a potential source of confusion.
 As emphasized above, the orbifold structure of the central fiber is \emph{not} unique and which structure is the relevant one depends on the physical question we ask.
 The space $A_\cg/W(\cg)$ may be seen as an orbifold in (at least) two quite different ways:
 as the naive quotient of the Looijenga Abelian variety $A_\cg$ by the non-free action of the
 Weyl group $W(\cg)$ or -- replacing $A_\cg/W(\cg)$ with its normalization $\mathbb{P}(d_0,d_1,\dots,d_r)$ --
 as the quotient of $\mathbb{P}^r$ by $\Z_{\mathrm{lcm}(d_i)}$. The two orbifold structures
 lead to quite different orbifold Euler numbers. 
 In this note we are interested in the second structure.
 This remark applies to the central fibers of all $\bigstar$-geometries,
 non-rigid as well as rigid.
 
To illustrate the difference between the two notions, we consider the four rank-$1$ $\bigstar$-geometries with Coulomb dimension $\Delta\equiv m\in\{2,3,4,6\}$.
 $E_m$ is an elliptic curve with complex multiplication by $\mathbb{Q}(e^{2\pi i/m})$
($E_m$ is an arbitrary elliptic curve for $m=2$). The central fiber is $E_m/\Z_m$
 and its normalization
 $\mathbb{P}^1$ is automatically smooth, so the orbifold Euler number in the second sense is just the classical one $2$.
 Instead the orbifold Euler number in the first sense is
 \be
 \chi(E_m/\Z_m)_{\!\!\!\!\!\!\text{naive}\atop\text{orbifold}}=\begin{cases} 6 \equiv e(I_0^*) & m=2\\
8 \equiv e(IV^*) & m=3\\
9 \equiv e(III^*) & m=4\\
10 \equiv e(II^*) & m=6\\
 \end{cases}
 \ee
where $e(F^*)$ is the Euler number of Kodaira's exceptional fiber
of type $F^*$ \cite{koda1,koda2} which is also
 the topological Euler characteristic of the \emph{resolved}\footnote{\ For $r=1$ the crepant resolution always exists.} special geometry
$\mathscr{X}_\text{res}$ with dimension $\Delta=m$.
This Euler characteristic yields the rank of the Grothendieck group of the
BPS category ($\equiv$ lattice of conserved charges) of the 
MN SCFT with dimension $\Delta=m$, equal to 2 plus the rank of the (maximal) flavor group.

%
%

 \subsubsection{$\cm_1$ as a (smooth) canonical Deligne-Mumford stack}\label{s:stack}
 As emphasized in \cite{gerb1,gerb2,gerb3}, to get the correct quantum geometry of a 2d (2,2)
 $\sigma$-model
 we need to keep track of the stack structure of 
 the target space. This is the only ``quantum'' geometric datum we shall
 consider in this note, and (with some abuse of language) we use the terms ``quantum'' geometry
 and ``stacky'' geometry as synonymous. We regard the low-energy effective manifold $\cm_1$
 as a Deligne-Mumford stack (DM) not a mere complex orbifold.
 We borrow some facts from \cite{barbara}
 especially their \textbf{Example 7.27}.
 
 By a \emph{stacky structure} over an algebraic variety (or scheme)
 $M$ we mean a Deligne-Mumford stack $\mathscr{M}$ together with
 a structure map $\varepsilon\colon\mathscr{M}\to M$ which makes $M$ into its coarse moduli space.
Any variety $M$ with only finite quotient singularities (in codimension at least 2)
is the coarse moduli space of an essentially unique canonical smooth Deligne-Mumford stack $\mathscr{M}^\text{can}$ which enjoys a universal property:
 for all smooth Deligne-Mumford stack $\mathscr{M}$ with  coarse moduli space $M$ and structure map 
 \be
 \varepsilon\colon \mathscr{M}\to M
 \ee 
 there is a unique morphism $\mathscr{M}\to\mathscr{M}^\text{can}$ through which $\varepsilon$ factors. Moreover the locus where 
 \be
 \epsilon\colon \mathscr{M}^\text{can}\to M
 \ee
 is an isomorphism is precisely 
$\epsilon^{-1}(M_\text{sm})$ where $M_\text{sm}$ is the smooth locus of $M$ \cite{barbara}.
These facts apply in particular to the weighted projective spaces. The corresponding 
canonical Deligne-Mumford stack (which is a toric stack in the sense of \cite{barbara}) is the global quotient stack 
\be
\mathbb{P}(w_0,w_1,\dots,w_r)^\text{can}=[(\C^{r+1}\setminus\{0\}/\C^\times]
\ee
where $\C^\times$ now acts through the
(unique) well-formed weights $\{w_0,\dots,w_r\}$ equivalent (at the level of coarse moduli) to the original weights $\{d_0,\dots,d_r\}$.

Going through the physics of color confinement in \emph{pure} $\cn=1$ SYM,
we conclude that the low-energy $\sigma$-model which computes its refined Witten index
has target space the \emph{canonical} Deligne-Mumford stack
$\mathbb{P}(d_0,\dots,d_r)^\text{can}$ of the Looijenga normal projective scheme 
$\mathbb{P}(d_0,\dots,d_r)_\text{sch}$, where the $d_i$'s are the dual Coxeter numbers of the gauge algebra $\mathfrak{g}$. 
That the identification
\be\label{stackM}
\cm_1=\mathbb{P}(d_0,d_1,\dots,d_r)^\text{can}\quad\begin{smallmatrix}\text{as smooth canonical}\\
\text{Deligne-Mumford stacks}\end{smallmatrix} 
\ee
is correct will be checked in the next paragraph using the 4d/2d correspondence and mirror symmetry.
It will also follow from the mathematical analysis in \S.\,\ref{s:strong}.
We claim that the following is true

\begin{fact}\label{mainfac} The refined and unrefined Witten indices for $\cn=1$ SYM with simple
gauge group $\cg$ are
\be
\mathbf{H}^\bullet=H^{2\bullet}(\cm_1)_\text{\rm quantum},\qquad \mathrm{Tr}(-1)^F=\chi(\cm_1)_\text{\rm quantum}.
\ee
\end{fact}
In this case the \emph{quantum} Euler characteristic is the stringy orbifold one,
while $H^{k}_\text{\rm quantum}$ is the orbifold cohomology \cite{orbcoh}. 
 A few cases of the claim may be checked easily. When $\cg=SU(N)$ or $Sp(n)$,
 $\cm_1$ is smooth -- hence isomorphic to its canonical Deligne-Mumford
 stack -- thus 
 $\chi(\cm_1)_\text{quantum}=\chi(\cm_1)_\text{top}$ and we have agreement. 
 
In some non-smooth cases where $\cm_1$ has a simple
presentation as a \emph{global} quotient of a smooth manifold,
  the computation may be done directly using the
formula \eqref{76zzzq123}.
For instance, consider the Weyl group $W(G_2)$ of $G_2$
which is isomorphic to the imprimitive Shephard-Todd reflection group $G(6,6,2)$ \cite{ST1}. 
Its reflection presentation makes it obvious that
  $W(G_2)$ is the extension of $\{\pm1\}$ by $G(3,3,2)\simeq \mathfrak{S}_3\simeq
W(A_2)$. Hence
   \be
 (\mathbf{T}_{G_2}\times \mathbf{T}_{G_2})/W(G_1)=
  \Big((\mathbf{T}_{A_2}\times \mathbf{T}_{A_2})/W(A_2)\Big)\Big/\{\pm1\}
 \ee 
 (cf.\! \S.\,3 of \cite{vectorF}). By Looijenga theorem, as an analytic space (and also as a stack)
 \be
 (\mathbf{T}_{A_2}\times \mathbf{T}_{A_2})/W(A_2)=\mathbb{P}(1,1,1)\equiv\mathbb{P}^2,
 \ee 
  is smooth, hence, in the analytic sense, $\cm(G_2)_1$ is a $\Z_2$-orbifold
  not a $G(6,6,2)$-one. The action of $\Z_2$ in homogeneous coordinates
  is $(x_1,x_2,x_3)\to(x_1,x_2,-x_3)$, the quotient is $\mathbb{P}(1,1,2)$
  with homogeneous coordinates $(x_1,x_2,x_3^2)$
  and a $\Z_2$ singularity at $(0:0:1)$.
  Its orbifold Euler character -- as a canonical stack --
  is then
  \be
  \chi(\cm(G_2)_0)_\text{an-orb}=\chi(\mathbb{P}^2/\Z_2)_\text{top}+
  \chi((0:0:1)/\Z_2)_\text{top}= 3+1\equiv 4,
  \ee
  which is indeed $h^\vee(G_2)$, the Witten index of $\cn=1$ SYM with gauge group $G_2$.
  As another simple example, consider the normal subgroup of index $2$
  \be
  W(A_3)\equiv \mathfrak{S}_4\simeq G(2,2,3)\triangleleft
  G(2,1,3)\simeq W(B_3). 
  \ee
  Then
  \be
 (\mathbf{T}_{B_3}\times \mathbf{T}_{B_3})/W(B_3)=
 \Big((\mathbf{T}_{A_3}\times \mathbf{T}_{A_3})/W(A_3)\Big)\Big/\{\pm1\}=\mathbb{P}^3/\{\pm1\},
  \ee
  where, again, $\Z_2$ acts as $(x_0,x_1,x_2,x_3)\to (x_0,x_1,x_2,-x_3)$
  so that $\cm(B_3)_\text{max}=\mathbb{P}(1,1,1,2)$ and 
    \be
  \chi(\cm(B_3)_0)_\text{an-orb}=\chi(\mathbb{P}^3/\Z_2)_\text{top}+
  \chi((0:0:0:1)/\Z_2)_\text{top}= 4+1\equiv 5
  \ee
  which is again the correct Witten index for the gauge group $\mathsf{Spin}(7)$.

  \subsubsection{Mirror symmetry}
Let us present a uniform argument to check \textbf{Fact \ref{mainfac}} for all $\cg$ (see also \S.\,\ref{s:strong}).
We use the 2d (2,2) picture where we have more QFT tools to perform the computation.
Then we are reduced to compute the quantum cohomology of a 2d $(2,2)$ $\sigma$-model
whose target is the canonical 
Deligne-Mumford stack $\mathbb{P}(d_0,d_1,\dots,d_r)^\text{can}$ with well-formed weights.
It is clear from \cite{gerb1,gerb2,gerb3} that this 2d QFT is just the Abelian (2,2) GLSM with $r+1$ chiral superfields
of charges $d_0$, $d_1$, $\cdots$, $d_r$. To reduce the computation of its quantum cohomology
to classical geometric considerations is the job of mirror symmetry.
Using the rules of \cite{hori}, after integrating out the
gauge superfield, we end up with a mirror $(2,2)$ Landau-Ginzburg (LG) model with $r$ (twisted)
chiral superfields and superpotential
 (for $d_0=1$)
\be\label{lgg}
W(X_i)=\sum_{i=1}^r e^{X_i}+t\exp\!\left(-\sum_{i=1}^r d_i X_i\right),
\ee
where the $e^{X_i}$'s are $\C^\times$-valued twisted chiral superfields, i.e.\! we 
identify periodically  
\be\label{perper}
X_i\sim X_i+2\pi i,
\ee
and $t$ is a non-zero complex parameter.
The LG model is gapped, that is, its chiral ring $\mathcal{R}$
is semisimple \cite{on}, and the Witten index $\dim\mathcal{R}$ is simply the number of critical points
of $W(X_i)$.  Let us count them: the equation $dW=0$ yields 
\be
\partial_{X_i}W(X_i)= e^{X_i}-d_i\,t\exp\!\left(-\sum_{i=1}^r d_i X_i\right)=0,
\ee
so the critical values of all variables $e^{X_i}$ are determined in terms of 
the critical values of the single function $\exp(\sum_i d_i X_i)$. Now
\be
\exp\!\left(X_i+\sum_{j=1}^r d_j X_j\right)= d_i\,t,
\ee
and hence
\be
\exp\!\left(h^\vee\sum_i d_i X_i\right)\equiv\exp\!\left((1+\sum_{j=1}^r d_j)\sum_i d_i X_i\right)=t^{h^\vee-1}\prod_{i=1}^r d_i^{d_i}\equiv C^{-h^\vee},
\ee
for a constant $C$.
Then
\be
e^{X_i}=C\,d_i\,t\,e^{2\pi i k/h^\vee}\qquad k=0,1,\dots,h^\vee-1
\ee
so, modulo periodic identifications \eqref{perper},  we have $h^\vee$
critical points, hence the Witten index of the $\sigma$-model with the stacky target \eqref{stackM} is 
\be
\mathrm{Tr}(-1)^F_\text{stacky}\equiv\dim\mathcal{R}=h^\vee,
\ee
which agrees with the index computed by Witten for the 4d $\cn=1$ SYM with the corresponding gauge group, as it was to check.
In facts more is true: as it should be, the equality in \textbf{Fact \ref{mainfac}} holds at the level of \emph{refined} indices.
Indeed, just as its parent 4d $\cn=1$ SYM, the $(2,2)$ LG theory \eqref{lgg} has an 
$R$-symmetry $\Z_{2h^\vee}$
\be
\theta_\pm \to e^{\pi i/h^\vee}\,\theta_\pm,\qquad e^{X_i}\to e^{2\pi i/h^\vee}e^{X_i},
\ee 
which is spontaneously broken down to $\Z_2$ in the gapped vacua. Thus we may define
a $\Z_{h^\vee}$-refined index for the 2d $(2,2)$ model as well. The SUSY ground states
form a single orbit under the quotient group $\Z_{2h^\vee}/\Z_2$, in fact they form the
regular representation of $\Z_{h^\vee}$. This checks that the ``clock rule'' for the refined index
 is satisfied by the quantum cohomology of the connected space $\cm_1$ canonically identified
with the chiral ring $\mathcal{R}$. It is easy to check that the Frobenius trace has degree $-r\bmod h^\vee$. This entails that the multiplicities of BPS domain walls coincide with the
number of BPS solitons for the 2d (2,2) LG model \eqref{lgg} \cite{AV}.

   \subsection{Back to $\cn=4$ SYM}\label{s:back}
   Now consider taking $m\to0$ in eq.\eqref{76zzq12}. Obviously the number of SUSY vacua will jump in the limit since lots of
   new ground states ``come in'' from infinity in field space. On the other hand,
   the SUSY ground states of the $m\neq0$ theory will not disappear as $m\to0$:
   they will be simply joined by many others. In other words: $\cn=4$ SYM has a special subset
   of SUSY vacua inherited from the interpolating model \eqref{76zzq12} as $m\to0$.
   These special vacua are analogue to the two points in the Coulomb branch of pure $\cn=2$
   $SU(2)$ SYM where the monopole and the dyon are massless: these two $\cn=2$ SUSY vacua
   are not lifted by the mass
   perturbation which triggers the flow to $\cn=1$ SYM in the original analysis by Seiberg and Witten \cite{SW1}. For our present purposes we are interested only in the special $\cn=4$ vacua which arise from the $SO(3)$-invariant vacua of the deformed model \eqref{76zzq12}.
These vacua preserve
   both $SU(2)_R$ and $U(1)_R$, i.e.\! they all sit at the origin of the $\cn=4$ Coulomb branch
   on top of each other.
   Conversely, $\cn=4$ ground states which preserve these symmetries will not be ``pushed
   away from the origin'' when we switch on an infinitesimal $m$. 
   As in ref.\cite{SW1} we may use the non-renormalization theorems, valid for the $\cn=1$
   interpolating model \eqref{76zzq12}, to control
   that these assertions remain true at the non-perturbative level even for finite non-zero $m$.
   
  We know that for $m$ arbitrary small but non-zero the local IR physics at the $SO(3)$-invariant vacua in a 3d box with two very small circles
  is effectively described by the 2d (2,2) $\sigma$-model with quantum target the stack
  $\mathbb{P}(d_0,\dots,d_r)^\text{can}$, that is, by the mirror LG model \eqref{lgg}. We have been cavalier with the global structure of the
  gauge theory, and we broke $S$-duality by hand, so this conclusion holds only modulo finite covers,
  that is, modulo isogeny. 
In these situations, when restricted to vacua at the origin of $\mathscr{C}$, the topologically twisted $\cn=4$
  SYM is described over the origin of the Coulomb branch by the topological twist of this stacky $\sigma$-model
  which then may be analyzed using $tt^*$ geometry. This leads to
  the identification
  \be\label{piuuy}
   \cm_1=\cz_0
   \ee
   where $\cz_0$ is the fiber of the fibration \eqref{789zza} over the origin.
 Both sides of this equality are
 geometric objects which may be non-classical (``stacky'', ``non-commutative'', etc.).
The normalization $\mathscr{X}_0$
of the fiber over $\{0\}$
of the ``usual'' (i.e.\! classical) special geometry $\mathscr{X}\to\mathscr{C}$
is a normal complex (K\"ahler) 
variety of dimension $r$ which is naturally identified with the classical analytic space
underlying $\cm_1$. 
Geometrically speaking, the identification of
the classical geometric objects underlying the two sides of eq.\eqref{piuuy}
is tautological: by their very construction they are both equal to the normalization of $A_\cg/W(\cg)$. 
We conclude
 \begin{php} For $\mathcal{N}=4$ SYM with simple gauge group $\cg$, after replacing
the special geometry with an isogeneous one (when necessary)
the normalization $\mathscr{X}_0$ of the central fiber over $0\in\mathscr{C}$ of the ``usual'' ($\equiv$ classical)
special geometry
is the weighted projective space
\be\label{exexexp}
\mathscr{X}_0=(\cm)_{\text{\rm underlying}\atop\text{\rm variety}}=\mathbb{P} (d_0,d_1,\dots,d_n),
\ee 
where the $d_i$'s are the dual Coxeter numbers of the gauge group $\cg$.
 In particular $\mathscr{X}_0$ is smooth iff $\cg=SU(N)$ or $Sp(n)$. Therefore the symplectic singularities
 $\C^r/W(\mathfrak{g})$ are crepant iff $\mathfrak{g}=A_r$ or $C_r$.
\end{php}

We stress again that  what matters here is not the underlying classical geometry but the quantum one.
As we shall see in the next section, the precise stack depends on a choice of
Dirac sheaf $\mathscr{L}$, which we think of as a refinement of the reading.

\medskip

%

We shall refer to the number $\dim\mathcal{R}_0$ of SUSY vacua which are not lifted when switching on a non-zero $m$
as the \emph{central Witten index} of the $\cn=2$ model. It is the quantum Euler number of the fiber over the origin of $\mathscr{C}$ computed in the appropriate quantum cohomology.
More physically it is the dimension of the chiral ring $\mathcal{R}$ of the $SO(3)$-invariant vacua.

 \paragraph{Solving a little paradox.}
 What does the central Witten index $\dim\mathcal{R}_0$ count?
 Naively one would say that it counts the vacua at the origin of the Coulomb branch.
 Since these are the vacua that don't break the superconformal symmetry, one would suspect
 that it counts the superconformal vacua. But this naive idea is inconsistent with the conformal
 operator-to-state correspondence which says that there is \emph{only one} superconformal vacuum, the state 
 associated to the identity operator.   
   
But there is no paradox. 
$\dim\mathcal{R}_0$ counts $SO(3)$-invariant vacua in a periodic box $(S^1)^3$ not in infinite volume: superconformal symmetry is broken by the boundary conditions.
Making two sides of the box very small and one large, we reduce to an effective 2d (4,4) model that
 we think of as a $(2,2)$ theory. Our $\dim\mathcal{R}$ states
are now \emph{Ramond} vacua in a small periodic circle; 
when quantized
in infinite volume this 2d model has an unique NS vacuum which is the spectral flow of 
the Ramond vacuum $|1\rangle\in H^\bullet(\mathscr{X}_0)_\text{quantum}$
associated to $1$ under the isomorphism $H^\bullet(\mathscr{X}_0)_\text{quantum}\simeq \mathcal{R}_0$ \cite{cRing}.
$1$ is the unique unipotent of $\mathcal{R}_0$ invariant under the $\Z_{h^\vee}$
symmetry.
This conclusion applies to our 4d theories: the only vacuum in infinite volume which preserves
superconformal symmetry arises from the identity. 
%
%
%
%

\subsection{Central fiber of a \emph{rigid} $\bigstar$-geometry}

The above discussion was in the context of $\bigstar$-SCFT with a weakly-coupled
Lagrangian formulation where we may use standard QFT arguments to get
a physical answer to match with the geometric theorems.
The method of perturbing the SCFT by a marginal deformation which preserve
an $\cn=1$ SUSY subalgebra and lifts almost all the Coulomb-phase
vacua, cannot be extended to the rigid $\bigstar$-SCFT not just because
they are inherently strongly-coupled, but -- more seriously -- because they have
no operator doing the job. The models have relevant perturbations which do not preserve any SUSY,
but in this case we have no efficient tool to study the deformed theories non-perturbatively.

In a few situation we may argue as follows. Some $\bigstar$-SCFT may be
thought of as the result of gauging a finite subgroup
of the $S$-duality group of a $\bigstar$-SCFT with a Lagrangian formulation \cite{M5}. 
The Lagrangian SCFT is then strongly coupled since $\tau$ should be a fixed
point of a subgroup of $SL(2,\Z)$ but the non-perturbative analysis above is valid even there.
It is natural to construct the global chiral ring for the original SCFT as an orbifold
of the one for the Lagrangian model. More or less by definition, this will give
\be
\mathcal{R} = H^\bullet(\mathscr{X})_\text{quantum}.
\ee 
However the arguments of \S.\,\ref{s:strong}
suggest that, while the underlying complex manifold
of the two models are related in the obvious way,
their quantum cohomology may be slightly different. This may be somehow natural, because
the discrete gauging of the duality group should be supplemented by
 an additional twist by the $U(1)_R$ symmetry  to preserve SUSY \cite{M5}.
 \medskip

While on the physical side we can say only little,
on the geometrical side the generalization of the results of this section to all $\bigstar$-SCFT
is straightforward:
all the math results we used to understand the central fiber in the Lagrangian case have natural and elegant
generalizations to rigid $\bigstar$-geometries. We have only to apply
the Looijenga-type general theorems to get the natural candidate central chiral rings
$\mathcal{R}$ for all $\bigstar$-geometry. 
This will be our task in the section \ref{s:strong}.
As it will be clear there,
here we have some freedom: as in \S.\,\ref{mirror} 
$\mathcal{R}$ depends on the quantum ($\equiv$ stacky) geometry (which is sensitive to the
particular Dirac sheaf), while the
underlying classical special geometry gives us only its coarse moduli space.
The canonical stack structure yields the minimal chiral ring which is a universal subalgebra
of all possible chiral rings
\be
\mathcal{R}^\text{can}\to \mathcal{R}.
\ee
One has to understand how to deal with the embarrassment of riches
induced by the apparent freedom in the choice of the Dirac bundle.

\section{Crystallographic \& strong reflection actions}\label{s:strong}

We have seen above that the intrinsic singularities of a $\bigstar$-geometry $\mathscr{X}$ are controlled
by its central fiber $\mathscr{X}_0$, and argued that -- as long as we are interested only
in the most severe singularities -- we may replace the geometry with an isogeneous one.
The RG scenario in \S.\,\ref{s:root} makes
 us to expect that the less singular geometry in an isogeny class is attained by
 a root geometry, that is, by a simply-connected one.
 
 \medskip  

We mentioned that when $G$ is a Weyl group $W$ there is a one-to-one correspondence
between the root geometries for $G$ and the affine Dynkin graphs with the property that
when we delete the leftmost node we remain with a Dynkin graph for a Lie algebra with
Weyl group $W$. We have still to explain how this comes about and then extend the
analysis to the general unitary reflection groups $G$. 
Moreover from the physical scenario in \S.\,\ref{s:cartoon} we have learned that
the quantum (or, more precisely \emph{stacky}) structures of the geometry
may be relevant. Hence we pose the 

\begin{que} Given a pair $(A_\text{\rm root},G)$ how many \emph{natural} stacky structures
are there (besides the canonical one)?
\end{que}
For instance in the example of \S.\,\ref{s:example1} we found 9 interesting
stacky geometries for the rank-$1$ geometry with $\Delta=6$, i.e.\! for
$(E_{e^{2\pi i/3}},\Z_6)$ in addition to the canonical one which yields the stack $\mathbb{P}(1,1)$.
Of course the physically relevant notion of ``natural'' should  be understood better,
but some quantum structures are particularly elegant and suggest themselves as \emph{natural.}

\medskip

Understanding the central fiber for general $\bigstar$-geometries
requires a generalization of the Looijenga theorem from quotients of 
Abelian varieties by Weyl groups to quotients by more general
discrete groups. The Looijenga theorem has various generalization that we now review.
They were conjectured by Bernstein and Schwarzman \cite{rus1,rus2} and 
have being  proven recently by Rains \cite{rains}.
To state the basic results we need some definitions.

\medskip  

An automorphism $r\colon X\to X$ of a connected complex variety
$X$ is a \emph{reflection} iff the locus of its fixed points has codimension 1.
A discrete group of automorphisms $S\subset \mathsf{Aut}\,X$ is a \emph{reflection group}
iff it is generated by reflections. 

For our applications we are interested to reflection groups of some special kinds.
A \emph{unitary reflection group} is a discrete subgroup
$G\subset U(r)$ whose action on $\C^r$ is generated by reflections (such a group is automatically finite).
These groups were classified by Shephard and Todd \cite{ST1,ST2}; we used them in \S.\,\ref{s:star} to construct $\bigstar$-geometries.

\smallskip

Let $R$ be a discrete group of \emph{affine} automorphisms of $\C^r$ of the form
\be\label{affGG}
0\to \Lambda\to R\to G\to 1,
\ee 
where $G\subset U(r)$ acts linearly and $\Lambda\subset\C^r$ is a $G$-invariant lattice of translations. 
$R$ is \emph{irreducible} if its linear part $G$ acts irreducibly on $\C^r$, and it  is \emph{split} iff $R= G\ltimes\Lambda$, i.e.\! iff there is a point $x\in\C^r$ whose isotropy group $R_x\equiv G$.
The group $R$ is \emph{crystallographic}
iff $\C^r/R$ is compact, i.e.\! iff $\Lambda$ is a $G$-invariant
 full lattice in $\C^r$. In this case $A=\C^r/\Lambda$ is an Abelian variety\footnote{\ The complex torus $\C^r/\Lambda$ is polarizable. See footnote \ref{fooot}.}
 with $G\subset \mathsf{Aut}(A)$. The underlying space of the central fiber in
  a $\bigstar$-geometry has the form
  \be
  \C^r/R\equiv A/G\quad\text{where}\quad A\equiv \C^r/\Lambda
  \ee
 with $R$ an irreducible split crystallographic group whose linear part $G$ is a unitary reflection group. 
The Abelian variety $A$ is the model fiber and the group $G$ is defined over the ring of integers in either $\mathbb{Q}$ or in an imaginary quadratic field. 
The real finite reflection groups are the Weyl groups of the simple Lie algebras. 
The finite complex reflection groups whose field of definition $\mathbb{K}_G$ is imaginary quadratic are
listed in table \ref{tab}.

\medskip

We need two stronger notions. 
The discrete group $R$ in eq.\eqref{affGG} is an \emph{(affine) crystallographic reflection group} iff it is crystallographic and generated
by affine reflections, that is, by affine automorphisms $\C^r\to\C^r$ fixing a hyperplane. Its linear part $G$ is then a unitary
reflection group. 
When $G$ is not one of the groups 
\be
G(4,2,r),\quad G_{12}\quad \text{or}\quad G_{31},
\ee
a crystallographic reflection group is automatically split.\footnote{\ See \textbf{Theorem 2.8.1} in \cite{popov}.}
 The pairs $(G,\Lambda)$ such that $G\ltimes\Lambda$ is an irreducible affine crystallographic
 reflection group are listed in tables 2, 3 of \cite{popov}: the groups with an asterisk 
 should be neglected because they are not split. 
For $r=2$, see also table II in \cite{jap}.\footnote{\ The entry with $G=G_{12}$ is missing from  that table.
For details on this special case, see \cite{bolza}.} 
We stress that for each Shephard-Todd group $G$ in table \ref{tab}, there is at least one 
split crystallographic
 reflection group $R$ which has $G$ as its linear part: the translation subgroup $\Lambda\subset R$ is
 a $G$-invariant full lattice. Then $A\equiv \C^r/\Lambda$ is an Abelian variety with a group $G$
 of automorphisms endowed with an invariant polarization given by a multiple of the unique
 Hermitian form\footnote{\ \label{fooot}Up to an overall factor, the skew-symmetric form on $\Lambda$ obtained by restricting the
 imaginary part of the $G$-invariant Hermitian form $\langle-,-\rangle$ is integer-valued. Since 
 $\langle\gamma_1,\gamma_2\rangle\in\mathbb{K}_G$ for $\gamma_1,\gamma_2\in\Lambda$, 
 this is obvious when $\mathbb{K}_G=\mathbb{Q}$.
 When $\mathbb{K}_G=\mathbb{Q}(\sqrt{-d})$,
 $(\langle\gamma_1,\gamma_2\rangle-\langle\gamma_2,\gamma_1\rangle)/\sqrt{-d}$ is
 invariant under complex conjugation, hence an element of $\mathbb{Q}$.} preserved by the irreducible action of $G$. Other $G$-invariant Abelian varieties are obtained
 by considering the isogeny $A\to A^*$ and taking the quotient by a subgroup of its kernel. 
 When $G$ is a Weyl group the invariant lattice form one-parameter families parametrized by $\tau$, otherwise the $G$-invariant lattices form a finite set. 
 
 \begin{thm}[\!\cite{popov}] Let $G$ be a rank-$r$ irreducible unitary reflection group
 (which is either a Weyl group or in table \ref{tab}) and $\Lambda\subset\C^r$ a $G$-invariant full lattice. The affine group $G\ltimes \Lambda$ is a crystallographic reflection group if and only if $\Lambda$ is a \emph{root lattice} for $G$. 
 \end{thm}
In other words, the crystallographic reflection groups precisely correspond to the $\bigstar$-geometries 
with the simplest RG scenario described in \S.\,\ref{s:root}, that is, to the simply-connected ones.
 From the list of
root lattices we read the classification of the split  crystallographic reflection groups. 

\paragraph{$G$ a Weyl group.}
We already stated that the crystallographic reflection groups with linear part $G$ a Weyl group $W$ are
in correspondence to the affine Lie algebras whose Dynkin graphs, drawn as in tables I-II-III of \cite{kac2}, have the property that deleting their leftmost node we get the graph of a Lie algebra with Weyl group $G$ \cite{kac2,rus1}. See ref.\cite{rus1} for the explicit construction of $\Lambda_\text{root}$ from the affine Lie algebra. The groups $W\ltimes \Lambda$ featuring in the Looijenga theorem
are the special instance of these 
crystallographic reflection groups where the affine Lie algebra is the \emph{dual} of the (untwisted)
affine extension of the gauge Lie algebra of $\cn=4$ SYM \cite{loo}. For the Weyl groups of non-simply-laced Lie algebras we get more
than one crystallographic reflection group which correspond to different $W$-invariant root lattices.
For the Weyl group of $C_r$ ($r\geq4$) we have five such algebras, see figure \ref{weylC}.
However the two reflection groups associated to $B_r^{(1)}$ and $C_r^{(1)}$
 are equivalent (in particular they have the same Dynkin labels) \cite{rus1,popov}
and we have just four crystallographic reflection groups with linear part $W(C_r)$.
The Looijenga ones are $A^{(2)}_{2r-1}$ and $D^{(2)}_{r+1}$ respectively for the gauge algebra $B_r$ and $C_r$.
The Weyl invariant lattices which lead to a reflection group have a non-principal polarization in general.    
 The degrees of the Looijenga polarizations are written in the upper part of table \ref{degdeg}.
 For the non-simply-laced Lie group there are other crystallographic reflection groups with principal
 polarization as we saw when classifying root lattices in \S.\ref{s:affaff}. They are associated to
 the other affine Lie algebras, see the lower part of table \ref{degdeg}.
The relevance of the crystallographic reflection groups is that their Abelian varieties $=\C^r/\Lambda_\text{root}$ have simple quotients by $G$, namely the weighted projective varieties 
 \be
A_\text{root}/G=\C^r/(G\ltimes \Lambda_\text{root})\simeq \mathbb{P}(d_0,d_1,\dots,d_r).
\ee
 The Looijenga theorem, was generalized to all affine Lie algebras in refs.\cite{kac2,rus2}.
The corresponding Coxeter labels $d_i$ may be read in tables I-II-III of \cite{kac2}.
 
\begin{table}
\caption{\label{degdeg}Polarization degrees of root varieties $A_\text{root}$ with $G=W(\mathfrak{g})$}
$$
\begin{tabular}{r|ccccccccc}\hline\hline
$\mathfrak{g}$ & $A_r$ & $B_r$ & $C_r$ & $D_r$ & $E_6$ & $E_7$ & $E_8$ & $F_4$ & $G_2$\\
degree & $r+1$ & $4$ & $1$ & $4$ & $3$ & $2$ & $1$ & $4$ & $3$\\
affine algebra& $A^{(1)}_r$ & $A^{(2)}_{2r-1}$ & $D^{(2)}_{r+1}$ & $D^{(1)}_r$ & $E_6^{(1)}$ &
$E_7^{(1)}$ & $E_8^{(1)}$ & $E_6^{(2)}$ & $D_4^{(3)}$
\\\hline\hline
degree  & & $2$ & $1$ & & & & &$1$ & $1$\\
affine algebra
& &  $B_r^{(1)}, C^{(1)}_r$ & $A^{(2)}_{2r}$ & & & & &$F_4^{(1)}$ & $G_2^{(1)}$\\\hline\hline
\end{tabular} 
$$ 
\end{table}

 \paragraph{Two generalizations.}
 The authors of \cite{rus2} made two more general conjectures. The first one extends the statement
 to crystallographic reflection groups which are not associated to Weyl groups i.e. whose linear parts are the unitary reflection groups in table \ref{tab}:

 \begin{fact}\label{1stgen} Let $G\subset U(r)$ a finite group and $\Lambda\subset \C^r$ a $G$-invariant full-lattice.
 Then the quotient $\C^r/(G\ltimes\Lambda)$ is a weighted projective space if and only if
 $G\ltimes \Lambda$ is a crystallographic reflection group, that is, if the $G$-lattice $\Lambda$ is \emph{root}.
 \end{fact}
 This was recently proven by Rains \cite{rains} who lists the weights $d_i$ for all crystallographic reflection groups (including the non-split ones). In the same paper he also proves the 
 second generalization conjectured in \cite{rus2}. To state it we need the notion of
 \emph{strong crystallographic action}. \textbf{Fact \ref{1stgen}} may be stated in the following way:
 Let $\mathscr{L}$ be a $G$-equivariant ample line bundle over $A\equiv\C^r/\Lambda$ whose Chern class is a multiple of
 the polarization and such that $G$ acts trivially on its fiber over the origin. Then
 the graded algebra
 \be\label{homalg}
 \bigoplus_{k\geq0}\Gamma(A,\mathscr{L}^k)^G
 \ee  
 is a polynomial algebra iff $G\ltimes\Lambda$ is a (crystallographic) reflection group.
 We may consider the more general case where $G$ acts by automorphisms
 of the total space $\mathscr{L}$ of the line bundle.
 We write $\mathbb{T}_\mathscr{L}$ for the complement of the origin in the cone
 $\mathsf{Spec}\bigoplus_{k\geq0}\Gamma(A,\mathscr{L}^k)$ and $\tilde{\mathbb{T}}_{\mathscr{L}}\simeq \C^{r+1}$ for 
 its universal cover via the quotient map $\C^r\to \C^r/\Lambda\simeq A$.
 One says that $G$ is \emph{strongly crystallographic} for $(A,\mathscr{L})$ iff the group
 $\pi_1(\mathbb{T}_\mathscr{L}).G$ is a (non-linear!) reflection group on $\tilde{\mathbb{T}}_{\mathscr{L}}$ \cite{rains}.

 \begin{thm}[Rains \cite{rains}] Let $G\curvearrowright(A,\mathscr{L})$
($\mathscr{L}$ ample) be an irreducible strongly crystallographic reflection action. Then $\bigoplus_{k\geq0}\Gamma(A,\mathscr{L})^G$ is a free polynomial ring.
 \end{thm}
 Under the assumptions of the theorem we have 
a morphism of stacks
 \be\label{mstack}
 [A/G]\simeq [\mathbb{T}_\mathscr{L}/\mathbb{G}_m\times G]\to 
 [(\mathbb{T}_\mathscr{L}/G)/\mathbb{G}_m]\simeq\mathbb{P}(w_0,\dots,w_r)_\text{stack}
 \ee
 where $\mathbb{G}_m$ is the multiplicative (algebraic) group, and $\mathbb{P}(w_0,\dots,w_r)_\text{stack}$
 is a (non-canonical\footnote{\ In other words the weights $\{w_0,\dots,w_r\}$ need not to be well-formed in the Delorme sense \cite{delorme,dolgachev}.}) weighted projective stack. The morphism induces an isomorphism at the level
 of coarse moduli spaces i.e.\! at the level of ``classical'' complex geometry.
 
 \begin{rem} The structure morphism from the stack $[A/G]$ to its coarse moduli factors through the morphism of stacks \eqref{mstack} so the weighted projective stack
 $\mathbb{P}(w_0,\dots,w_r)_\text{stack}$ is a refinement of the coarse moduli space (the underlying ``classical'' geometry) which is the weighted projective space with well-formed weights.
 \end{rem}

 The possible $G$-invariant bundles and strong crystallographic 
 actions are listed in the tables of ref.\cite{rains}, together with the
 weights $\{w_i\}$ of the generators of the polynomial ring \eqref{homalg} for all such actions and 
 the degree of the minimal invariant polarization.  For full details 
 we refer the reader to that reference. 
 
 \begin{rem} Ref.\cite{rains} lists a further invariant, namely the ramification divisor $D$ of
 \be
 A\to A/G.
 \ee 
 Two examples with the same weighted projective stack may have different
 divisors $D_1$, $D_2$ with the same Chern class $c_1(D_1)=c_1(D_2)$. When $G$ is the Weyl group of
 a simply-laced Lie algebra $\mathfrak{g}\in ADE$
  the Chern class of $D$ is the Coxeter number times the minimal polarization.
 For the Weyl groups of $F_4$ and $G_2$ we have two possibilities corresponding to the Dynkin graphs of the affine Lie algebra and its dual; $c_1(D)$ is, respectively, the Coxeter number and the dual Coxter number times the minimal polarization. For the Weyl group of $C_r$ we have four possibilities corresponding, respectively, to the Dynkin graphs $A^{(2)}_{2r-1}$, $B^{(1)}_r$ (equivalently $C^{(1)}_r$), $A^{(2)}_{2r}$ and $D^{(2)}_{r+1}$. $c_1(D)$ is the Coxeter number of the associated graph. 
 \end{rem}

 \subsection{Strong ``quantum'' $\bigstar$-geometries}
  We discuss some implication of the above theorem
 and the related tables \cite{rains} for special geometry.

\paragraph{Classical root $\bigstar$ geometries}
The \emph{classically} distinct \emph{root} $\bigstar$-geometries
may be
read from the tables of \cite{popov}. For $G=W(\mathfrak{g})$ they correspond to
 the affine Lie algebras (table \ref{degdeg}) with the proviso that $B_r^{(1)}$ and
 $C_r^{(1)}$ define the same classical geometry. 
The root $\bigstar$-geometries for $G=G(m,q,r)$ ($m\in\{3,4,6\}$, $q\mid m$, $r\geq3$) where 
listed in \S.\,\ref{s:examples}:  
\be
\begin{tabular}{l@{\hskip20pt}l@{\hskip20pt}l}\hline\hline
$G$ & $\#$ root geometries & degrees of polarization\\\hline
$G(3,q,r)$ & $2$ & $1$, $3$\\
$G(4,q,r)$ & $2$ & $1$, $2$\\
$G(6,q,r)$ & $1$ & $1$\\\hline\hline
\end{tabular}
\ee
The central fiber of these root geometries has the form
\be
\C^r/(G\ltimes \Lambda_\text{root})\equiv \C^r/R
\ee
with $R$ a crystallographic reflection group in the sense of \cite{popov}. 
The non-root $\bigstar$-geometries were classified in \S.\,\ref{s:examples}.
They are isogeneous to at least one root geometry.

However to a ``classical'' special geometry there may correspond several inequivalent ``quantum'' ones.

 \paragraph{Strong $\bigstar$-geometries.}
 The ``best'' classical $\bigstar$-geometries in a isogeny class are the root ones.
 The corresponding notion of ``best'' \emph{quantum} $\bigstar$-geometry in a class
 are the \emph{strong} $\bigstar$-geometries.
 They are related to the strong crystallographic actions as
 the root $\bigstar$-geometries are related to crystallographic reflection groups. By definition
 \emph{the classical $\bigstar$-geometry which underlies a strong $\bigstar$-geometry
 is a root one.}
 
A ``strong'' geometry may be non-principal;  
 the degree $d$ of its polarization may be read in the tables of \cite{rains}.
 We call $d$ the \emph{degree} of the strong crystallographic action. When $d=1$ we also say that the strong geometry is principal. 
 
 \begin{fact}Going through the tables of ref.{\rm\cite{rains}} we observe:
 \begin{itemize}
 \item[\rm(a)] For each $G$ there is at least one 
 strong crystallographic action. Hence all $\bigstar$-geometries are classically isogeneous to a strong one.
 The statement holds also for the sporadic groups $G$;
 \item[\rm(b)] When $G$ is a Weyl group there is a single strong action per crystallographic reflection group;
 \item[\rm(c)] For the groups $G(3,1,r)$, $G(4,1,r)$, $G(4,2,r)$, $G(6,q,r)$ all ``classical''
 $\bigstar$-geometries are classically isomorphic
 (not just isogeneous) to some strong $\bigstar$-geometry;
 \item[\rm(d)]  Except for $W(A_r)$ and $G(m,m,r)$ there are more than one strong crystallographic action for a given $G$. When $G$ has more than one root lattice with polarizations of diverse degrees $d$,
 for each pair $(G,d)$ there may be more than one strong action.
 The number of strong crystallographic actions with \emph{principal} polarization ($d=1$) are ($r\geq3$)
\begin{small}\be\label{uyyyy2}
\hskip-1cm\begin{tabular}{c|ccccccc}\hline\hline
$G$&$G(2,1,r)$ & $G(3,1,r)$ & $G(4,1,r)$ & $G(4,2,r)$ & $G(6,1,r)$ & $G(6,2,r)$ & $G(6,3,r)$\\\hline
$\#$ &$2$ & $2$ & $3$ &$2$ & $9$ & $3$ & $3$\\\hline\hline
\end{tabular}
\ee
\end{small}
\item[\rm(e)] 
The distinct strong actions with a given $(G,d)$ are
 distinguished by
 the weights $(w_0,w_1,\dots,w_r)$ of the associated weighted projective \emph{stalk}
 and by the ramification divisor $D$.
 \end{itemize}
 \end{fact}
 
 \begin{exe} The two \emph{principal} strong actions for $G(2,1,r)\equiv W(C_r)$ correspond to the two
 affine Dynkin graphs $D^{(2)}_{r+1}$ and $A^{(2)}_{2r}$. 
 \end{exe}

\begin{exe}
 For $G(6,1,r)$ ($r\geq3$)
 there are 9 strong crystallographic actions, all \emph{principal},
  in correspondence to the nine entries in table \ref{9cases};
they lead to 9 \emph{principal} strong $\bigstar$-geometries with stacky weights ($r\geq3$)
 \be\label{7611z}
 (w_0,w_1,w_2,\dots,w_r)=(q_0,q_1,6,\dots,6)
 \ee
 where the 9 pairs of integers $(q_0,q_1)$ are listed in the last column of table \ref{9cases} 
 (actions with equal weights are distinguished by the divisor $D$).
 \end{exe} 
 Going through the tables of \cite{rains} one notices:
 \begin{fact} Let $\mathbb{P}(w_0,\dots,w_r)_\text{\rm stack}$
 and $\mathbb{P}(\tilde w_0,\dots,\tilde w_r)_\text{\rm stack}$ be the weighted projective stacks
 of two strong crystallographic actions which share the same irreducible
 unitary reflection group $G$ (of rank $\geq3$) and the same degree $d$. The two stacks have
isomorphic
 coarse moduli space (given by isomorphic weighted projective space
 $\mathbb{P}(d_0,\dots,d_r)$ with well-formed weights). That is:
 \emph{all strong $\bigstar$-geometry with a fixed group $G$ and polarization degree $d$
 have the same underlying ``classical'' root geometry.}
 However the unique root geometry may correspond to several inequivalent ``quantum'' (i.e.\! stacky) strong geometries.
 \end{fact}
 
In view of the classical construction of the $\bigstar$-geometry in \S.\,\ref{s:classstr} (which depends only on $G$ and $\Lambda$) this \textbf{Fact} is a tautology. However it gives a highly non-trivial consistency check on the full story. 
 
 \begin{exe} 
 The coarse moduli space of all nine weighted projective stacks in eq.\eqref{7611z} are all
 isomorphic (as complex schemes) to $\mathbb{P}^r$. In particular they are all smooth as complex manifolds (but not, in general, as stacks).
 Thus the central fiber of the (unresolved) classical $\bigstar$-geometry with
 $G=G(6,1,r)$ is smooth. 
 \end{exe}
 
 \begin{exe} More generally for all principal strong actions with unitary group
 $G=G(p,1,r)$ ($p=2,3,4,6$) the underlying ``classical'' central fiber is a copy $\mathbb{P}^r$
 hence smooth. Indeed these are the ``classical'' versions of the Seiberg-Witten geometry
 of the higher rank Minhan-Nemeshaski SCFTs.  
 \end{exe}
 
 \begin{fact}\label{xxxyr} The coarse moduli space of $[A/G]$ for a strong crystallographic action
 is smooth only for $G=W(A_r)$ ($d=r+1$) and $G=G(p,1,r)$ ($p=2,3,4,6$) with $d=1$
 and $r=1,2,\dots$.
 The number of stacky (strong) $\bigstar$-geometries whose central fiber is smooth
 (in the classical sense) is written in the second row of table \eqref{uyyyy2}.    
 \end{fact}
 
 We stress that this statement holds also for the sporadic cases, that is,
 $A/G$ is non smooth for all exceptional Shephard-Todd groups.

 \section*{Acknowledgements}
 
 The author has benefit of fruitful discussions with B. Acharya, P. Argyres, M. Del Zotto, and M. Martone.
 
 \appendix

\section{Examples of absolutely incomplete geometries}\label{G!2}

Consider the three rank-$2$ complex reflection groups
\be
G=G_4,\ G_5,\ G_8\quad\text{with}\quad |G_4|=24,\quad |G_5|=72,\quad |G_{8}|=48.
\ee
The first two groups are defined over $\Z[e^{2\pi i/3}]$ while $G_8$ over $\Z[i]$.
There are several full lattices $\Lambda_G\subset\C^2$ which are invariant under these groups: for their list see \cite{jap,fuji}.
An invariant Hermitian form over the corresponding number field yields
a $G$-invariant polarization,
so $\C^2/\Lambda_G$ is a polarized Abelian surface whose automorphism group contains $G$.

\begin{claim} No choice of lattice $\Lambda_G$ produces a principally polarized Abelian
surface with $G$ a group of automorphisms.
\end{claim} 

\begin{proof} 
\textbf{Theorem 4.8} of \cite{fuji} lists the maximal groups of automorphisms of principally polarized
Abelian surfaces. In our case the Abelian surfaces invariant under $G_4$ or $G_5$ are isogeneous to $E_{e^{2\pi i/3}}\times E_{e^{2\pi i/3}}$
while the ones invariant under $G_8$ are isogeneous to $E_i\times E_i$.
In both cases the maximal automorphism groups are rank-2 imprimitive reflection groups,
respectively $G(6,1,2)$ and $G(4,1,2)$. In the list there is also the surface
$E_{\sqrt{-3}}\times E_{\sqrt{-3}}$ with automorphism group $G(6,3,2)\triangleleft G(6,1,2)$.
Since a rank-2 primitive reflection group cannot be a subgroup of an imprimitive one
of the same rank, 
we conclude that $\C^2/\Lambda_G$ does not admit a $G$-invariant \emph{principal} polarization. 
\end{proof}

\section{The principal $G_{12}$ geometry}\label{G!3}

In the lists of Abelian surfaces whose automorphism group is a complex reflection group
there is no surface for $G_{12}$ (of order $|G_{12}|=48$). However the surface in item 7 of table 9 
(and item 10 of table 12) of \cite{fuji} (or the 7th surface in \textbf{Theorem 13.4.5} of \cite{cristine})
has a maximal automorphism group generated by the four matrices
\be
\begin{aligned}
b_1&=\begin{pmatrix}-1+\sqrt{-2} & -2\\ -\sqrt{-2} & 1-\sqrt{-2}\end{pmatrix}
&\qquad b_2&=\begin{pmatrix}1 & \sqrt{-2}\\ \sqrt{-2} & -1\end{pmatrix}\\
b_3&=\begin{pmatrix}-1 & -1\\ 1 & 0\end{pmatrix}
&b_4&=\begin{pmatrix}1 & 1+\sqrt{-2}\\ 0 & -1\end{pmatrix}
\end{aligned}
\ee
The first observation is that the two elements $b_3$ and $b_4$ suffice to generated the group since
\be
b_1=-(b_3b_4)^2=(b_4b_3)^4(b_3b_4)^2\qquad b_2=(b_1b_4)^3b_4^{-1}.
\ee
Now redefine the generators as $(b_3,b_4)\leadsto (S,T)$
\be
S=b_4,\qquad T=-b_3\equiv b_3(b_4b_3)^4
\ee
(so that $b_3=T(ST)^4$). The two generators satisfy the relations
\be
S^2=1,\quad T^3=-1,\quad (ST)^4=-1,
\ee
which is the presentation of $G_{12}$ given in \cite{ST1}.
From \textbf{Theorem 4.8} of \cite{fuji} we get

\begin{fact} A principal polarized Abelian surface with automorphism group the Shephard-Todd
group $G_{12}$ is isomorphic as a torus to $E_{\sqrt{-2}}\times E_{\sqrt{-2}}$
and has a period matrix
\be
\tau=\begin{pmatrix}\frac{1}{2}+\sqrt{-2} & \frac{\sqrt{-2}}{2}\\
\frac{\sqrt{-2}}{2} & \frac{1}{2}+\sqrt{-2}\end{pmatrix}
\ee
\end{fact}
This Abelian variety is the Jacobian of the Bolza curve \cite{bolza,bolza2}
\be
y^2=x^5-x
\ee
which is the only genus-2 curve with automorphism group $G_{12}$.
In particular we have the following automorphisms
\begin{align}
&\text{order 2} &&(x,y)\to \left(-\frac{x+i}{1+ix},\frac{2\sqrt{2} y}{(1+i x)^3}\right)\\
&\text{order 6} &&(x,y)\to \left(i\,\frac{x-1}{x+1},\frac{8 y}{(1-i)^3(x+1)^3}\right)\\
&\text{order 8} &&(x,y)\to \left(i x, e^{i\pi/4}y\right).
\end{align}

\begin{small}

\end{small}

\end{document}